\documentclass{article}
\usepackage[utf8]{inputenc}

\usepackage{xcolor,amsmath,amsfonts,fullpage,amsthm,mathrsfs,fancyhdr,verbatim,framed,algorithmic,algorithm}
\usepackage[colorlinks=true, citecolor=magenta]{hyperref}
\usepackage{enumerate}
\usepackage[capitalize]{cleveref}
\usepackage{mathpazo}
\usepackage{bm}
\usepackage{cancel}

\newcommand{\eps}{\varepsilon}

\newcommand{\ket}[1]{{|#1\rangle}}
\newcommand{\bra}[1]{{\langle#1|}}

\DeclareMathOperator{\tr}{tr}
\newcommand{\Id}{\mathbf{1}}

\newcommand{\ot}{\otimes}

\newcommand{\calH}{\mathcal{H}}

\renewcommand{\cal}[1]{\mathcal{#1}}
\DeclareMathOperator*{\E}{\mathbb{E}}

\newcommand{\NP}{\mathsf{NP}}

\newcommand{\MIP}{\mathsf{MIP}}
\newcommand{\NEXP}{\mathsf{NEXP}}
\DeclareMathOperator{\poly}{poly}

\renewcommand{\sf}[1]{\mathsf{#1}}

\newcommand{\tlp}{\sf{teleport}}
\newcommand{\CHSH}{\mathsf{CHSH}}
\newcommand{\comm}{\mathsf{com}}

\newcommand{\crp}{\mathrm{crypto}}
\newcommand{\acom}{\mathrm{anticom}}
\newcommand{\com}{\mathrm{com}}
\newcommand{\phase}{\mathrm{phase}}

\newtheorem{theorem}{Theorem}
\newtheorem{definition}[theorem]{Definition}

\newtheorem{lemma}[theorem]{Lemma}

\newtheorem{claim}[theorem]{Claim}
\newtheorem{remark}[theorem]{Remark}
\newif\ifnotes\notesfalse
\ifnotes
\newcommand{\znote}[1]{\textcolor{blue}{(Tina: #1)}}
\newcommand{\anote}[1]{\textcolor{red}{(Anand: #1)}}
\else
\newcommand{\znote}[1]{}
\newcommand{\anote}[1]{}
\fi

\usepackage{mathrsfs}

\newcommand{\scrB}{\mathscr{B}}

\newcommand{\Gen}{\mathsf{Gen}}
\def\Eval{\mathsf{Eval}}

\newcommand{\secp}{\lambda}

\newcommand{\sk}{\mathsf{sk}}

\newcommand{\pST}{\; \middle| \;}

\newcommand{\Nat}{\mathbb{N}} %natural numbers

\newcommand{\C}{C} %circuit

\renewcommand{\v}{v}

\newcommand{\PPT}{\mathsf{PPT}}
\newcommand{\QPT}{\mathsf{QPT}}
\newcommand{\Enc}{\mathsf{Enc}}
\newcommand{\Dec}{\mathsf{Dec}}

\newcommand{\ct}{\mathsf{ct}}

\newcommand{\QHE}{\mathsf{QHE}}

\newcommand{\LQHE}{\mathsf{QHE}}

\newcommand{\RegA}{\mathcal{A}}
\newcommand{\RegB}{\mathcal{B}}

\def\cA{{\cal A}}
\def\cB{{\cal B}}
\def\cC{{\cal C}}

\def\cE{{\cal E}}

\def\cH{{\cal H}}

\def\cQ{{\cal Q}}

\def\cT{{\cal T}}

\newcommand{\tensor}{\otimes}

\newif\ifnicematrix
\nicematrixfalse
\ifnicematrix
\newcommand{\nicematrix}[1]{#1}

\else

\newcommand{\nicematrix}[1]{}

\fi

\ifnicematrix
\usepackage{nicematrix}
\NiceMatrixOptions{
code-for-first-row = \color{blue} ,
code-for-last-row = \color{blue} ,
code-for-first-col = \color{blue} ,
code-for-last-col = \color{blue}
}
\fi

\begin{document}
%\setmainfont[Ligatures = Rare, Style=Historic, StylisticSet =4]{CMU Serif}%

%\title{Disquisitions on the Quantum Value of Cryptographically Compiled Nonlocal Games\\
%\small{\textit{Or, a simple and facile Method, whereby the proper Functioning of a quantum computational Engine may be establish'd}}}

\title{Bounding the Quantum Value of Compiled Nonlocal Games: From
  CHSH to BQP Verification}

\newif\ifnames\namestrue
%\newif\ifnames\namesfalse
\ifnames
\author{Anand Natarajan\thanks{\texttt{anandn@mit.edu}}\\\small{MIT}
  \and Tina Zhang\thanks{\texttt{tinaz@mit.edu}}\\\small{MIT}
}
\else
\fi

\maketitle

\begin{abstract}
  In the classical world, an extremely fruitful technique for
  constructing interactive protocols is ``compiling'' a multiprover
  game, using cryptography to simulate the separation between the
  provers. In the quantum world, the study of compiled nonlocal games
  was introduced by Kalai et al. (STOC'23), who defined a compilation
  procedure that applies to any nonlocal game and preserves the
  \emph{classical} value; however, they did not show any bounds on the
  \emph{quantum} value of their protocols. In this work, we make
  progress towards a full understanding of the quantum value of
  compiled nonlocal games. For the special case of the CHSH game, we
  show that the Tsirelson bound holds for the compiled game in two
  ways: by extending the ``macroscopic locality'' argument of
  Rohrlich, and by showing that strategies for the compiled game yield
  feasible solutions to the Tsirelson SDP. We conjecture that the
  latter argument can be extended to all XOR games. Using our SDP
  argument, we are able to recover a strong version of the
  ``rigidity'' property that makes CHSH so useful in applications;
  specifically, we show that compiled CHSH is a ``computational
  self-test'' in the sense of Metger and Vidick. As an application, we
  give a classical verification protocol for BQP based on a compiled
  nonlocal game and prove soundness. Our protocol replicates the
  functionality of Mahadev '18 but with two advantages: (1) the
  soundness analysis is much simpler, and directly follows the
  analysis of the nonlocal case, and (2) the soundness does not
  ``explicitly'' use the assumption of a TCF or an adaptive hardcore
  bit, and only requires QFHE as a black box (though currently the
  only known constructions of QFHE use TCFs).
\end{abstract}

\section{Introduction}

% \begin{itemize}
% \item nsMIP protocols can be compiled into single-prover delegation for P
% \item Can $MIP^*$ protocols be compiled to single prover quantum
%   protocols? KLVY give a partial answer to this
% \item We completed this by showing the quantum soundness of the KLVY
%   protocol for the CHSH game.
% \item From this we are able to derive a protocol for classical
%   verification of BQP computations, whose security relis only on
%   quantum fully homormophic encryption.
% \end{itemize}
%\anote{CHSH is the harmonic oscillator of nonlocality}
%\anote{This paper is dedicated to Tony Metger.}

The study of \emph{multiprover interactive proofs} (MIPs) is indispensable to complexity theory and cryptography. In complexity theory, the study of the power of the MIP model of computation, in which one computationally bounded verifier interacts with two (sometimes more) untrusted and unbounded provers who are not allowed to communicate, has led to many of the most celebrated results and fundamental techniques in the field. For example, the 1991 work of Babai, Fortnow and Lund showed that $\MIP = \NEXP$ \cite{BFL91}, and the techniques used in the proof were adapted to prove several other important results, including $\sf{PCP} = \NP$ \cite{arora1998probabilistic,arora1998proof}. Meanwhile, the study of variants on the classic two-prover MIP model has yielded a number of equally interesting lines of research. For example, we now know that when the two unbounded provers are allowed to share entanglement, the deciding power of the model increases to $\sf{RE}$ \cite{JNVWY20}; if the verifier's questions are forced to be uncorrelated, then the deciding power of the model decreases to $\sf{AM}$ \cite{AIM}; and if the two unbounded provers are allowed to share not only entanglement but any \emph{non-signalling} correlations, then the deciding power of the model decreases to $\sf{PSPACE}$ \cite{ito2010polynomial} (although, if we allow the verifier to interact with polynomially many provers instead of 2, the deciding power of the model goes back up to $\sf{EXP}$ \cite{kalai2014delegate}).

The study of MIPs is also, somewhat more obliquely, important to cryptography, because techniques and ideas that originate in the study of MIPs often find application in cryptographic settings. In cryptographic settings, we are commonly interested in situations where all the parties are efficient and thus can be restrained using cryptographic tools; we are also more commonly interested in an interaction between two parties (e.g. between a verifier and a \emph{single} prover) than an interaction between multiple parties in which a no-communication assumption between two or more of them is meaningful. Nonetheless, ideas from MIPs sometimes find a surprising amount of traction in cryptography. As early as 2000, for example, the idea of creating succinct arguments for NP by \emph{compiling} MIPs (specifically, PCPs) using cryptographic tools was proposed by Aiello et al. \cite{aiello2000fast}, although questions were subsequently raised about the soundness of such a compilation procedure by Dwork et al. \cite{dwork2004succinct}. These works were followed up by Kalai, Raz and Rothblum in 2013~\cite{kalai2014delegate}, who observed that any MIP sound against \emph{nonsignalling} provers can in fact be `compiled', using a homomorphic encryption (HE) scheme or a private  information retrieval (PIR) scheme, into a single-prover protocol in which the single prover is efficient and controlled by cryptography. `Compilation' here refers to a black-box procedure which takes any non-signalling MIP and turns it into a single-prover, cryptographically sound protocol by simulating the separation between the provers using cryptography. Kalai, Raz and Rothblum then showed that the non-signalling MIP (nsMIP) model has the same deciding power as $\sf{EXP}$, from which it follows that there is a single-prover cryptographic delegation protocol for all of $\mathsf{EXP}$, assuming that there is subexponentially secure homomorphic encryption which cannot be broken by the powerful prover.

We may ask if the same principle could be applied fruitfully in the quantum setting. The quantum version of the MIP model, known as the MIP$^*$ model, has been extensively studied, starting with the historical observation by Bell~\cite{bell1964einstein} that entangled players can win certain `nonlocal games' (games involving two or more players and a referee, in which the players are only allowed to communicate with the referee) with higher probability than classical players. Bell's observation led to a long line of research, and now the MIP$^*$ model---in which the two noncommunicating provers of the MIP model are allowed to share quantum entanglement---is one of the best understood models in quantum complexity theory. \cite{CHTW04,IV12,reichardt2013classical,FNT14,Ji17} As a consequence of this line of research, we have a rich repertoire of techniques for proving the soundness of multiprover entangled interactive proofs, not all of which have easy cryptographic analogues in the comparatively young area of single-prover quantum delegation (in which a \emph{classical} verifier interacts with an untrusted quantum polynomial-time prover and uses cryptography in order to achieve certain ends, such as randomness generation \cite{brakerski2021cryptographic} or the verification of BQP instances \cite{mahadev2018classical}). It is natural, then, to ask whether there is some way to translate techniques from the former model into the latter model, the same way that \cite{kalai2014delegate} translates between nsMIP and single-prover classical delegation using cryptography.

Suppose that we are allowed to use (quantum) homomorphic encryption, \`a la \cite{mahadev2020classical}, \cite{brakerski2018quantum}; then, a natural first attempt at creating such a `translation' procedure might be as follows. Given a two-prover nonlocal game between two players Alice and Bob and a referee, the verifier transforms it into a single-prover protocol by \emph{encrypting} Alice's and Bob's questions under different encryption keys. The verifier then sends both encrypted questions to its single computationally bounded prover, and the single prover is expected to homomorphically compute two answers, one `Alice' answer and one `Bob' answer, using the same strategies that nonlocal Alice and Bob would have used. (The prover can do this because the very purpose of homomorphic encryption is to allow computations on encrypted data.) This homomorphic computation results in two ciphertexts, one of which encrypts the `Alice answer' and one of which encrypts the `Bob answer'. The prover sends both ciphertexts to the verifier; the verifier decrypts the ciphertexts and decides whether to accept or reject in the same way that the referee of the nonlocal game would have.

As it turns out, this perhaps natural first attempt fails in an interesting way. It was suspected since 2004 \cite{dwork2004succinct}, and confirmed in 2016 \cite{dodis2016spooky}, that the compilation procedure just described does \emph{not} necessarily preserve the value (maximum winning probability), either classical or quantum, of the original nonlocal game. Dodis et al. show that there are certain homomorphic encryption schemes for which this compilation procedure only preserves the non-signalling value, which corresponds to the maximum winning probability that can be achieved by a class of strategies even more general than quantum strategies. This result complements the result of Kalai, Raz and Rothblum, which shows precisely that the same compilation procedure prevents the single prover from simulating any two-prover strategy that involves signalling.

Nonetheless, the techniques we have for controlling entangled nonlocal players are so useful that efforts have persisted to translate these, more or less generally, into the cryptographic single-prover setting. For example, the single-prover randomness generation protocol proposed by Brakerski, Christiano, Mahadev, Vazirani and Vidick \cite{brakerski2021cryptographic} relies on a specific cryptograhic version of a `self-test', a powerful type of MIP$^*$ protocol in which the verifier forces the provers to execute certain quantum operations---in spite of being only classical itself---by using a property of certain entangled nonlocal games known as \emph{rigidity}. Intuitively, rigidity guarantees that, if entangled and noncommunicating players pass with high probability in a certain nonlocal game, there is essentially a \emph{unique} quantum strategy that they must be using (characterised by the algebraic relations between the measurement operators that each prover applies). This property allows a classical verifier to control entangled and noncommunicating provers only by testing their classical measurement statistics, and is fundamental to the considerable power of the MIP$^*$ model.

The \cite{brakerski2021cryptographic} protocol uses a cryptographic version of a `self-test' in order to extract randomness from a quantum prover. It compels the prover to generate randomness by essentially forcing it to prepare an eigenstate of one measurement basis (e.g. a $\ket{+}$ state) and measure that eigenstate in the complementary basis (e.g. the standard basis). This is done by exploiting special properties of so-called \emph{noisy trapdoor claw-free functions} (NTCFs), a classical cryptographic primitive with tailor-made characteristics that facilitate precisely the kind of `computational self-test' just described. The same cryptographic tools which make up this `computational $\ket{+}$ state self-test' appear again in the celebrated work of Mahadev \cite{mahadev2018classical}, which allows a classical verifier to verify BQP instances by interacting with a cryptographically bounded quantum prover. Mahadev's protocol works by forcing the quantum prover to do certain measurements in complementary bases on its internal state, which can then be interpreted as measurements of certain local Hamiltonian terms.

Other works have followed in this line. For example, \cite{metger2021self} extends the `computational $\ket{+}$ state self-test' to a more general set of measurements on an EPR pair, also using NTCFs. (\cite{metger2021self} also inspired followups such as \cite{gheorghiu2022quantum}, \cite{fu2022computational}, \cite{mizutani2022computational}, which progressively expanded the set of states and operators that could be cryptographically self-tested using similar techniques.) \cite{kahanamoku2022classically} presents a much simplified version of the `computational $\ket{+}$ state self-test' of \cite{brakerski2021cryptographic} that is still a proof of quantumness (meaning that an efficient classical prover cannot pass with high probability), which allows more efficient proofs of quantumness from assumptions other than Learning With Errors, the only cryptographic assumption from which we are currently able to derive the full range of properties that (NTCFs) can have. All of these protocols, however, are bespoke protocols tailored for particular applications, and all of them rely heavily on the specific structure of NTCFs. One might ask whether there is a more general and black-box way to translate useful nonlocal techniques into the cryptographic setting.

A candidate for such a transformation was proposed by Kalai, Lombardi, Vaikuntanathan and Yang in 2022 \cite{KLVY21}. They propose a compilation procedure, along the lines of the compilation procedure for turning nsMIP protocols into single-prover delegation protocols, which can be applied to \emph{any} MIP$^*$ protocol, and which produces a single-prover cryptographic protocol that we might hope preserves the quantum value of the original MIP$^*$ protocol. The compilation procedure itself is very simple, and can be instantiated using any quantum homomorphic encryption scheme (or, more broadly, any blind quantum delegation scheme). Recall that a general two-player one-round MIP$^*$ protocol runs as follows:
\begin{enumerate}
\item Before the interaction begins, the two honest provers (`Alice' and `Bob') prepare a shared state $\ket{\psi}$---usually some number of EPR pairs---and divide it up between them, so that Alice keeps some portion of the qubits and Bob keeps the rest.
\item The verifier generates a question pair $(x,y)$ from some specified set of questions, and sends the question $x$ to the first prover (`Alice') and the question $y$ to the second prover (`Bob').
\item Alice replies to the verifier with an answer $a$, and Bob replies to the verifier with an answer $b$. In the honest case, these answers are generated through measurements of the shared state $\ket{\psi}$. Since the provers cannot communicate, Alice measures only her qubits in order to produce $a$, and Bob likewise measures only his qubits in order to produce $b$.
\item The verifier computes a decision predicate $V(x,y,a,b)$, and accepts iff $V(x,y,a,b) = 1$.
\end{enumerate}

The compiled version of this protocol, given a quantum homomorphic encryption scheme $\QHE = (\Gen, \Enc, \Dec, \Eval)$, is as follows:
\begin{enumerate}
\item Before the interaction begins, the honest prover prepares the same state $\ket{\psi}$ that the nonlocal provers would have prepared, and divides it up into `Alice's qubits' and `Bob's qubits'.
\item The verifier generates a secret key $\sk \gets \Gen(1^\lambda)$, along with a question pair $(x,y)$. The verifier sends $c := \Enc_\sk(x)$ to the prover.
\item The prover \emph{homomorphically} evaluates the quantum measurement which Alice would have evaluated in the nonlocal game, using the homomorphic capabilities of the encryption scheme. This is necessary because the prover does not know $x$, and only has an \emph{encryption} of $x$; however, the very purpose of a homomorphic encryption scheme is to allow computations on encrypted data. The result of the homomorphic computation is an encryption $\alpha := \Enc_\sk(a)$, and the prover sends $\alpha$ back to the verifier.
\item The verifier sends $y$ to the prover \emph{in the clear}.
\item The prover evaluates the quantum circuit which Bob would have evaluated in the nonlocal game on `Bob's qubits', again in the clear. Note that the homomorphic part of the computation (step 3) occurred only on `Alice's qubits', and the `Bob' measurement occurring in this step is therefore happening on a disjoint set of qubits from the `Alice' measurement of step 3. A homomorphic encryption scheme with sufficiently strong correctness properties will ensure that the measurement outcome $b$ produced in this step has the same joint statistics with $a = \Dec_\sk(\alpha)$ as Bob's answer would have had with Alice's answer in the nonlocal case.
\item The verifier decrypts $\alpha$ to get $a$, computes $V(x,y,a,b)$, and accepts iff $V(x,y,a,b) = 1$.
\end{enumerate}

Intuitively, this protocol is using the encryption scheme to hide Alice's question $x$ from the prover, so that the prover cannot take advantage of knowing Alice's question when it is doing the `Bob' part of its computation. Moreover, it is using the \emph{round structure} of the protocol to ensure that the prover does not know the `Bob' question when it is doing the `Alice' part of its computation. In this way, the compilation procedure uses cryptography to simulate the no-communication assumption in the original protocol. Note that this compilation procedure differs from the `na\"ive' procedure first proposed by Aiello et al. \cite{aiello2000fast} (which we sketched earlier in this introduction) because the prover is forced by the protocol's round structure to provide its answer to the encrypted `Alice' question \emph{before it is given the Bob question}; this effectively ensures that the prover cannot craft its answers to the two questions simultaneously, which is what the results of \cite{dodis2016spooky} relied on to show that the prover could simulate non-signalling strategies.

Kalai, Lombardi, Vaikuntanathan and Yang prove that this compilation procedure preserves \emph{classical} value. That is, they prove that the maximum probability with which a pair of non-entangled, noncommunicating provers can pass in the original nonlocal protocol is also the maximum probability with which a single classical prover can pass in the compiled protocol, assuming the QFHE scheme is IND-CPA secure. They also prove that the quantum value of the compiled game is \emph{at least} that of the nonlocal game. This is already enough to produce single-prover `proofs of quantumness' by applying the compilation procedure to classic nonlocal games such as the CHSH game \cite{clauser1969proposed}, which is known to have classical value $\frac{3}{4}$ and quantum value $\cos^2(\pi/8) \approx 0.85$. However, the authors of \cite{KLVY21} leave open the question of whether their transformation preserves important \emph{quantum} properties of the nonlocal game to which it is applied, such as upper bounds on the quantum value and rigidity. For example, it was left open whether the quantum value of the compiled version of the CHSH game is equal to 1, as it is in the non-signalling world.
%\anote{In particular, it was open whether the quantum value of the compiled CHSH game is equal to one (as it is in the non-signalling world).} 
\subsection{Our results}

\subsubsection{Building blocks}
\label{sec:intro-building-blocks}
We make progress towards a full understanding of the quantum consequences of the KLVY transformation. We prove the following core technical lemmas, with an eye toward using them in order to recover the results of \cite{brakerski2021cryptographic}, \cite{metger2021self}, \cite{mahadev2018classical}, and others.
\begin{enumerate}
\item The quantum value of the \emph{CHSH game} \cite{clauser1969proposed} is preserved under the KLVY transformation: that is, the maximum winning probability which a single prover can achieve in the compiled version of the CHSH game is $\cos^2(\pi/8)$ (up to negligible corrections in the security parameter of the encryption scheme). We show this in two ways:
\begin{enumerate}
\item Directly constructing an operator from the prover's `Bob measurements' in the compiled CHSH game which, conditioned on the prover winning with probability better than $\cos^2(\pi/8)$, allows the prover to guess the `Alice question' $x$ with better than $\frac{1}{2}$ probability, thus violating IND-CPA security. This argument is based on arguments for the nonlocal CHSH value which were introduced in \cite{rohrlich2014stronger}, and we present our version of it in \Cref{sec:macroscopic-locality}. This argument is somewhat specific to CHSH, and it is comparatively difficult to see how to generalise it.
\item Decomposing the game value $p_{win}$ of the compiled CHSH game in terms of certain expectation values of the prover's `Alice' and `Bob' measurements in the compiled protocol, and then rewriting the decomposition in the form $p_{win} = \omega^* - \text{terms}^2$, where $\text{terms}^2$ is manifestly non-negative and $\omega^* = \cos^2(\pi/8)$. This argument is based on a common argument for the nonlocal CHSH value (and the nonlocal quantum value of other games) in terms of \emph{sum-of-squares decompositions}, and is presented in \Cref{sec:chsh-sos}. We believe that some version of this argument may generalise to other games whose value is captured by the first level of the non-commutative sum-of-squares (ncSoS) hierarchy~\cite{navascues2008convergent,doherty2008quantum}, and possibly even to higher levels, but we were not able to generalise it due to a technical cryptographic obstruction. See `Open questions' for more discussion.
\end{enumerate}
\item The KLVY transformation preserves an important \emph{rigidity} property of the CHSH game. The property is that the (square of the) anticommutator $\{B^0, B^1\}$ is approximately zero, where $B^0$ is the measurement the prover applies (to the post-measurement state left behind by its `Alice measurement') in the second round when it receives the `Bob question' 0, and $B^1$ is the measurement it applies when it receives the `Bob question' 1. More precisely, we show that $\{B^0, B^1\}^2$ approximately annihilates (has as a zero-eigenvalue eigenvector) the post-measurement state left behind by the prover's `Alice measurement' in the first round. This technical condition is an extremely important property of the CHSH game, because it means that CHSH functions as a `self-test' for a pair of anticommuting operators (namely, $B^0$ and $B^1$), and so also as a self-test for a qubit, if we identify $B^0$ with the single-qubit $Z$ operator and $B^1$ with the single-qubit $X$ operator. Our argument for this property is presented in \Cref{sec:chsh-anticom-robust}.
% \item \anote{For general projective nonlocal games, if the quantum value of the nonlocal game is $<1$, then the quantum value of the compiled version under the KLVY transformation is also $<1$. This result comes from a more direct generalisation of the \cite{KLVY21} proof of classical soundness than the previous two results, and it is of interest mainly because it shows that the value of a game under the KLVY transformation does \emph{not} coincide with its value under any finite level of the ncSoS hierarchy, even though the previous two results are concerned with a game (CHSH) whose level-1 value coincides with its quantum value. This result is some evidence, therefore, that the KLVY transformation does preserve the quantum value in general.}
\end{enumerate}

\subsubsection{Applications}

The most important two of our core lemmas concern the properties of the CHSH game under compilation. While focusing on CHSH alone may seem like an worrisomely specialised approach, this approach is justified by the enormous range of applications that CHSH (and similar nonlocal games) have found in the construction of quantum protocols. One could describe CHSH as the `harmonic oscillator' of nonlocal games: it is the simplest example which captures the important properties of nonlocal games that make them fruitful objects of study, such as quantum advantage (a quantum value higher than the classical value) and rigidity. As a result, almost every nonlocal protocol in the literature uses CHSH---or a similar game such as Magic Square---as a building block.

The power which understanding CHSH affords us can be seen when we turn our focus towards applications. Compiled CHSH is automatically a self-test for a single qubit and the associated complementary measurement operators, and therefore can be used to recover the randomness generation results of \cite{brakerski2021cryptographic}.\footnote{A qubit self-test by itself does not necessarily yield a full randomness generation protocol but only a single-round randomness generation protocol; however, the work of \cite{merkulov2023computational}, to appear presently, makes it easy to turn a single-round randomness generation protocol into a many-round randomness generation protocol by modularising Sections 7 and 8 of \cite{brakerski2021cryptographic}.} In combination with the `commutation test' (described in more detail in \Cref{sec:commutation}), compiled CHSH can be extended to a constant-robustness self-test for $n$ EPR pairs using the `Pauli braiding' idea that is found in \cite{natarajan2016robust}. This tool then makes it easy to recover the results of \cite{metger2021self}, as well as those of follow-ups such as \cite{gheorghiu2022quantum}, \cite{fu2022computational}, \cite{mizutani2022computational} using nonlocal remote state preparation techniques.\footnote{We did not actually try to recover these results, but we believe it would be relatively straightforward to proceed given the work we did do on making a computational version of the `Pauli braiding test' of \cite{natarajan2016robust}. One of the main attractions of our approach is that the proofs seem to follow their nonlocal models fairly closely; as such, the analyses of protocols which recover the aforementioned results using our approach would likely be simpler than the originals.}

Our main application, which we worked out to demonstrate the use of our techniques, is to recover the title result of \cite{mahadev2018classical}---a classical verification protocol for BQP instances under cryptographic assumptions---using a conceptually different approach. The problem of BQP verification, namely, that of designing a proof system by which a fully classical verifier can decide instances in BQP through an unbounded and untrusted prover, and in which the \emph{honest} prover is quantum polynomial-time, is one of the important open problems in quantum complexity theory. In particular, it has been known for some years that BQP verification can be done if the `classical' verifier has the ability to do very limited (one-qubit) quantum operations \cite{fitzsimons2018post}, or if the classical verifier interacts with \emph{two} noncommunicating entangled provers instead of one \cite{reichardt2013classical}. Mahadev's celebrated work of 2018 showed that BQP verification can also be done if the untrusted prover is efficient and subject to post-quantum cryptographic assumptions. Her work led to a host of follow-up work which found various applications for her new cryptographic techniques \cite{gheorghiu2019computationally, alagic2020non, chia2020classical, Zhang22, bartusek2022succinct}. As far as we know, we are the first to recover Mahadev's result using a markedly different approach.

The idea behind our new verification protocol is simple: we compile a nonlocal BQP verification protocol using the KLVY transformation. The nonlocal protocol in question is not the same as the one of \cite{reichardt2013classical}, and is more similar, though not identical, to the one of \cite{grilo2017simple}. Intuitively, this nonlocal verification protocol uses a variant of the \emph{Pauli braiding test} of \cite{natarajan2016robust} in order to establish a correspondence between the operators that `Bob' applies and the Pauli operators on $n$ qubits. Once the correspondence is established, the verifier can then simply ask Bob to prepare a certain state and measure it using the Pauli operators, and interpret the measurement outcomes as measurements of Hamiltonian terms on a witness state. In order that the Pauli braiding test and the subtest in which Hamiltonian terms are measured are indistinguishable to Bob, so that he has to use the same operators in both tests, the verifier asks Alice to \emph{teleport} the witness state to Bob during the Hamiltonian test, using their shared entanglement, and report the teleportation corrections that arise; the verifier then asks Bob to do the same types of measurements regardless of whether the Pauli braiding test or the Hamiltonian test is being performed. Since Alice and Bob are noncommunicating, Bob cannot tell when the witness state was teleported to him and when he is being subjected to the Pauli braiding test. In order to recover an effective measurement of the witness state, the verifier then corrects Bob's reported measurement outcomes using Alice's reported teleportation corrections and interprets the result as an energy measurement.

The main technical ingredient in the soundness analysis of this verification protocol is our analysis of a compiled version of the Pauli braiding test, which is a versatile and robust self-test for $n$ EPR pairs that can support many applications. The intuition behind the Pauli braiding test is that it tests that the prover's measurements `look like' the $n$-qubit Pauli operators $\{\sigma_Z(a)\}_{a \in \{0,1\}^n}, \{\sigma_X(b)\}_{b \in \{0,1\}^n}$ (see \Cref{sec:quantum-info-prel} for a more formal definition of these) by certifying that the prover's operators satisfy, on average, the commutation and anticommutation relations that the $n$-qubit Pauli operators satisfy, in addition to a \emph{linearity} property which says that, for any $W \in \{X,Z\}$ and any $a, a' \in \{0,1\}^n$, $\sigma_W(a) \sigma_W(a') = \sigma_W(a \oplus a')$. In our case, because local Hamiltonians with pure-$X$ and pure-$Z$ terms only are already QMA-complete, meaning that we only ever need Bob to measure in one basis at a time, we can modify the test to get linearity `for free' by forcing the linearity relations to be satisfied by construction: we only ask Bob two possible questions (`measure all qubits in $Z$ basis', and `measure all qubits in $X$ basis'), and we can construct Pauli operators that satisfy the linearity relations from the two resulting measurement operators. (See \Cref{sec:verification-modelling} for more details.) The commutation and anticommutation relations are then certified by choosing two random Pauli operators $\sigma_Z(a), \sigma_X(b)$ for uniformly random $a,b \in \{0,1\}^n$, and having the verifier referee the compiled CHSH game involving these operators if they anticommute, or the compiled commutation game involving these operators if they anticommute. Therefore, once we are equipped with the right lemmas about CHSH rigidity (\Cref{lem:chsh-anticom}) and commutation game rigidity (\Cref{lem:commutation-game-rigidity}), the analysis is, if not straightforward, at least familiar. We think it remarkable that the KLVY transformation makes it easy to write down a computational version of the Pauli braiding test whose analysis follows fairly naturally from the nonlocal analysis: if such a building block had been available before, we think it might have conceptually simplified many quantum delegated computation protocols in the literature.

The full protocol is presented in \Cref{sec:protocol}, and its soundness analysis is presented in \Cref{sec:soundness}. We remark, for the interested reader, that the most interesting individual step in the soundness analysis (apart from the proof of the CHSH rigidty lemma, \Cref{lem:chsh-anticom}) is perhaps the proof of \Cref{lem:zxz-close-to-x}. Here, it becomes clear why it is important to prove that the \emph{squared} anticommutator is zero: the anticommutator itself might have been sufficient if it weren't for the teleportation corrections, but the latter force us to condition on certain `Alice' outcomes, and then it becomes vital that all of our `error terms' are non-negative, since the conditioning causes us to remove certain terms from a sum.

\subsubsection{Discussion}

Our approach to BQP verification has several natural advantages:
\begin{itemize}
\item It is conceptually more modular: the nonlocal protocol and the blind delegation protocol can be treated more or less separately. We do not rely on the specific properties of NTCFs (in fact, our analysis never even mentions them), which we consider a boon given that NTCFs can be cumbersome to work with, notwithstanding their considerable power. The more modular nature of our protocol may open the door for the development of verification protocols from alternative assumptions.
\item Our analysis, if we view the blind delegation protocol as a black box, is comparatively simple in contrast with the original analysis of Mahadev. It is also familiar given some degree of experience with the nonlocal techniques, which may be useful from a pedagogical point of view.
\item Previously, the two-prover BQP verification protocol of \cite{reichardt2013classical} and Mahadev's cryptographic verification protocol were viewed as separate instantiations of the objective in different computational models. Our protocol establishes a conceptual link between them.
\item Our protocol confirms the perhaps natural intuition that blind delegation and verification are closely related, and that blind delegation should imply verification without further assumptions. Previously, it was not clear why Mahadev's verification protocol had a markedly different analysis, and used new assumptions (e.g. the `adaptive hardcore bit' property of NTCFs) compared with her QFHE protocol of 2017 \cite{mahadev2020classical}.
\end{itemize}

Viewed more generally, we believe our work represents an important step in the general program of translating powerful nonlocal quantum techniques into the cryptographic setting. As we mentioned, this program has already received a good deal of attention, but so far progress has been made mostly through the application of `ad hoc' cryptographic techniques that rely heavily on the structure of NTCFs. We believe it is possible to unify most of these previous results under our framework, which is significantly less complicated than existing heuristics for `translating' nonlocal results to the cryptographic setting, and which also generalises more easily to new tasks. As such, we claim that our approach is in some sense the `right way' to do computational self-testing: we get the closest and most general analogy to the nonlocal setting using the simplest machinery. In particular, the `Pauli braiding test' is a powerful tool in the nonlocal world, capable of supporting nearly any delegated computation or remote state preparation application, and our cryptographic version of it may facilitate the simplification of earlier work.

\paragraph{Related work}
While this work was in preparation, we became aware of an independent work \cite{brakerski2023simple} by Brakerski, Gheorghiu, Kahanamoku-Meyer, Porat, and Vidick that achieves similar results to items 1 and 2 of \Cref{sec:intro-building-blocks}. Specifically, they recover a tight bound on the quantum value for a family of single-prover cryptographically sound protocols obtained from nonlocal games, including the KLVY compilation of CHSH and the protocol of \cite{kahanamoku2022classically}, but \emph{not} for the KLVY compilation of general games. They also show that any protocol in this family of protocols is rigid (in the same sense in which we define it in \Cref{sec:intro-building-blocks}), and is therefore a `qubit test' which can be used to recover the randomness generation results of \cite{brakerski2021cryptographic}.
\subsection{Open questions}

\begin{itemize}
    \item We build a BQP verification protocol assuming QFHE as a black box. Unfortunately, known constructions of QFHE \cite{mahadev2020classical,brakerski2018quantum} all rely on heavy cryptographic machinery, and in particular on noisy trapdoor claw-free functions, which are currently more or less the only way we can control a single untrusted quantum party through purely classical interaction. However, our results about verification still hold even if the QFHE is replaced by any (potentially interactive) blind (not necessarily verifiable) quantum delegation protocol, and they still hold if the QFHE is replaced by a form of blind delegation that can only handle the specific circuit which the \emph{Alice} part of the prover performs. In its current form, it seems unlikely that there would be a blind delegation scheme which could handle the Alice circuit associated with our verification protocol but not general circuits. However, we might ask: could there be a nonlocal verification protocol where the `Alice' circuit was significantly simpler than the `Bob' circuit or the work required to prepare the witness state? For example, is there a two-prover verification protocol in which Alice performs only controlled Pauli measurements (and there are no restrictions on Bob)? Then, given such a protocol, is there a blind quantum delegation scheme which can handle such measurements without relying on TCFs, and perhaps even relying on a different assumption from LWE (e.g. LPN or quantum resistant one-way functions)?
    \item We gave a specialised proof of the quantum soundness of the KLVY transformation for CHSH, using a degree-1 SoS certificate for the CHSH value, but we conjecture that it can be extended to handle all degree-1 SoS certificates on the quantum value of any game. The main challenge along the way appears to be the following.
    The canonical degree-1 SoS for CHSH has the convenient property that, in each squared term $p_i^\dagger p_i$, the polynomial $p_i$ contains a \emph{single} $A$ monomial: $p_i = A - (B + B + \dots + B) $. This means that, when we expand the square and apply the pseudo-expectation operator $\tilde{\E}[\cdot]$, we only have to work with states that are post-measurement states of the $A$ measurements. However, a general SoS will contain multiple $A$ and $B$ monomials, and expanding the square will yield states of the form 
    \[ (\alpha_1 A_1 + \alpha_2 A_2 + \dots + \alpha_k A_k) \ket{\psi}\]
    Such a state never occurs operationally in the protocol, and in fact there is no obvious way to efficiently prepare such a state using only Alice's measurements from the protocol. (The usual way to measure the sum of two operators involves doing phase estimation on the outcomes---but the outcomes here are encrypted, and decrypting them would require the secret key.) As such, the IND-CPA security of the cryptography doesn't directly help us analyse such a state; in particular, it doesn't help us argue that such a state is indistinguishable from other states of a similar form. Could we get around this obstacle with a more sophisticated cryptographic argument, or do we need additional assumptions? We note that, if an analogue of our result for CHSH could be shown for all degree-1 SoS certificates, it would imply that the compiled value is equal to the quantum value for all XOR games, and would likely yield rigidity for these games as well.
    \item Of course, a natural next question is to study higher-degree SoS certificates, and convert them into arguments for the compiled value. Can we write a hierarchy of SoS relaxations that captures the compiled game value? 
    \item Alternatively, could we find an example game where we can provably separate the quantum and compiled values? 
    \item Is there a meaningful notion of ``commuting-operator value" for games compiled under the KLVY transformation? Intuitively, we might expect the answer to be no, because a finite security parameter $\lambda$ should restrict the prover to a finite-dimensional Hilbert space---but how do we show this?
    \item The current prover running time of our verification protocol is some unspecified polynomial in $n$ (the number of qubits in the witness state). However, it seems plausible to us that our techniques---in particular, our compiled version of the `Pauli braiding test'---could be used to significantly simplify the linear-time verification protocol of Zhang~\cite{Zhang22}. In aiming for a linear-time verification protocol, using history states already amounts to a loss, because the reduction involves a polynomial blowup. Is it possible to use the KLVY framework to compile a more MBQC-based or gate-by-gate type of nonlocal verification protocol, in order to make computational verification more efficient?
%    \item Could our techniques help analyze other protocols that are computational self-tests? For example, could the SoS tactic be applied to the protocol of Kahanamoku-Meyer et al. \cite{} or that of Metger and Vidick \cite{}?
    \item Could we recover the functionality of remote state preparation (\`a la \cite{gheorghiu2019computationally} and others) under our framework? This seems possible if we can obtain rigidity for a tomographically complete set of measurements for Bob, which might be possible using ideas from the extended CHSH game of \cite{reichardt2013classical} and followups. Could this yield better verification protocols as discussed in the first bullet point?
\end{itemize}

%%% Local Variables:
%%% mode: latex
%%% TeX-master: "../protocol"
%%% End:

\section{Preliminaries}

\subsection{Notation for norms and expectations}
For a square matrix $A$, the matrix absolute value is defined by
\[ |A| = \sqrt{A^\dagger A}, \: |A|^2 = A^\dagger A. \]

The \emph{state-dependent norm} is defined as 
\[ \|A\|^2_\psi = \tr[A^\dagger A \psi]. \]
If $\psi$ is a pure state $\ket{\psi}\bra{\psi}$, then this is equal to
\[ \|A\|^2_\psi = \bra{\psi} A^\dagger A \ket{\psi} = \|A \ket{\psi}\|_2^2. \]
If $\psi$ is the maximally mixed state, note that $\| \cdot \|_\psi$ coincides with the normalized Frobenius norm. It is useful to write the Cauchy-Schwarz inequality for this norm:

\begin{equation}
|\langle A^{\dagger}  B \rangle_\psi| \leq \| A\|_\psi \cdot \| B \|_\psi. \label{eq:cs-state}
\end{equation}

In general we will refer to the expectation of an operator $A$ on a state $\ket{\psi}$ as $\bra{\psi}A\ket{\psi}$; however, when the state in question is clear from context, we may shorten this to $\langle A \rangle$.

The notation $a \approx_\delta b$, always used when $a$ and $b$ are both real numbers and $\delta \geq 0$, indicates that
\begin{equation}
a - \delta \leq b \leq a + \delta.
\end{equation}

\subsection{Quantum information}

\label{sec:shjw}
\label{sec:quantum-info-prel}

For a detailed overview of quantum computation preliminaries, we refer the reader to \cite{NC02}. We establish any somewhat nonstandard notation in this section.

We may specify a projective measurement by specifying a set of orthonormal projectors: for example, the standard basis measurement on $n$ qubits may be specified in this way as the set $\{\ket{x_1}\bra{x_1} \cdots \ket{x_n}\bra{x_n} \}_{x \in \{0,1\}^n}$.

We will use the shorthand $\sigma_Z$ for the Pauli-Z operator, and
$\sigma_X$ for the Pauli-X operator. We will also use the shorthand
$\sigma_Z(a)$, for $a \in \{0,1\}^n$, to indicate the $n$-qubit Pauli
operator that is defined as $\sigma_Z$ on the qubits where $a = 1$ and
identity on all other qubits. More precisely, $\sigma_Z(a) :=
\bigotimes_{i=1}^n (\sigma_{Z,i})^{a_i}$, where $\sigma_{Z,i}$ is the
$\sigma_Z$ operator on the $i$th qubit out of $n$. Similarly,
$\sigma_X(b) := \bigotimes_{i=1}^n (\sigma_{X,i})^{b_i}$.

% We will need the following version of Uhlmann's theorem in our zero-error argument.
% \begin{theorem}
%   \label{thm:not-uhlman}
%   Let $AB$ be a bipartite system, $\rho_B$ be a mixed state on the $B$
%   system, and $\ket{\Psi}_{AB}$ and $\ket{\Phi}_{AB}$ be two pure
%   bipartite states. If $\ket{\Psi}_{AB}$ and $\ket{\Phi}_{AB}$
%   are both purifications of $\rho_B$, then there is a unitary $U_A$
%   acting only on the $A$ system such that $U_A \ot I_B \ket{\Psi}_{AB} = \ket{\Phi}_{AB}$.
% \end{theorem}
% \begin{proof}
%   This follows from Theorem 5.1.1 of~\cite{wildebook}.
% \end{proof}

% \begin{fact}\label{fact:diagonal-channel}
% \anote{THIS FACT IS FALSE}
% Let $\Xi$ be a quantum channel on two registers, and fix an orthonormal basis $\{\ket{u_i}\}_i$ of the second register. Then if it holds that for all $i$,
% \[ \Xi( \ket{0}\bra{0} \otimes \ket{u_i}\bra{u_i}) = \rho_i \otimes \ket{u_i}\bra{u_i} \]
% for some states $\rho_i$, then for any $i\neq j$,
% \[ \Xi( \ket{0}\bra{0} \otimes \ket{u_i} \bra{u_j}) = 0.\]
% \end{fact}
% \begin{proof}
% This is easy to see by dilating the channel to a unitary map: let $U$ be a unitary such that 
% \[ \Xi(\rho) = \tr_{\mathsf{aux}}[ U (\ket{0}\bra{0}_{\mathsf{aux}} \otimes \rho ) U^\dagger]. \]
% Then it must hold that for any $i$,
% \[ U (\ket{0}  \ot \ket{0}\ot \ket{u_i}) = \ket{\psi_i} \ot \ket{u_i}, \]
% where $\ket{\psi_i}$ is a normalized state
% \end{proof}

\subsection{Nonlocal games}
\begin{definition}
    A nonlocal game $G$ is given by natural numbers $n_1, n_2, m_1, m_2$, a distribution $\cQ$ over pairs $(x,y) \in \{0,1\}^{n_1} \times \{0,1\}^{n_2}$, and a polynomial-time verification predicate $V( x, y, a, b) \in \{0,1\}$, where $a \in \{0,1\}^{m_1}$ and $b \in \{0,1\}^{m_2}$.
\end{definition}
\begin{definition}
   A \emph{quantum strategy} $\mathscr{S}$ for a nonlocal game $G$ consists of the following:
   \begin{itemize}
       \item A bipartite finite-dimensional state $\ket{\psi} \in \mathcal{H}_A \otimes \mathcal{H}_B$.
       \item For every $x \in \{0,1\}^{n_1}$, a projective measurement $\{A^x_a\}_{a}$ acting on $\mathcal{H}_A$ with outcomes $a \in \{0,1\}^{m_1}$ (the ``Alice measurements").
       \item For every $y \in \{0,1\}^{n_2}$, a projective measurement $\{B^y_b\}_{b}$ acting on $\mathcal{H}_B$ with outcomes $b \in \{0,1\}^{m_2}$ (the ``Bob measurements").
   \end{itemize}
   The \emph{value} or \emph{winning probability} of this strategy is given by
   \begin{align}
       \omega^*(G, \mathscr{S}) &= \E_{(x,y) \sim \cQ} \sum_{a,b} V(x,y,a,b) \cdot \bra{\psi} A^x_a \otimes B^y_b \ket{\psi}.
   \end{align}
\end{definition}
\begin{definition}
    The \emph{entangled value} of a game $G$ is defined as
    \begin{align}
        \omega^*(G) &= \sup_{\mathscr{S}} \omega^*(G, \mathscr{S}).
    \end{align}
\end{definition}
\subsection{Cryptography}
\label{sec:qhe}
\begin{definition}
  A procedure is \emph{quantum polynomial time} or \emph{QPT} in this section if it
  can be implemented by a logspace-uniform family of quantum circuits with size polynomial in 1) the number of qubits $n$ which they take as input, and 2) the security parameter $\lambda$.
\end{definition}

The following definitions are taken with some modifications from \cite{KLVY21}.
\begin{definition}[Quantum Homomorphic Encryption (QHE)]\label{def:QHE-aux}
A quantum homomorphic encryption scheme $\QHE=(\Gen,\Enc,\Eval,\Dec)$ for a class of quantum circuits $\cC$ is a tuple of algorithms with the following syntax:
\begin{itemize}
    \item $\Gen$ is a $\PPT$ algorithm that takes as input the security parameter $1^\secp$ and outputs a (classical) secret key $\sk$ of $\poly(\secp)$ bits;
    \item $\Enc$ is a $\PPT$ algorithm that takes as input a secret key $\sk$ and a classical input $x$, and outputs a ciphertext $\ct$;
    \item $\Eval$ is a $\QPT$ algorithm that takes as input a tuple $(\C,\ket{\Psi},\ct_{\mathrm{in}})$, where $C:\calH\times(\mathbb{C}^2)^{\tensor n}\rightarrow (\mathbb{C}^2)^{\tensor m}$ is a quantum circuit, $\ket{\Psi}\in\calH$ is a quantum state, and $\ct_{\mathrm{in}}$ is a ciphertext corresponding to an $n$-bit plaintext. 
    $\Eval$ computes a quantum circuit
    $\Eval_C(\ket{\Psi}\tensor \ket{0}^{\otimes \poly(\secp, n)},\ct_{\mathrm{in}})$ which outputs a ciphertext $\ct_{\mathrm{out}}$. If $C$ has classical output, we require that $\Eval_C$ also has classical output.

    \item $\Dec$ is a $\sf{QPT}$ algorithm that takes as input a secret key $\sk$ and ciphertext $\ct$, and outputs a state $\ket{\phi}$. Additionally, if $\ct$ is a classical ciphertext, the decryption algorithm outputs a classical string $y$.
\end{itemize}

We require the following two properties from $(\Gen,\Enc,\Eval,\Dec)$:
\begin{itemize}
    \item \textbf{Correctness with auxiliary input:} For every security parameter $\secp\in\Nat$, any quantum circuit $C:\cH_{\RegA} \times (\mathbb{C}^2)^{\tensor n} \to \{0,1\}^*$ (with classical output), any quantum state $\ket{\Psi}_{\RegA \RegB} \in\cH_{\RegA} \tensor \cH_{\RegB}$,  any message $x\in \{0,1\}^n$, any secret key $\sk \gets \Gen(1^\secp)$ and any ciphertext $\ct \gets \Enc(\sk,x)$, the following states have negligible trace distance:
    %\[
    %C:\cH\times\{0,1\}^n\rightarrow \cH\tensor (\mathbb C^2)^{\tensor n} \tensor (\mathbb C^2)^{\tensor m},
    %\]
    %  $$
    %y\approx_s y',
    %$$
    \begin{description}
        \item \textit{Game $1$.} Start with $(x, \ket{\Psi}_{\RegA \RegB})$. Evaluate $ C$ on $x$ and register $\RegA$, obtaining classical string $y$. Output $y$ and the contents of register $\RegB$.
        \item \textit{Game $2$.} Start with $\ct \gets \Enc(\sk, x)$ and $\ket{\Psi}_{\RegA \RegB}$. Compute  $\ct' \gets\Eval_C(\cdot \tensor \ket{0}^{\poly(\secp, n)} ,\ct)$ on register $\RegA$. Compute $y'= \Dec(\sk,\ct')$. Output $y'$ and the contents of register $\RegB$.
    \end{description}
    
In words, ``correctness with auxiliary input'' requires that if QHE evaluation is applied to a register $\RegA$ that is a part of a joint (entangled) state in $\cH_{\RegA}\tensor \cH_{\RegB}$, the entanglement between the QHE evaluated output and $\RegB$ is preserved.
    %\yael{General question:  Is it the case that for known schemes it also holds that $\ket{\Psi'}\equiv\ket{\Psi''}$?  Such a guarantee is needed if we want to do multi-hop, no?}\alex{Yes, this is covered in Brakerski18.}
    
    \item \textbf{IND-CPA security against quantum distinguishers:} For any two messages $x_0, x_1$ and any $\sf{QPT}$ adversary ~$\cA$:
    \[\left|\Pr\left[\cA^{\Enc_\sk(\cdot)}(\ct_0) = 1 \pST
        \begin{array}{l}
        \sk\gets\Gen(1^\secp)\\
        \ct_0\gets\Enc(\sk,x_0)\\
        \end{array}
        \right]
    -\Pr\left[\cA^{\Enc_\sk(\cdot)}(\ct_1) = 1 \pST
        \begin{array}{l}
        \sk\gets\Gen(1^\secp)\\
        \ct_1\gets\Enc(\sk,x_1)\\
        \end{array}
        \right]\right|
    \le \sf{negl}(\secp)\enspace.
    \]
\end{itemize}
\end{definition}

%\begin{remark}
%\emph{A \emph{quantum fully homomorphic encryption ($\QFHE$)} is a $\QHE$ for the class of all poly-size quantum circuits. While \cite{Mah18a,Bra18} construct $\QFHE$ (with security against quantum distinguishers),%We define $\QHE$ for general classes of circuits $\cC$ since for some classes (subclasses of all poly-size quantum circuits), there exist simpler or more efficient $\QHE$ schemes. We also note that security only needs to hold against classical adversaries. We discuss how both of 
%~weaker forms of $\QHE$ may yield more efficient quantum advantage protocols (see Section~\ref{sec:future-work} for discussion).}
%\end{remark}
%
%For the purposes of this paper, it suffices to know the following claim about the instantiability of \cref{def:QHE-aux}.
%
%\begin{claim} \label{claim:QFHE}
%The \cite{Mah18a,Bra18} QFHE schemes satisfy \cref{def:QHE-aux} with correctness holding for the class of all poly-size quantum circuits.
%\end{claim}
%\cref{claim:QFHE} can be verified by inspecting the constructions given in \cite{Mah18a,Bra18}. In \cref{app:locality}, we show mild generic conditions under which a QFHE scheme satisfies correctness with respect to auxiliary input, and sketch a proof of \cref{claim:QFHE} .

\subsubsection{Security of the cryptography}
In this section we present a number of special utility lemmas about the cryptography that we will need going forward.

\begin{definition}[Parallel repeated IND-CPA security]
\label{def:PR-IND-CPA}
We define the `$m(\lambda)$-parallel repeated IND-CPA game' for a secret key encryption scheme $\cE = (\Gen, \Enc, \Dec)$ as follows:
\begin{enumerate}
\item Fix a security parameter $\lambda$. The challenger generates a secret key $\sk \gets \Gen(1^\lambda)$.
\item The adversary makes polynomially many encryption queries to an encryption oracle $\Enc_\sk(\cdot)$.
\item The adversary produces two messages $x_0, x_1$, and sends these to the challenger.
\item The challenger chooses $b \gets \{0,1\}$ uniformly at random and sends back $m$ independent encryptions of $x_b$ as $c_1, \dots, c_m$.
\item The adversary outputs a guess $b'$ for $b$. It wins if $b' = b$.
\end{enumerate}
We say that $\cE$ is parallel-repeated IND-CPA secure against quantum distinguishers if no QPT adversary can win the above game with better than $\frac{1}{2} + \sf{negl}(\lambda)$ probability.
\end{definition}

\begin{lemma}
\label{lem:PR-IND-CPA}
$m(\lambda)$-parallel-repeated IND-CPA security for polynomial $m(\lambda)$ reduces to regular IND-CPA security.
\end{lemma}
\begin{proof}
This proof is routine and we only sketch it here: given an adversary for $m$-parallel-repeated IND-CPA security, choose uniformly at random where to `plant' the challenge of interest and do a hybrid argument. 
\end{proof}
%This lemma is needed in the following lemma.

\begin{lemma}\label{lem:crypto-ind}
Let $\lambda \in \mathbb{N}$ be a security parameter. There exists a negligible function $\eta(\lambda)$ such that the following holds. For any two efficiently (in QPT) sampleable distributions $D_1, D_2$ over plaintext Alice questions, for any efficiently preparable state $\ket{\psi}$, and for any two-outcome measurement $\{M, I-M\}$ that can be implemented by a circuit with size $\poly(\lambda)$ acting on $m = \poly(\lambda)$ copies of Alice's register, it holds that
\begin{equation}
\left| \E_{x \leftarrow D_1} \tr[M \rho_x] - \E_{x \leftarrow D_2} \tr[M \rho_x] \right| \leq \eta(\lambda),
\label{eq:crypto-binary}
\end{equation}
where
\begin{equation}
\label{eq:rho_x}
\rho_x := \E_{c_1, \dots, c_m = \Enc(x)} \sum_{\alpha_1, \dots, \alpha_m} (A^{c_1}_{\alpha_1}) \otimes \dots \otimes (A^{c_m}_{\alpha_m}) (\ket{\psi} \bra{\psi})^{\otimes m} (A^{c_1} _{\alpha_1})^\dagger \otimes \dots \otimes (A^{c_m}_{\alpha_m})^\dagger. 
\end{equation}
\end{lemma}

\begin{proof}
This follows from the $m$--parallel repeated IND-CPA security (\Cref{def:PR-IND-CPA}) of the QFHE scheme against quantum adversaries, which follows from its regular IND-CPA security by \Cref{lem:PR-IND-CPA}. Define a distinguisher $\mathcal{A}$ for the parallel-repeated IND-CPA game (\Cref{def:PR-IND-CPA}) as follows:
\begin{enumerate}
    \item $\mathcal{A}$ samples $x_1 \leftarrow D_1, x_2 \leftarrow D_2$. It submits $x_1, x_2$ as the messages it wishes to be challenged on.
    \item $\mathcal{C}$, the challenger, chooses $x_b \in \{x_1, x_2\}$ uniformly at random, computes $m$ encryptions $c_1, \dots, c_m \leftarrow \mathsf{Enc}(k, x_b)$, and sends these encryptions to $\mathcal{A}$.
    \item $\mathcal{A}$ prepares $\rho_{x_b}$ (ref. \Cref{eq:rho_x}), measures $M$ on $\rho_{x_b}$, and outputs the outcome.
\end{enumerate}
If the conclusion of the lemma does not hold, then $\mathcal{A}$ breaks IND-CPA security by our hypothesis that \Cref{eq:crypto-binary} is false.
\end{proof}

\begin{lemma}\label{lem:mean-to-ind}
Let $\lambda \in \mathbb{N}$ be a security parameter. There exists a negligible function $\eta'(\lambda)$ such that the following holds. For any two efficiently (in QPT) sampleable distributions $D_1, D_2$ over plaintext Alice questions, for any efficiently preparable state $\ket{\psi}$, and for any POVM measurement $\{M_\beta\}_\beta$ with outcomes in $[0,1]$ which can be implemented by a circuit with size $\poly(\lambda)$, it holds that
\begin{equation}
\left|  \E_{x \leftarrow D_1} \E_{c = \Enc(x)} \sum_{\alpha} \sum_\beta \beta \bra{\psi} (A^c_\alpha)^\dagger M_\beta (A^c_\alpha) \ket{\psi} -\E_{x \leftarrow D_2} \E_{c = \Enc(x)} \sum_{\alpha} \sum_\beta \beta \bra{\psi} (A^c_\alpha)^\dagger M_\beta (A^c_\alpha) \ket{\psi} \right| \leq \eta'(\lambda).
\end{equation}
\end{lemma}
\begin{proof}
Suppose that the conclusion of the lemma is false. This means that there exists a family of QPT sampleable distributions $D_1^{\lambda}, D_2^{\lambda}$ and QPT implementable POVMs $\{M^{\lambda}_\beta\}$ indexed by $\lambda$ such that the quantity 
\begin{align}
    \Delta(\lambda) &:=  \left|  \E_{x \leftarrow D_1} \E_{c = \Enc(x)} \sum_{\alpha} \sum_\beta \beta \bra{\psi} (A^c_\alpha)^\dagger M_\beta (A^c_\alpha) \ket{\psi} -\E_{x \leftarrow D_2} \E_{c = \Enc(x)} \sum_{\alpha} \sum_\beta \beta \bra{\psi} (A^c_\alpha)^\dagger M_\beta (A^c_\alpha) \ket{\psi} \right|
\end{align}
is a non-negligible function of $\lambda$ (here all of the objects on the RHS are functions of $\lambda$). Unpacking the definition of negligible, we have that there exists $c > 0$ such that for all $N > 0$, there exists $\lambda_N > N$ for which 
\begin{equation}
    \Delta(\lambda_N) > \frac{1}{\lambda_N^c}.
\end{equation}

We will construct a two-outcome measurement $\{Q, I-Q\}$ that will violate the conclusion of \Cref{lem:crypto-ind}. Define two distributions $B_1, B_2$, by
\begin{align}
    \Pr[\beta \leftarrow B_1] &= \E_{x \leftarrow D_1} \E_{c = \Enc(x)} \sum_{\alpha}  \bra{\psi} (A^c_\alpha)^\dagger M_\beta (A^c_\alpha) \ket{\psi}, \\
     \Pr[ \beta \leftarrow B_2 ] &= \E_{x \leftarrow D_2} \E_{c = \Enc(x)} \sum_{\alpha}  \bra{\psi} (A^c_\alpha)^\dagger M_\beta (A^c_\alpha) \ket{\psi}.
\end{align}
Further define $\mu_1:= \E_{\beta \leftarrow B_1}[\beta], \mu_2 := \E_{\beta \leftarrow B_2}[\beta]$. By definition, it holds that $|\mu_1 - \mu_2| = \Delta(\lambda)$. WLOG assume that $\mu_1 > \mu_2$. 

We now define the POVM $Q$. This acts on $m$ copies of $\ket{\psi}$ (where $m$ will be chosen below) as follows:
\begin{enumerate}
    \item Measure $M$ on each copy to obtain outcomes $\beta_1, \dots, \beta_m$. 
    \item Accept if $\sum_i \beta_i \geq m(\mu_1 + \mu_2)/2$.
\end{enumerate}
Symbolically, we have
\begin{equation}
    Q = \sum_{\beta_1, \dots, \beta_m} \mathbf{1}[\sum_i \beta_i \geq m(\mu_1 + \mu_2)/2] \cdot  M_{\beta_1} \otimes \dots \otimes M_{\beta_m}.
\end{equation}
To prove that this is a good distinguisher, we must show two things. First, we must show that the chance that it accepts samples from $B_1$ is high, and second, we must show that the chance it accepts samples from $B_2$ is low. For the first, by a Chernoff bound
\begin{align}
    \Pr_{\beta_1, \dots, \beta_m \leftarrow B_1}[\sum_i \beta_i < m(\mu_1 + \mu_2)/2] &= \Pr_{\beta_1, \dots, \beta_m \leftarrow B_1}[\sum_i \beta_1 < m\mu_1(1 - \underbrace{\frac{(\mu_1 - \mu_2)}{2\mu_1}}_{\delta_1})]\\
    &\leq \exp\left(-\frac{\delta_1^2  \cdot m \mu_1}{2}\right) \\
    &\leq \exp\left( -\frac{(\mu_1 - \mu_2)^2 \cdot m}{8 \mu_1} \right) \\
    &\leq \exp\left( -\frac{(\Delta)^2 \cdot m}{8 } \right)
\end{align}
Note that $\mu_1 > \mu_2 \geq 0$ so $\mu_1$ is strictly positive, so the expressions above are well defined.

For the second, there are now two cases. If $\mu_2 = 0$, then it follows that for samples $\beta_1, \dots, \beta_m$ drawn from $B_2$, we have that $\sum_i \beta_i = 0$ with certainty, and thus the probability that the distinguisher accepts is $0$. Otherwise, if $\mu_2 > 0$, we will apply concentration bounds. apply a Chernoff bound
\begin{align}
    \Pr_{\beta_1, \dots, \beta_m \leftarrow B_2} [\sum_i \beta_i \geq m(\mu_1 + \mu_2)/2]     &=  \Pr_{\beta_1, \dots, \beta_m \leftarrow B_2} [\sum_i \beta_i \geq m\mu_2 (1 + \underbrace{\frac{(\mu_1 - \mu_2)}{2\mu_2}}_{\delta_2})]. \\
\end{align}
If $\delta_2 > 1$, then by Markov's inequality we have
\begin{equation}
    \Pr_{\beta_1, \dots, \beta_m \leftarrow B_2} [\sum_i \beta_i \geq m(\mu_1 + \mu_2)/2] \leq \frac{1}{2}.
\end{equation}
On the other hand, if $\delta_2 \leq 1$, we may apply a Chernoff bound:
\begin{align}
     \Pr_{\beta_1, \dots, \beta_m \leftarrow B_2} [\sum_i \beta_i \geq m(\mu_1 + \mu_2)/2] &\leq \exp\left( -\frac{\delta_2^2 \cdot m \mu_2}{3} \right) \\
     &\leq \exp\left(-\frac{(\mu_1 - \mu_2)^2 \cdot m}{12 \mu_2}\right) \\
     &\leq \exp\left(-\frac{(\Delta)^2 \cdot m}{12 }\right)
\end{align}

Let us now choose $m$ to separate the two acceptance probabilities. We will take $m = \lambda^{10c}$. For every $N > 0$, there exists $\lambda_N > N$ for which we have
\begin{align}
    \Delta &> \frac{1}{\lambda_N^c} \\
     \Pr_{\beta_1, \dots, \beta_m \leftarrow B_1}[\sum_i \beta_i < m(\mu_1 + \mu_2)/2] & \leq \exp\left(-\frac{\lambda_N^{8c}}{8}\right) \label{eq:B1-reject-prob-lambda} \\
    \Pr_{\beta_1, \dots, \beta_m \leftarrow B_2} [\sum_i \beta_i \geq m(\mu_1 + \mu_2)/2]  &\leq \max\left\{ \frac{1}{2}, \exp\left(-\frac{\lambda_N^{8c}}{12 }\right) \right\}. 
\end{align}
In particular, for all $N$ sufficiently large, we have that the right-hand side of \Cref{eq:B1-reject-prob-lambda} is at most $1/4$.

Hence, we obtain that for every $N$ sufficiently large, there exists a $\lambda_N$ such that the following bounds hold: 
\begin{align}
    \E_{x \leftarrow D_1} \tr[Q \rho_x] &= \Pr_{\beta_1, \dots, \beta_m \leftarrow B_1}[\sum_i \beta_i \geq m(\mu_1 + \mu_2)/2] \\
    &= 1 - \Pr_{\beta_1, \dots, \beta_m \leftarrow B_1}[\sum_i \beta_i < m(\mu_1 + \mu_2)/2] \\
    &\geq 3/4 \\
    \E_{x \leftarrow D_2} \tr[Q \rho_x] &\leq 1/2 \\
    \left| \E_{x \leftarrow D_1} \tr[Q \rho_x] - 
     \E_{x \leftarrow D_2} \tr[Q \rho_x] \right| &\geq \frac{1}{2}.
\end{align}
This violates the conclusion of \cref{lem:crypto-ind}, which states that the RHS of the last line is a negliglible function of $\lambda$. Hence, our starting assumption on $\Delta$ is false.

\end{proof}
\begin{definition}
Let $H$ be a Hermitian matrix. A block encoding of $H$ with scale factor $t$ is a unitary matrix $U$ such that $U$ in the computational basis has the form
\[ U = \begin{pmatrix} t H & * \\ * & * \end{pmatrix}. \]
We say a block encoding $U$ is \emph{implemented} by a circuit if $U$ is the unitary transformation performed by the circuit.
\end{definition}

\begin{lemma}
Let $B^0, B^1$ be two QPT-measurable binary observables. Then there is a QPT circuit implementing a block encoding of
\[ \mathscr{B}_{\pm} = \frac{B^0 \pm B^1}{\sqrt{2}} \]
with scale factor $O(1)$.
\label{lem:block-enc-bplusminus}
\end{lemma}
\begin{proof}
  Taking
\nicematrix{
\begin{align}
    V_{\pm} &= \ket{+}\bra{+} \otimes B_0 \pm \ket{-} \bra{-} \otimes B_1 \\
    &= \left(\frac{1}{2}\right) \cdot \begin{pNiceMatrix}[first-row, first-col] & 0 & 1 \\
    0 & B_0 \pm B_1 & B_0 \mp B_1 \\ 
    1 & B_0 \mp B_1 & B_0 \pm B_1
    \end{pNiceMatrix}
\end{align}
}
does the job: here $V_{\pm}$ requires one additional ancilla qubit,
and the blocks are labeled by the state of this ancilla. The $0,0$
block of $V_{\pm}$ contains the desired operator. It also holds that
$V_{\pm}$ is efficiently implementable by a quantum circuit. To see
this, recall that by assumption there exist circuits that measure
$B_0, B_1$; by a simple application of uncomputation, these yield
circuits to implement $B_0$ and $B_1$ as unitaries. Replacing each
gate in these circuits by the appropriate controlled gate and
concatenating them yields a circuit for $V_{\pm}$. 
\end{proof}

\begin{lemma}
  Let $B^0, B^1$ be two QPT-measurable binary
  observables. Then there is a QPT circuit implementing a block
  encoding of
  \[ \mathscr{B}_{\pm}^2 = \frac{(B^0 \pm B^1)^2}{2} \]
  with scale factor $O(1)$.
  \label{lem:block-enc-bplusminus-squared}
\end{lemma}
\begin{proof}
  We use a similar construction to the previous part but using an
  extra ancilla qubit. Define
  \nicematrix{
  \begin{align}
    V^{(1)}_{\pm} &= \ket{+}\bra{+} \ot I \ot B_0 \pm \ket{-}\bra{-}
                    \ot I \ot B_1 \\
                  &= \left(\frac{1}{2}\right) \cdot \begin{pNiceMatrix}[first-row,
      first-col] & 00 & 01 & 10 & 11 \\
      00 & B^0 \pm B^1  & 0 & B^0 \mp B^1& 0  \\
      01 & 0 & B^0 \pm B^1 & 0 & B^0 \mp B^1 \\
      10 & B^0   \mp B^1 & 0 & B^0 \pm B^1 & 0 \\
      11 & 0 & B^0   \mp B^1& 0 &  B^0   \pm B^1
    \end{pNiceMatrix} \\
    V^{(2)}_{\pm} &= I \ot \ket{+}\bra{+} \otimes B^0  \pm I \ot \ket{-}\bra{-}
              \otimes B^1  \\
    &= \left(\frac{1}{2}\right) \cdot \begin{pNiceMatrix}[first-row,
      first-col] & 00 & 01 & 10 & 11 \\
      00 & B^0 \pm B^1 & B^0 \mp B^1& 0 & 0 \\
      01 & B^0 \mp B^1 & B^0 \pm B^1 & 0 & 0 \\
      10 & 0 & 0 & B^0   \pm B^1 & B^0 \mp B^1 \\
      11 & 0 & 0 &  B^0   \mp B^1&   B^0   \pm B^1
    \end{pNiceMatrix} 
  \end{align}
  }
  Taking $V_{\pm}^{(2)} V_{\pm}^{(1)}$ gives us the desired block
  encoding, with $\mathscr{B}^2_{\pm}$ encoded in the $00,00$
  block. It straightforward to see that this is QPT-implementable
  given QPT implementations of $B^0$ and $B^1$.
\end{proof}
\begin{lemma}
Let $B^0, B^1$ be two QPT-measurable binary observables. Then there is a QPT circuit
implementing a block encoding of
\[ \mathscr{O}_{\pm} = (B^0 B^1 \pm B^1 B^0)^\dagger (B^0 B^1 \pm B^1
  B^0) = \pm (B^0 B^1 \pm B^1 B^0)^2\]
with scale factor $O(1)$.
\label{lem:block-enc-commutator}
\end{lemma}
\begin{proof}
  We will use two ancilla qubits. Define
  \nicematrix{
  \begin{align}
    W^{(1)}_{\pm} &= \ket{+}\bra{+} \ot I \otimes B^0 B^1 \pm
                    \ket{-}\bra{-} \ot I
              \otimes B^1 B^0 \\
    &= \left(\frac{1}{2}\right) \cdot \begin{pNiceMatrix}[first-row,
      first-col] & 00 & 01 & 10 & 11 \\
      00 & B^0B^1 \pm B^1B^0  & 0 & B^0B^1 \mp B^1B^0& 0  \\
      01 & 0 & B^0B^1 \pm B^1B^0 & 0 & B^0B^1 \mp B^1B^0 \\
      10 & B^0 B^1  \mp B^1B^0 & 0 & B^0B^1 \pm B^1B^0 & 0 \\
      11 & 0 & B^0 B^1  \mp B^1B^0& 0 &  B^0 B^1  \pm B^1B^0
    \end{pNiceMatrix} \\
    W^{(2)}_{\pm} &= I \ot \ket{+}\bra{+} \otimes B^0 B^1 \pm I \ot \ket{-}\bra{-}
              \otimes B^1 B^0 \\
    &= \left(\frac{1}{2}\right) \cdot \begin{pNiceMatrix}[first-row,
      first-col] & 00 & 01 & 10 & 11 \\
      00 & B^0B^1 \pm B^1B^0 & B^0B^1 \mp B^1B^0& 0 & 0 \\
      01 & B^0B^1 \mp B^1B^0 & B^0B^1 \pm B^1B^0 & 0 & 0 \\
      10 & 0 & 0 & B^0 B^1  \pm B^1B^0&  B^0B^1 \mp B^1B^0  \\
      11 & 0 & 0 &  B^0 B^1  \mp B^1B^0&   B^0 B^1  \pm B^1B^0
    \end{pNiceMatrix} 
  \end{align}
  }
  Since $B^0$ and $B^1$ are QPT-measurable binary observables, by an
  application of uncomputation they are QPT-implementable as unitaries
  as well. Thus, the unitaries $W^{(1)}_\pm$ and $W^{(2)}_\pm$ are
  both QPT-implementable.
  Taking the product $(W^{(2)}_{\pm})^\dagger (W^{(1)})$ yields the
  desired block encoding, with $\mathscr{O}_{\pm}$ encoded in the
  $00,00$ block. 
\end{proof}
\begin{lemma}
Suppose we have a QPT-implementable block encoding for a (not
necessarily binary) observable $\mathscr{B}$ with $O(1)$ blowup, and
suppose that $\| \mathscr{B} \| \leq O(1)$. Then $\mathscr{B}$ is
QPT-measurable up to precision $\eps$ for any $\eps^{-1} =
\poly(\lambda)$. That is, there exists a QPT-measurable POVM
$\{M_\beta\}_{\beta}$ such that for any state $\rho$, 
\begin{equation}
    \Big| \sum_\beta \beta \cdot \tr[ M_\beta \rho] - \tr[ \mathscr{B} \rho] \Big| \leq \eps.
\end{equation}
\label{lem:estimate-blockenc}
\end{lemma}
\begin{proof}
At a high level, we will view the Hermitian operator $\mathscr{B}$ as
a Hamiltonian and use the energy estimator of \cite{rall21} to measure it on $\rho$.

To start, let us renormalize and shift $\mathscr{B}$ by multiples of
identity so that the resulting operators are PSD and have eigenvalues
contained in a smaller subinterval of $[0,1]$. Let $r = \|
\mathscr{B}\|$ and define
\begin{align}
    H := \frac{\mathscr{B} + 4r I}{6r}.
\end{align}
It holds that
\begin{equation}
    0 \prec \frac{3r}{6r} I \preceq H \preceq \frac{5r}{6r} I \prec I.
\end{equation}
The reason that we require the stronger bounds above, rather than
merely $0 \preceq H \preceq I$ is that the energy estimation
algorithm has a chance of \emph{overflow} or \emph{underflow} error
for eigenvalues very close to $0$ or $1$.  

Next, we need to prepare a \emph{block encoding} of $H$, that
is, a unitary $U$ such that the upper-left block of this unitary
is proportional to $H$. We will accomplish this in two stages,
using the linear combination of unitaries construction of
\cite{CW12}. First, we will let $V$ be the block encoding of
$\mathscr{B}$ given by the hypothesis of the theorem, and let $t$ be
its scale factor (so the top left block of $V$ is equal to $t\mathscr{B}$). Next, we
will use this to implement the shift by a multiple of identity. Adjoin
another ancilla and write 
\nicematrix{
\begin{align}
R &:= \begin{pmatrix} \sqrt{\frac{4rt}{4rt + 1}} & -\frac{1}{\sqrt{4rt
    + 1}} \\
\frac{1}{\sqrt{4rt + 1}} & \sqrt{\frac{4rt}{4rt + 1}} \end{pmatrix} \\
U &= R^\dagger ( \ket{0}\bra{0} \otimes I + \ket{1} \bra{1} \otimes V) R \\
&= \left(\frac{1}{4rt + 1}\right) \cdot \begin{pNiceMatrix}[first-row, first-col] & 0 & 1 \\
    0 & 4rt I + V & \dots\\ 
    1 & \dots & \dots
    \end{pNiceMatrix}\\
&=  \left(\frac{t}{4rt + 1}\right)\begin{pNiceMatrix}[first-row, first-col] & 0 & 1 \\
    0 & 4r I + \mathscr{B} & \dots \\ 
    1 & \dots & \dots
    \end{pNiceMatrix}
    % &= \left(\frac{\sqrt{2}}{10}\right)\begin{pNiceMatrix}[first-row, first-col] & 0 & 1 \\
    % 0 & 4\sqrt{2}I + \mathscr{B}_{\pm} & -2\sqrt{2} I + 2 \mathscr{B}_{\pm} \\ 
    % 1 & -2\sqrt{2} I + 2 \mathscr{B}_{\pm} & \sqrt{2}I + 4\mathscr{B}_{\pm}
    % \end{pNiceMatrix}\\
\end{align}
}
Now we see that the $00$ block of $U$ is proportional to
$H$. Moreover, $U$ is efficiently implementable using the circuit for
$V$.

Now, equipped with the block encoding, we are ready to analyze the performance of the energy estimation algorithm. To set notation, let the dimension of the space on which $H$ acts be $r$. We now apply Corollary 16 of \cite{rall21}. This states that, for parameters $n, \alpha, \delta$ to be chosen below, there is an algorithm making 
\[ Q(n, \alpha, \delta) = O(\alpha^{-1} \log(\delta^{-1}) (2^n + \log(\alpha^{-1})))\]
queries to $U$ that implements a channel $\Lambda$ on two registers $\mathsf{Out}, \mathsf{In}$, where the first has dimension $2^n$ and the second has dimension $r$. This channel satisfies the property that $\| \Lambda - \Lambda_{\mathrm{ideal}}\|_{\diamond} \leq \delta$, where $\Lambda_{\mathrm{ideal}}$ is some channel such that for any eigenstate $\ket{\psi_j}$ of  $H$ with eigenvalue $E_j$,
\begin{equation} 
\Lambda_{\mathrm{ideal}}(\ket{0}\bra{0}_{\mathsf{Out}} \otimes \ket{\psi_j}\bra{\psi_j}_{\mathsf{In}}) = \underbrace{\left( p_j \ket{\lfloor 2^n E_j \rfloor}\bra{\lfloor 2^n E_j \rfloor} + (1-p_j)\ket{E'_j}\bra{E'_j} \right)_{\mathsf{Out}}}_{\sigma_j} \otimes \ket{\psi_j}\bra{\psi_j}_{\mathsf{In}}, \label{eq:ideal-energy-measurement}
\end{equation}
where $p_j \in [0,1]$ and $E' = \lfloor 2^n E_j \rfloor - 1 \mod 2^n$. In words, this says that with probability $1-\delta$, the algorithm outputs an estimate for the eigenvalue of $H$ that is correct up to precision $1/2^n$---with the possiblity of ``underflow error" for eigenvalues that are in the range $[0, 1/2^n)$.

This characterizes the action of the energy estimation algorithm on an eigenstate, but we would like to know how it acts on a \emph{general} state. Given a general state $\rho$, write it in the basis given by eigenstates of $\mathscr{B}$ as a sum of a diagonal component $\rho_{\mathrm{diag}} = \sum_j \lambda_j \ket{\psi_j} \bra{\psi_j}$ and an off-diagonal component $\rho_{\mathrm{off}}$.  Then we have
\begin{align}
    \Lambda_{\mathrm{ideal}}(\ket{0}\bra{0}_{\mathsf{Out}} \otimes \rho) &=  \sum_j \lambda_j \sigma_j \otimes \ket{\psi_j}\bra{\psi_j} + \Lambda_{\mathrm{ideal}}(\ket{0}\bra{0}_{\mathsf{Out}} \otimes \rho_{\mathrm{off}})
\end{align}
To understand the second term in the equation above, let us consider the purification of the channel $\Lambda_{ideal}$. This is a unitary $U$ that takes in three registers, which we may label $\mathsf{Aux}, \mathsf{Out}, \mathsf{In}$. By \Cref{eq:ideal-energy-measurement}, it follows that for all $j$,
\begin{align}
   U(\ket{0}_{\mathsf{Aux}} \otimes \ket{0}_{\mathsf{Out}} \otimes \ket{\psi_j}_{\mathsf{In}}) &= \ket{\chi_j}_{\mathsf{Aux}, \mathsf{Out}} \otimes \ket{\psi_j}_{\mathsf{In}},
\end{align}
where $\ket{\chi_j}$ is some normalized state. Hence, for all $j \neq k$,
\begin{align}
    \Lambda_{\mathrm{ideal}}( \ket{0}\bra{0}_{\mathsf{Out}} \ot \ket{\psi_j}\bra{\psi_k}_{\mathsf{In}}) &= \tr_{\mathsf{Aux}}[U (\ket{0}_{\mathsf{Aux}} \ket{0}_{\mathsf{Out}} \ket{\psi_j}_{\mathsf{In}}) (\bra{0}_{\mathsf{Aux}} \bra{0}_{\mathsf{Out}} \bra{\psi_k}_{\mathsf{In}})U^\dagger] \\
    &= \tr_{\mathsf{Aux}}[\ket{\chi_j}\bra{\chi_k}_{\mathsf{Aux, Out}} \ot \ket{\psi_j} \bra{\psi_k}_{\mathsf{In}}]  \\
    &= \tr_{\mathsf{Aux}}[\ket{\chi_j}\bra{\chi_k}_{\mathsf{Aux, Out}}] \otimes \ket{\psi_j} \bra{\psi_k}_{\mathsf{In}}.
\end{align}
In other words, the channel maps the off-diagonal component to an output matrix with only off-diagonal components as well. Now, ultimately, we are only interested in the expectation value of measurements on the $\mathsf{Out}$ register, so we may take the partial trace of the $\mathsf{In}$ register. Upon taking the partial trace, all the off-diagonal terms vanish, yielding:
\begin{align}
    \tr_{\mathsf{In}}[ \Lambda_{\mathrm{ideal}}(\ket{0}\bra{0}_{\mathsf{Out}} \otimes \rho)] &= \sum_j \lambda_j \sigma_j + \tr_{\mathsf{In}}[ \Lambda_{\mathrm{ideal}}(\ket{0}\bra{0}_{\mathsf{Out}} \otimes (\rho_{\mathrm{off}})_{\mathsf{In}}] \\
    &= \sum_j \lambda_j \sigma_j + 0.
\end{align}
Now, suppose we measure the resulting $\mathsf{Out}$ state in the standard basis to obtain a measured energy $E \in \{0,1, \dots, 2^{n-1}\}$. Since we constructed $H$ to have eigenvalues that are well-separated from $0$ and $1$, we are guaranteed that overflow error never occurs as long as $1/2^n$ is much less than $1 - 5\sqrt{2}/10$. Supposing this is true and, the expectation value of the measurement is guaranteed to satisfy
\begin{equation}
    \left| \frac{1}{2^n} \E[E \leftarrow \Lambda_{\mathrm{ideal}}] -  \tr[H \sum_j \lambda_j \sigma_j]\right| \leq \frac{1}{2^n}.
\end{equation}
This guarantee is for the output of $\Lambda_{\mathrm{ideal}}$. For the output of the algorithm, we thus have that
\begin{equation}
    \left| \frac{1}{2^n} \E[E \leftarrow \Lambda] -  \tr[H \sum_j \lambda_j \sigma_j] \right| \leq \frac{1}{2^n} + \delta.
\end{equation}
To complete the proof of the lemma, we must now choose $\alpha, \delta$, and $n$. If we set $\alpha = 1/4$, $\delta = \eps/2$, and $2^n = 2/\eps$, we obtain that the algorithm returns an estimate of the energy (and thus the eigenvalue of $\mathscr{B}$) that is accurate up to error $\eps$. The number of queries is 
$Q = O(\eps^{-1} \log \eps^{-1})$. It remains only to bound the
\emph{runtime} of the algorithm. This can be seen to be polynomial in
the number of queries by examining the circuit for the algorithm given
in the proof of Theorem 15 of~\cite{rall21}.
\end{proof}

\begin{lemma}\label{lem:bplusminus-ind}
Let $B^0, B^1$ be two QPT-measurable binary observables, and define the (not necessarily binary) observables
\[ \mathscr{B}_{\pm} = \frac{B^0 \pm B^1}{\sqrt{2}}\]
as in \Cref{lem:estimate-blockenc}. Further let $D_1, D_2$ be any two QPT sampleable distributions over plaintext Alice questions, and let $\ket{\psi}$ be any efficiently preparable Alice state. Then, there exists a negligible function $\delta_{\crp}(\lambda)$ such that, and for any $s \in \{+,-\}$,
\begin{equation}
\left|  \E_{x \leftarrow D_1} \E_{c = \Enc(x)} \sum_{\alpha} \bra{\psi} (A^c_\alpha)^\dagger \mathscr{B}_s (A^c_\alpha) \ket{\psi} -\E_{x \leftarrow D_2} \E_{c = \Enc(x)} \sum_{\alpha}  \bra{\psi} (A^c_\alpha)^\dagger \mathscr{B}_s (A^c_\alpha) \ket{\psi} \right| \leq \delta_{\crp}(\lambda). \label{eq:bplusminus-ind}
\end{equation}
\end{lemma}
\begin{proof}
    We show this by contridiction. Suppose the lemma is false for some $s$. Then there exists some polynomial function $f(\lambda)$ such that for infinitely many $\lambda$, the RHS of \Cref{eq:bplusminus-ind} is greater than or equal to $1/f(\lambda)$.

    Now, choose $\eps(\lambda)$ so that $\eps^{-1}(\lambda)$ is a polynomial function of $\lambda$ and $\eps^{-1}(\lambda) > 100 f(\lambda)$ for all sufficiently large $\lambda$---such an $\eps$ exists since $f$ is a polynomial function of $\lambda$. 
    Let $\{M_{s,\beta}\}$ be the POVM guaranteed by \Cref{lem:estimate-blockenc} applied with this choice of $\eps$ to the block encoding for $\mathscr{B}_s$ given by \Cref{lem:block-enc-bplusminus}. Let $M_s = \sum_\beta \beta M_{s,\beta}$ be the corresponding observable. Then we obtain that for infinitely many $\lambda$
    \begin{equation}
\left|  \E_{x \leftarrow D_1} \E_{c = \Enc(x)} \sum_{\alpha} \bra{\psi} (A^c_\alpha)^\dagger M_{s} (A^c_\alpha) \ket{\psi} -\E_{x \leftarrow D_2} \E_{c = \Enc(x)} \sum_{\alpha}  \bra{\psi} (A^c_\alpha)^\dagger M_s (A^c_\alpha) \ket{\psi} \right| \geq \frac{1}{f(\lambda)} - \eps \geq \frac{0.99}{f(\lambda)} . 
\label{eq:bplusminus-obs-ind}
\end{equation}
    This is now a contradiction to \Cref{lem:mean-to-ind}, which says that the RHS of \Cref{eq:bplusminus-obs-ind} must be a negligible function of $\lambda$. 
\end{proof}

\begin{lemma}\label{lem:bplusminus-squared-ind}
Let $B^0, B^1$ be two QPT-measurable binary observables, and define the (not necessarily binary) observables
\[ \mathscr{B}^2_{\pm} = \frac{(B^0 \pm B^1)^2}{2}\]
as in \Cref{lem:estimate-blockenc}. Further let $D_1, D_2$ be any two QPT sampleable distributions over plaintext Alice questions, and let $\ket{\psi}$ be any efficiently preparable Alice state. Then, there exists a negligible function $\delta_{\crp}(\lambda)$ such that, and for any $s \in \{+,-\}$,
\begin{equation}
\left|  \E_{x \leftarrow D_1} \E_{c = \Enc(x)} \sum_{\alpha} \bra{\psi} (A^c_\alpha)^\dagger \mathscr{B}^2_s (A^c_\alpha) \ket{\psi} -\E_{x \leftarrow D_2} \E_{c = \Enc(x)} \sum_{\alpha}  \bra{\psi} (A^c_\alpha)^\dagger \mathscr{B}^2_s (A^c_\alpha) \ket{\psi} \right| \leq \delta_{\crp}(\lambda).
\end{equation}
\end{lemma}
\begin{proof}
    Analogously to the proof of \Cref{lem:bplusminus-ind}, combine \Cref{lem:block-enc-bplusminus-squared}, \Cref{lem:estimate-blockenc}, and \Cref{lem:mean-to-ind}.
\end{proof}

\begin{lemma}\label{lem:commutator-ind}
Let $B^0, B^1$ be two QPT-measurable binary observables, and define the (not necessarily binary) observables
\[ \mathscr{O}_{\pm} = (B^0 B^1 \pm B^1 B^0)^\dagger (B^0 B^1 \pm B^1 B^0) = \pm (B^0 B^1 \pm B^1 B^0)^2.\]
as in \Cref{lem:block-enc-commutator}. Further let $D_1, D_2$ be any two QPT sampleable distributions over plaintext Alice questions, and let $\ket{\psi}$ be any efficiently preparable Alice state. Then, there exists a negligible function $\delta_{\crp}(\lambda)$ such that, and for any $s \in \{+,-\}$,
\begin{equation}
\left|  \E_{x \leftarrow D_1} \E_{c = \Enc(x)} \sum_{\alpha} \bra{\psi} (A^c_\alpha)^\dagger \mathscr{O}_s (A^c_\alpha) \ket{\psi} -\E_{x \leftarrow D_2} \E_{c = \Enc(x)} \sum_{\alpha}  \bra{\psi} (A^c_\alpha)^\dagger \mathscr{O}_s (A^c_\alpha) \ket{\psi} \right| \leq \delta_{\crp}(\lambda).
\end{equation}
\end{lemma}
\begin{proof}
    Analogously to the proof of \Cref{lem:bplusminus-ind}, combine \Cref{lem:block-enc-commutator}, \Cref{lem:estimate-blockenc}, and \Cref{lem:mean-to-ind}. 
\end{proof}

\subsection{Compiling nonlocal games using cryptography: the KLVY transformation}

Kalai, Lombardi, Vaikuntanathan and Yang give a transformation that maps a $k$-player $1$-round nonlocal game into a $2k$-message ($k$-round) interactive protocol between a single prover and verifier. For simplicity, we will only present their transformation as it is applied to two-player nonlocal games, because this is the only context in which we need to use it. The general transformation, applicable to $k$-player nonlocal games for arbitrary $k$, is described in \cite[Section 3.2]{KLVY21}. The following presentation is taken with some modifications from \cite[Section 3.1]{KLVY21}.

\cite{KLVY21} presents a $\PPT$-computable transformation $\cT$ that converts any $2$-prover non-local game $G$ with question set $\cQ$ and verification predicate $V$ into a single-prover protocol $\cT^G$ (associated with security parameter $\secp$), defined as follows. The main theorem about this transformation which the authors of \cite{KLVY21} prove is presented in \Cref{thm:klvy-main}.
\begin{definition}[Compiled version of nonlocal game $G = (\cQ, V)$]
\label{def:compiled-game}
Fix a quantum homomorphic encryption scheme $\LQHE=(\Gen,\Enc,\Eval,\Dec)$.
\begin{enumerate}
    \item The verifier samples $(x, y)\gets\cQ$, $\sk\gets\Gen(1^\secp)$, and $c\gets\Enc(\sk,x)$. The verifier then sends $c$ to the prover as its first message.
    \item The prover replies with a message $\alpha$.
    \item The verifier sends $y$ to the prover in the clear.
    \item The prover replies with a message $b$.
    \item Define $a := \Dec(\sk,c)$. The verifier accepts if and only if $V(x,y,a,b)=1$.
\end{enumerate}
\end{definition}

\subsubsection{The value of a compiled game}
We make use of the `computationally sound value' or `CS value' as defined
in \cite[Definition 3.1]{KLVY21}. As a shorthand, we may refer to the CS value of a single-prover protocol simply as the `value'.
\begin{definition}
  A single-prover interactive protocol $G$, specified by an interactive verifier Turing machine $V$, has classical CS value $\geq \omega$ if and only if there exists an interactive PPT Turing machine $P$ 
  such that for every $\lambda \in \mathbb{N}$,
  \begin{equation}
      \Pr[ \langle P, V \rangle(1^\lambda) = 1 ] \geq \omega,
  \end{equation}
  where the probability is taken over the random coin tosses of $V$, and where $\langle P, V \rangle$ denotes the output bit of $V(1^{\lambda})$ after interacting with $P$.
\end{definition}
\begin{definition}
  A single-prover interactive protocol $G$, specified by an interactive verifier Turing machine $V$, has quantum CS value $\geq \omega^*$ if and only if there exists an interactive QPT Turing machine $P$ 
  such that for every $\lambda \in \mathbb{N}$,
  \begin{equation}
      \Pr[ \langle P, V \rangle(1^\lambda) = 1 ] \geq \omega^*,
  \end{equation}
  where the probability is taken over the random coin tosses of $V$, and where $\langle P, V \rangle$ denotes the output bit of $V(1^{\lambda})$ after interacting with $P$.
\end{definition}

The following theorem is identical to \cite[Theorem 3.2]{KLVY21} (except for notational changes), and guarantees that the quantum value of $\cT^G$ is at least that of $G$ and that the classical value of $\cT^G$ is no more than that of $G$ (with respect to efficient provers). To state the theorem precisely, we must define the quantum circuit associated with the Alice measurements in a game strategy (since these will be performed using quantum homomorphic encryption).
\begin{definition}
    For a strategy $\mathscr{S}$ with Alice measurements $\{A^{x}_a\}_{a}$, the \emph{Alice circuit} is the following unitary acting on $(\mathbb{C}^2)^{\otimes n_1} \otimes \mathcal{H}_A \otimes (\mathbb{C}^2)^{\otimes m_1}$
    \begin{align}
        C_A &= \sum_{x \in \{0,1\}^{n_1}} \sum_{a \in \{0,1\}^{m_1}} \ket{x}\bra{x} \otimes A^x_a \otimes \Big( \sum_{z \in \{0,1\}^{m_1}}\ket{z + a}\bra{z}\Big),
    \end{align}
    where the addition is taken over $\mathbb{Z}_2^{m_1}$.
    Operationally, this corresponds to coherently performing the measurement $A^x_a$ controlled on the value $x$ in the first register, and adding the outcome to the contents of the third register.
\end{definition}

\begin{theorem}
\label{thm:klvy-main}
Fix any quantum homomorphic encryption scheme $\LQHE$ for a circuit class $\cC$, and any 2-player non-local game $G=(\cQ,V)$ with classical value $\omega$ and quantum value $\omega^*$, such that the value $\omega^*$ is obtained by a prover strategy $\mathscr{S}$ with a quantum state $\ket{\psi}\in\cH_\cA\otimes \cH_\cB$ and Alice circuit $C_A$ where $C_A \in \cC$. Denote by $|q_A|$ and $|s_A|$ respectively the length of the question given to Alice in $G$ and the length of Alice's (decrypted) answer in $G$. If $\LQHE$ is IND-CPA secure against all \emph{classical} distinguishers running in time polynomial in $T(\secp)=2^{|q_A|+|s_A|}\cdot\poly(\secp)$, then the following holds:
\begin{enumerate}
    \item There exists a strategy for $\cT^G$ which can be executed in quantum polynomial time (polynomial in $\lambda$ and the size of the prover strategy $(C^*_A,C^*_B)$) and which attains quantum CS value at least $\omega^*$. 
    \item Any strategy for $\cT^{G}$ that can be executed in classical probabilistic polynomial time (polynomial in $T(\lambda)$) has CS value at most $\omega+\sf{negl}(\secp)$.
\end{enumerate}
\end{theorem}

\subsubsection{Modeling prover strategies in a compiled game}\label{sec:general-modeling}
Let the initial state used by the prover in a compiled game (with syntax specified by \Cref{def:compiled-game}) be denoted $\ket{\psi}$. 
In the first round, the prover receives an encrypted Alice question $c$, and computes an an encrypted answer $\alpha$. 

In general, the prover's action can be modeled as follows. The prover
starts with some initial (pure) state $\ket{\psi}$. In the first
round, it performs a POVM measurement depending on the ciphertext
question $c$ to obtain an outcome $\alpha$, \emph{followed} by a unitary depending on $c$ and $\alpha$, to obtain a post-measurement state. This post-measurement unitary is usually not relevant\footnote{For an instance where it is relevant in the nonlocal setting, see~\cite[page 18]{DSV13}.} in the study of nonlocal games since the provers act on separate subsystems, but it is crucial to consider in our setting because both ``provers" act sequentially on the same quantum register. 

By the Naimark dilation theorem, the prover's POVM measurement depending on $c$ can be simulated by a projective measurement. Thus, to specify the prover's behavior in the first round, we need to specify a collection of projective measurements $\{\Pi^{c}_{\alpha}\}_{\alpha}$ and unitaries $U_{c,\alpha}$. We can unify these into a single set of matrices indexed by $c, \alpha$. Specifically, we model the action of the prover by a collection of \textbf{non-Hermitian} operators $A^{c}_\alpha = U_{c,\alpha} \Pi^c_{\alpha}$. These are (non-positive) ``square-roots" of the projectors corresponding to the measurement applied by the prover. More precisely, they satisfy the following conditions.
\begin{enumerate}
    \item For any given $c$, the collection of positive Hermitian matrices 
    \[ \{(A^{c}_{\alpha})^\dagger (A^{c}_\alpha)\}_{\alpha} \]
    forms a projective measurement. Thus, the probability that Alice returns outcome $\alpha$ in response to question $c$ is
    \[ \Pr[\alpha] = \bra{\psi} (A^c_\alpha)^\dagger (A^c_\alpha) \ket{\psi}. \]
    \item The \emph{un-normalized post-measurement state} after receiving question $c$ and responding with answer $\alpha$ is
    \begin{equation} \ket{\psi^{c}_{\alpha}} = A^c_\alpha \ket{\psi}. \label{eq:def-post-meas-states-chsh} \end{equation}
    Note that $\| \ket{\psi^c_\alpha}\|^2 = \Pr[\alpha]$. 
\end{enumerate}

Next, the verifier sends the prover a second question $y$, this time in the clear, and the prover measures the state to obtain an outcome $b$. Once again, in general, the prover may apply a POVM measurement followed by a unitary depending on the question and the measurement outcome. However, the measurement may be assumed to be projective by once again applying the Naimark dilation theorem, and the post-measurement unitary is irrelevant because we will no longer interact with the prover. Hence, to model this step, it suffices to specify a collection of projective measurements
\[ \{B^{y}_{b}\}. \]

%%% Local Variables:
%%% mode: latex
%%% TeX-master: "../protocol"
%%% End:

\section{The computational commutation game}

\label{sec:commutation}

In this section we study the compiled version of a very basic nonlocal game which forms an important subroutine in many nonlocal protocols: the \emph{commutation game}. 
In this game, Alice receives an empty question and returns two outcomes $a_0, a_1$. Bob receives a question $y \in \{0,1\}$ and returns an outcome $b$. The players win iff $b = a_y$.

\subsection{The compiled commutation game}
\label{sec:commutation-compiled-game}
The compiled version of this game is as follows. Note that we do not need to use any cryptography in this compiled game: this is because Alice's question is empty in the nonlocal version of this game, and so there is no Alice question to hide from the prover using cryptography. Intuitively, this game simply certifies that $B^0$ and $B^1$ (defined immediately below) stabilise the same state, which is sufficient to show that they can be simultaneously measured.
\begin{enumerate}
\item Before the interaction begins, the honest prover selects a state $\ket{\psi}$ and two observables $B^0$, $B^1$ (not necessarily binary) which both have $\ket{\psi}$ as an eigenstate. The prover sends an answer $\alpha$ to the verifier. We expect the honest prover to send two outcomes $\alpha_0, \alpha_1$ corresponding to the results of measuring $B^0$ and $B^1$ on $\ket{\psi}$, respectively.
\item The verifier sends a single bit $y \in \{0,1\}$, chosen uniformly at random.
\item The honest prover measures $B^y$ on $\ket{\psi}$ and returns the outcome $s$ to the verifier. The verifier accepts iff $\alpha_y = s$.
\end{enumerate}

\subsection{Modeling the compiled game}
We follow the notation in \Cref{sec:general-modeling} with some modifications. Firstly, since there is no challenge sent in the first round, the corresponding measurement operators are denoted $A_{\alpha}$ (with no question index). 
% We work with a pure un-normalized post-measurement state, given by
% \[ \ket{\psi_\alpha} = A_{\alpha} \ket{\psi}. \]
In the second round, the prover receives a plaintext question $y$ and measures a binary observable $B^y$. We denote the outcome projectors for these observables by $B^y_b$, so 
\begin{align}
    B^y &= \sum_b (-1)^b B^y_b \\
    B^y_b &= \frac{1}{2} (I + (-1)^b B^y)
\end{align}

\subsection{Approximate commutation}
\begin{lemma}\label{lem:commutation-game-rigidity}
For any strategy that succeeds in the compiled commutation game (see \Cref{sec:commutation-compiled-game}) with probability $1 - \eps$, it holds that 
\begin{equation}
    \sum_\alpha \bra{\psi} (A_\alpha)^\dagger \cdot  | [B^0, B^1]|^2 \cdot A_\alpha \ket{\psi} \leq \delta_{\com}(\eps),
\end{equation}
where
\begin{equation}
    \delta_{\com}(\eps) = 128 \eps.
\end{equation}
\end{lemma}
\begin{proof}
The condition for success in the game is that
\begin{align}
    p_{win} &= \E_{y} \sum_{\alpha} \sum_{b = \Dec(\alpha)_y} \bra{\psi} (A_\alpha)^\dagger B^y_b (A_\alpha) \ket{\psi} \\
    &= \frac{1}{2} + \frac{1}{2} \E_{y} \sum_{\alpha} (-1)^{\Dec(\alpha)_y)} \bra{\psi} (A_\alpha)^\dagger B^y (A_\alpha) \ket{\psi}.
\end{align} 
Suppose that $p_{win} \geq 1 - \eps$. Then we automatically obtain that for every $y \in \{0,1\}$, 
\begin{align}
    &\sum_{\alpha} \| B^y A_\alpha \ket{\psi} - (-1)^{\Dec(\alpha)_y} A_\alpha \ket{\psi} \|^2 \\
    &\quad = 2 - 2 (-1)^{\Dec(\alpha)_y} \bra{\psi} (A_\alpha)^\dagger B^y (A_\alpha) \ket{\psi} \\
    &= 4 - 4 \left(\frac{1}{2} - \frac{1}{2} (-1)^{\Dec(\alpha)_y} \bra{\psi} (A_\alpha)^\dagger B^y (A_\alpha) \ket{\psi}\right) \\
    &\leq 4 - 4(1 - 2\eps) \\
    &= 8\eps.
\end{align}
From this, we will deduce approximate commutation of the $B$ observables. First, observe by the triangle inequality:
\begin{align}
    &\sum_\alpha \| B^0 B^1 A_\alpha \ket{\psi} - (-1)^{\Dec(\alpha)_0 + \Dec(\alpha)_1} A_\alpha \ket{\psi} \|^2 \\
    \begin{split}
    &\quad \leq 2\sum_\alpha \| B^0 B^1 A_\alpha \ket{\psi} - B^0 \cdot (-1)^{\Dec(\alpha)_1} \ket{\psi}\|^2 \\
    &\quad + 2\sum_\alpha \| B^0 \cdot (-1)^{\Dec(\alpha)_1} \ket{\psi} - (-1)^{\Dec(\alpha)_1} (-1)^{\Dec(\alpha)_0} \ket{\psi} \|^2
    \end{split} \\
    &\quad \leq 32 \eps.
\end{align}
Moreover, by symmetry, the same holds if we exchange $B^0$ and $B^1$.
Now, expanding the square of the commutator and applying the triangle inequality again, we get
\begin{align}
    \sum_\alpha \| [B^0, B^1] A_\alpha \ket{\psi}\|^2 &= \sum_\alpha \| B^0 B^1 A_\alpha \ket{\psi} - B^1 B^0 A_\alpha \ket{\psi} \|^2 \\
    &\leq  2 \sum_\alpha \| B^0 B^1 A_\alpha \ket{\psi} - (-1)^{\Dec(\alpha)_0 + \Dec(\alpha)_1} A_\alpha \ket{\psi} \|^2  \\
    &\qquad + 2 \sum_\alpha \| B^1 B^0 A_\alpha \ket{\psi} - (-1)^{\Dec(\alpha)_0 + \Dec(\alpha)_1} A_\alpha \ket{\psi} \|^2 
    \\&\leq 128\eps := \delta_{\com}(\eps).
\end{align}

By expanding out the squared norm on the LHS we obtain the conclusion of the lemma.
\end{proof}
%%% Local Variables:
%%% mode: latex
%%% TeX-master: "../protocol"
%%% End:

\section{The computational CHSH game}

In this section we study the compiled version of the important nonlocal game known as the CHSH game. The protocol associated with this nonlocal game is as follows:
\begin{enumerate}
\item The verifier samples two questions $x,y \gets \{0,1\}$ uniformly at random. The verifier sends $x$ to Alice and $y$ to Bob.
\item Alice responds with a bit $a$ and Bob responds with a bit $b$.
\item The verifier accepts if and only if $x \cdot y = a \oplus b$.
\end{enumerate}

The classical value of this game is $\frac{3}{4}$, and the quantum value of this game is $\cos^2(\pi/8) = \frac{1}{2} + \frac{\sqrt{2}}{4}$. This bound on the quantum value is sometimes known as the \emph{Tsirelson bound}.

\subsection{The compiled CHSH game}

\label{sec:chsh-compiled-game}

The compiled version of the CHSH game is as follows.

Before the interaction begins, the honest prover prepares $n$ EPR pairs, and designates half of each pair as an `Alice qubit' and the other half as a `Bob qubit' (so there are $n$ Alice qubits and $n$ Bob qubits).

\begin{enumerate}
\item Fix a homomorphic encryption scheme $\QHE = (\Gen, \Enc, \Dec, \Eval)$ as defined in \Cref{sec:qhe}. The verifier chooses a secret key $\sk \gets \Gen(1^\lambda)$, and samples two questions $x, y$ uniformly at random. The verifier sends $c := \mathsf{Enc}_\sk(x)$ to the prover.
\item The honest prover responds with $\alpha := \mathsf{Enc}_\sk(a)$, a ciphertext obtained by homomorphically evaluating the canonical Alice strategy for CHSH given question $x$ on the Alice qubits.
\item The verifier sends $y$ to the prover in the clear.
\item The prover responds with a bit $b$ in the clear, obtained by evaluating the canonical Bob strategy for CHSH on the Bob qubits.
\item The verifier decrypts $\alpha$ to obtain $a$, and accepts iff $x \cdot y = a \oplus b$.
\end{enumerate}

\subsection{Modeling the cryptographic game}

\label{sec:chsh-modelling}

We recall the notation in \Cref{sec:general-modeling} used to define a
strategy for a compiled game. In the case of the CHSH game, we define
two additional pieces of notation.

First, for the first round, in the case of the CHSH game, the answer $\alpha$ is supposed to be an encryption of a single bit. We may define an associated ``decrypted" binary observable for every \emph{encrypted} ciphertext $c$:

\begin{equation} A^c :=  \sum_{\alpha} (-1)^{\Dec(\alpha)} (A^x_\alpha)^\dagger (A^x_\alpha). \label{eq:crypto-a-observable}
\end{equation}

Since the $A$ measurement is projective, it follows that $A^c$ is indeed a binary observable, viz. it is Hermitian and squares to identity. It is important to note that the quantity $\E_{c = \Enc(x)} A^x$ is \emph{not} necessarily a binary observable as the different $A^c$ measurements may not commute for different ciphertexts $c$ corresponding to the same plaintext question $x$. 

% Next, the verifier sends the prover a second question $y$, this time in the clear, and the prover measures the state to obtain an outcome $b$. Once again, in general, the prover may apply a POVM measurement followed by a unitary depending on the question and the measurement outcome. However, the measurement may be assumed to be projective by once again applying the Naimark dilation theorem, and the post-measurement unitary is irrelevant because we will no longer interact with the prover. Hence, to model this step, it suffices to specify a collection of projective measurements
% \[ \{B^{y}_{b}\}. \]

Next, in the second round, In the case of the CHSH game, the outcome
$b$ is a single bit, so we may define binary observables out of the
projective measurements $\{B^y_b\}$, in the usual way:
\begin{equation}
\label{eq:def-binary-observables-chsh}
    B^y  := \sum_b (-1)^{b} B^y_b. 
\end{equation}

\subsection{Macroscopic locality in the cryptographic game}
\label{sec:macroscopic-locality}

In this section, we present an argument to show that the quantum value of the cryptographically compiled CHSH game is at most $\omega^* + \epsilon$ for some negligible $\epsilon$. This argument is based on a formalisation of arguments contained in \cite{rohrlich2014stronger}. The crux of this argument, assuming the existence of a QPT prover which wins in the compiled CHSH game with probability at least $\omega^* + \epsilon$ for some non-negligible $\epsilon$, is to show a contradiction with the IND-CPA security of the encryption scheme by defining a concrete efficiently measurable operator that allows an adversary to guess the decryption of $c$ with probability better than $\frac{1}{2} + \sf{negl}(\lambda)$ by measuring this operator. In the nonlocal case, this corresponds to an argument which shows that, if two nonlocal provers Alice and Bob win in CHSH with probability better than $\omega^*$, then Bob has it within his power to guess Alice's question $x$ with probability better than $\frac{1}{2}$, using some combination of the measurements with which he would win the game.

Formally, we will show the following lemma.

\begin{lemma}
\label{lem:macloc-main}
If there is a QPT cheating prover $P^*$ for the cryptographically compiled CHSH game, consisting of operators $\{A^c_\alpha\}_{c,\alpha}$ and $\{B^y_b\}_{y,b}$ as defined in \Cref{sec:chsh-modelling}, which wins with probability $\omega^* + \epsilon$ for non-negligible $\epsilon$, where $\omega^* = \cos^2(\pi/8)$, then there is a POVM measurement $\{M_\beta\}_\beta$ with eigenvalues in $[0,1]$ which can be implemented with a polynomial (in $\lambda$) sized circuit such that
\begin{equation}
\left|  \E_{x \leftarrow D_1} \E_{c = \Enc(x)} \sum_{\alpha} \sum_\beta \beta \bra{\psi} (A^c_\alpha)^\dagger M_\beta (A^c_\alpha) \ket{\psi} -\E_{x \leftarrow D_2} \E_{c = \Enc(x)} \sum_{\alpha} \sum_\beta \beta \bra{\psi} (A^c_\alpha)^\dagger M_\beta (A^c_\alpha) \ket{\psi} \right| \leq \eta'(\lambda).
\end{equation}
\end{lemma}

Comparison with \Cref{lem:mean-to-ind} yields a contradiction.

Define binary observables $B^y$ for $y \in \{0,1\}$ as in \Cref{eq:def-binary-observables-chsh}, and define the post-measurement states $\ket{\psi^c_\alpha}$ as they are defined in \Cref{eq:def-post-meas-states-chsh}. In order to show \Cref{lem:macloc-main}, we will show the following:
\begin{lemma}
If there is a QPT cheating prover $P^*$ for the cryptographically compiled CHSH game, consisting of operators $\{A^c_\alpha\}_{c,\alpha}$ and $\{B^y_b\}_{y,b}$ as defined in \Cref{sec:chsh-modelling}, which wins with probability $\omega^* + \epsilon$ for non-negligible $\epsilon$, where $\omega^* = \cos^2(\pi/8)$, then
\begin{gather}
\left( \E_{c \: : \: \sf{Dec}(c) = 0} \sum_\alpha \bra{\psi^c_\alpha} (B^0 + B^1)^2 \ket{\psi^c_\alpha} \right)
-
\left( \E_{c \: : \: \sf{Dec}(c) = 1} \sum_\alpha \bra{\psi^c_\alpha} (B^0 + B^1)^2 \ket{\psi^c_\alpha} \right)
\geq 16\epsilon.
\end{gather}
\end{lemma}

Since $(B^0 + B^1)^2$ is an operator of bounded norm, normalising appropriately yields \Cref{lem:macloc-main}.

Let us use the shorthand notation
\begin{gather}
	\langle X \rangle_0 := \E_{c \: : \: \sf{Dec}(c) = 0} \sum_\alpha \bra{\psi^c_\alpha} X \ket{\psi^c_\alpha} \\
	\langle X \rangle_1 := \E_{c \: : \: \sf{Dec}(c) = 1} \sum_\alpha \bra{\psi^c_\alpha} X \ket{\psi^c_\alpha} \\
\Delta_0(B^0 \pm B^1) := \langle (B^0 \pm B^1)^2 \rangle_0 := \E_{c \: : \: \sf{Dec}(c) = 0} \sum_\alpha \bra{\psi^c_\alpha} (B^0 \pm B^1)^2 \ket{\psi^c_\alpha} \\
\Delta_1(B^0 \pm B^1) := \langle (B^0 \pm B^1)^2 \rangle_1 := \E_{c \: : \: \sf{Dec}(c) = 1} \sum_\alpha \bra{\psi^c_\alpha} (B^0 \pm B^1)^2 \ket{\psi^c_\alpha}.
\end{gather}

In addition, let us define the following notation for two-point correlators.
\begin{align}
    \langle A_x, X\rangle &:= \E_{c: \Dec(c) = x} \sum_\alpha (-1)^{\Dec(\alpha)} \bra{\psi^c_\alpha} X \ket{\psi^c_\alpha} ,
\end{align}
where $X$ is a ``Bob observable" (any linear combination of $B_0$ and $B_1$).

The variances and the correlators are related by the following inequality.
\begin{lemma}\label{lem:tight-chsh-0}
$| \langle A^x, B^0 \pm B^1 \rangle |^2 \leq \Delta_x(B^0 \pm B^1).$
\end{lemma}
\begin{proof}
The proof follows by applying Jensen's inequality twice. Recall that Jensen's says that for a real-valued random variable $X$, $(\E X)^2 \leq \E X^2$. When the random variable $X$ arises from measuring an observable $O$ on a quantum state $\ket{\psi}$, this can be written as $(\bra{\psi} O \ket{\psi})^2 \leq \bra{\psi} O^2 \ket{\psi}$. We use both forms of the inequality below.
\begin{align}
    |\langle A^x, B^0 \pm B^1 \rangle|^2 &= \left( \E_{c: \Dec(c) = x} \sum_\alpha (-1)^{\Dec(\alpha)} \bra{\psi^c_\alpha} (B^0 \pm B^1) \ket{\psi^c_\alpha}\right)^2 \\
    &\leq \E_{c: \Dec(c) = x} \sum_\alpha ( \bra{\psi^c_\alpha} (B^0 \pm B^1) \ket{\psi^c_\alpha})^2 \\
    &\leq \E_{\alpha c: \Dec(c) = x} \sum_\alpha  \bra{\psi^c_\alpha} (B^0 \pm B^1)^2 \ket{\psi^c_\alpha} \\
    &= \Delta_x (B^0 \pm B^1).
\end{align}
\end{proof}

Now we proceed to analysing the game. We firstly make the following observation:
\begin{lemma}
\label{lem:tight-chsh-1}
$\Delta_1(B^0 + B^1) = 4 - \Delta_1(B^0 - B^1).$
\end{lemma}
\begin{proof}
Observe that
\begin{align}
\Delta_1(B^0 + B^1) + \Delta_1(B^0 - B^1) &=
2\langle (B^0)^2 \rangle_1 + 2\langle (B^1)^2 \rangle_1 \\
&= 4.
\end{align}
\end{proof}

Next we prove that:
\begin{lemma}
\label{lem:tight-chsh-2}
$\Delta_0(B^0+B^1) + \Delta_1(B^0-B^1) \geq 4 + 16\epsilon$.
\end{lemma}
\begin{proof}
Let $\delta$ be the real number such that
\begin{equation}
\langle (B^0 + B^1)^2 \rangle_0 +
\langle (B^0 - B^1)^2 \rangle_1
=
4+\delta.
\end{equation}
Note that 
\begin{align}
&\frac{1}{2} + \frac{1}{8} \left( \langle A^0, B^0 + B^1 \rangle +
\langle A^1, B^0 - B^1 \rangle \right) \\
&\quad = 
\frac{1}{2} + \frac{1}{8} \left( \E_{c \: : \: \sf{Dec}(c) = 0} \sum_\alpha (-1)^{\Dec(\alpha)} \bra{\psi^c_\alpha} (B^0 + B^1) \ket{\psi^c_\alpha} \right. \nonumber \\
&\qquad +
\left. \E_{c \: : \: \sf{Dec}(c) = 1} \sum_\alpha (-1)^{\Dec(\alpha)} \bra{\psi^c_\alpha} (B^0 - B^1) \ket{\psi^c_\alpha} \right) \\
&\quad=
\omega^* + \epsilon,
\end{align}
by the CHSH condition.

Moreover, by \Cref{lem:tight-chsh-0}, we have
\[
\Delta_0(B^0 + B^1) = \langle (B^0 + B^1)^2 \rangle_0
\geq
\langle A^0, B^0 + B^1 \rangle^2.
\]
 Similarly,
\[
\Delta_1(B^0 - B^1) = \langle (B^0 - B^1)^2 \rangle_1
\geq
\langle A^1, B^0 - B^1 \rangle^2.
\]
Hence
\begin{align}
	4+\delta &= \langle (B^0 + B^1)^2 \rangle_0 +
\langle (B^0 - B^1)^2 \rangle_1 \\
&\geq \langle A^0, B^0 + B^1 \rangle ^2 + 
\langle A^1, B^0 - B^1 \rangle^2.
\end{align}
Using the inequality $|x+y| \leq (2x^2 + 2y^2)^{1/2}$, we have
\begin{align}
8\left(v^* + \epsilon - \frac{1}{2}
\right) &=
\big|\langle A^0, B^0 + B^1 \rangle + \langle A^1, B^0 - B^1 \rangle \big| \\
&\leq \sqrt{2 \big(\langle A^0, B^0 + B^1 \rangle ^2 + 
\langle A^1, B^0 - B^1 \rangle^2 \big)} \\
&\leq \sqrt{2 \big(\langle  (B^0 + B^1)^2 \rangle_0 + 
\langle (B^0 - B^1)^2 \rangle_1 \big)} \\
&\leq \sqrt{2(4+\delta)} \\
&= \sqrt{2}(4+\delta)^{1/2} \\
&\leq 2\sqrt{2} + \frac{1}{2}\delta. \\
&\implies \delta \geq 16\epsilon.
\end{align}
\end{proof}

Now, putting Lemmas \ref{lem:tight-chsh-1} and \ref{lem:tight-chsh-2} together, we get that
\[\Delta_0(B^0 + B^1) - \Delta_1(B^0 + B^1) \geq 16 \epsilon.\]
This concludes the proof.

\subsection{ An SoS for the cryptographic game}
\label{sec:chsh-sos}

\subsubsection{Warmup: SoS in the nonlocal case}
We start by reviewing the argument that obtains the Tsirelson bound in the nonlocal case using a sum-of-squares decomposition.

Define the \emph{game polynomial} to be the following polynomial in
the binary observables $A^0, A^1, B^0, B^1$ used by Alice and Bob,
respectively, given a question $0$ or $1$.
\begin{equation}
  p_{CHSH} = A^0 B^0 + A^0 B^1 + A^1 B^0 - A^1 B^1.
\end{equation}
For any strategy $\mathscr{S}$,
\[ \omega^*(G_{CHSH}, \mathscr{S}) = \frac{1}{2} + \frac{1}{8} \bra{\psi} p_{CHSH} \ket{\psi}. \]

The operators $A^0, A^1, B^0, B^1$ satisfy certain constraints. First,
they must each square to the identity $\Id$, since each one is a binary
observable. Second, the Alice and Bob operators must commute with each
other: for all $a, b \in \{0,1\}$, it holds that $A^a B^b =
B^bA^a$. Subject to these constraints, we will show that the following sum-of-squares decomposition of $p_{CHSH}$ holds: 

\begin{align}
  p_{CHSH} &= 2\sqrt{2} \cdot \Id - \frac{\sqrt{2}}{2} (q_1^2 + q_2^2), \\
  q_1 &= A^0 - \frac{B^0 + B^1}{\sqrt{2}} \\
  q_2 &= A^1 - \frac{B^0 - B^1}{\sqrt{2}}.
\end{align}
We can check this by direct computation:
\begin{align}
  q_1^2 &= (A^0)^2 + \frac{1}{2} ((B^0)^2 + (B^1)^2 + B^0 B^1 + B^1
          B^0) - \frac{1}{\sqrt{2}}( A^0 B^0 + B^0 A^0 + A^0 B^1 + B^1
          A^0) \\
        &=  \Id+ \frac{1}{2} (2\Id + B^0 B^1 + B^1
          B^0) - \frac{1}{\sqrt{2}}( A^0 B^0 + B^0 A^0 + A^0 B^1 + B^0
          A^1) \\
  q_2^2 &=  \Id+ \frac{1}{2} (2\Id - B^0 B^1 - B^1
          B^0) - \frac{1}{\sqrt{2}}( A^1 B^0 + B^0 A^1 - A^1 B^1 - B^1A^1
          ) \\
  q_1^2 + q_2^2 &= 4 \Id  - \frac{2}{\sqrt{2}} p_{CHSH},
\end{align}
where we have used the commutation between Alice and Bob operators in
the last line.

Since for any state $\psi$ and $i \in \{1,2\}$ it holds that $\bra{\psi} q_i^2 \ket{\psi} \geq 0$, this implies that $\bra{\psi} p_{CHSH} \ket{\psi} \leq 2\sqrt{2}$.  This in turn implies that $\omega^*(G_{CHSH}) \leq \frac{1}{2} + \frac{\sqrt{2}}{4}$ which is exactly the Tsirelson bound.

There is a more cumbersome way of phrasing the preceding argument, that will lead more naturally to the cryptographic case. (It also arises when one writes down a semidefinite program to search for SoS certificates.) Let us define a matrix $\Gamma$, called the \emph{covariance matrix}, whose rows and columns are indexed by Alice and Bob observables $A^x$ and $B^y$. The entries of $\Gamma$ are defined as follows:
\begin{align}
\Gamma_{A^x A^{x'}} &:= \bra{\psi} A^x A^{x'} \ket{\psi} \\
\Gamma_{B^y A^x} = \Gamma_{A^x B^y} &:= \bra{\psi} A^x B^y \ket{\psi} \\
\Gamma_{B^y B^{y'}} &:= \bra{\psi} B^y B^{y'} \ket{\psi}.
\end{align}
By construction, $\Gamma$ is a Gram matrix and hence it is positive semidefinite. Moreover, the expectation value $\bra{\psi} q_i^2 \ket{\psi}$ can be written in terms of $\Gamma$ and the 4-dimensional vector $\vec{q}_i$ of coefficients of the polynomial $q_i$:
\begin{align}
    q_i &= \vec{q}_{i,1} A^0 + \vec{q}_{i,2} A^1 + \vec{q}_{i,3} B^0 + \vec{q}_{i,4} B^1 \\
    \bra{\psi} q_i^2 \ket{\psi} &= \vec{q}_k^\dagger \Gamma \vec{q}_k.
\end{align}
Since $\Gamma$ is PSD, it follows that $\bra{\psi} q_i^2 \ket{\psi}$ is nonnegative, and the argument proceeds as above.

\subsubsection{The cryptographic case}

In the cryptographic case, the expression $\bra{\psi} p_{CHSH} \ket{\psi}$ (as defined in the previous section) is no longer operationally meaningful: although we could define ``encrypted Alice observables", as in \Cref{eq:crypto-a-observable}, these have no reason to commute with the Bob observables, and therefore quantities like $\bra{\psi} A^x B^y \ket{\psi}$ no longer occur operationally in the protocol. To modify the SoS argument to work in the cryptographic case, we will define a modified ``covariance matrix" $\Gamma$, and replace all expressions of the form $\bra{\psi} p \ket{\psi}$, where $p$ is a polynomial in the provers' operators, with a linear combination of entries of our modified $\Gamma$. We will design our $\Gamma$ so that ultimately, the probability that the prover succeeds in the cryptographic CHSH game can be written as
\begin{align} 
\Pr[\text{win}] &= \omega^*_{CHSH} - \sum_j \vec{v}_j^\dagger \Gamma \vec{v}_j,
\end{align}
for some vectors $\vec{v}_j$. We will show that the terms $\vec{v}_j^\dagger \Gamma \vec{v}_j$ are nonnegative, thus establishing that $\omega^*_{CHSH}$ is an upper bound on the winning probability.

Let us now give the details for our definition of $\Gamma$. As before, it will have rows and columns indexed by $A^{x}$ and $B^y$. We define the entries as follows:
%\anote{As Tina points out, the $AA$ parts of the matrix are actually entirely arbitrary in the case of CHSH except for the diagonals which must be $1$. This is because there are no AA cross terms that appear in each square $p^\dagger p$ of the CHSH game polynomial. This is really a special coincidence for CHSH though, so the more general definition is good. Note also that this definition for the AA part is not necessarily the correct one since it doesn't reduce to the nonlocal case when the encryption is taken to be deterministic.} \znote{I'm not sure I'm on board with having this complicated definition for the AA terms when we're not even sure it's the `right' one for circumstances where it does matter... perhaps we could just set it to 0 or something when $x \neq x'$ to emphasise that it's irrelevant for this particular case?}
% Gamma_{A^x A^{x'}} &:= \begin{cases} 1 & \text{if $x = x'$} \\ \E_{c = \Enc(x), c' = \Enc(x')} (-1)^{\Dec(\alpha)}\cdot \bra{\psi} (A^c_{\alpha})^\dagger A^{c'} (A^c_{\alpha}) \ket{\psi} & \text{otherwise} \end{cases}
\begin{align}
    \Gamma_{A^x A^{x'}} &:= \begin{cases} 1 & \text{if $x = x'$} \\ 0 & \text{otherwise} \end{cases} \\
    \Gamma_{B^y A^x} = \Gamma_{A^x B^y} &:= \E_{c = \Enc(x)} \sum_{\alpha} (-1)^{\Dec(\alpha)} \cdot \bra{\psi} (A^c_\alpha)^\dagger B^y A^c_\alpha \ket{\psi} \\
    \Gamma_{B^y B^{y'}} &:=  \E_{c = \Enc(x)} \sum_{\alpha} \bra{\psi} (A^c_{\alpha})^\dagger B^y B^{y'} (A^c_{\alpha}) \ket{\psi}.
\end{align}

\begin{remark}
The definition of the entries $\Gamma_{A^x A^{x'}}$ is actually almost completely arbitrary: the only constraint they need to satisfy is that the diagonal entries with $x = x'$ must be equal to 1. This is because these `AA cross terms' never appear in the SoS decomposition which we use below. This is, however, a coincidence special to the SoS decomposition for the CHSH game in particular, and any analysis of a different game may need to define the `AA cross terms' more meaningfully.
\end{remark}

% \anote{OLD:}
% \begin{itemize}
%     \item  $\Gamma_{A^x A^{x'}} = \bra{\psi} A^x A^{x'} \ket{\psi}$.
%     %\item $\Gamma_{A^x B^{y}} = \Gamma_{B^y A^x} = \sum_{a \in \{\pm 1\}} a \bra{\psi} A^x_a B^y A^x_a \ket{\psi}$.
%     \item $\Gamma_{A^x B^y } = \Gamma_{B^x A^y} = \sum_{\alpha} \Dec(\alpha) \cdot \bra{\psi} A^x_{\alpha} B^y A^x_{\alpha} \ket{\psi}.$
%     %\item $\Gamma_{B^y B^{y'}} = \E_{x} \sum_{a \in \{\pm 1\}} \bra{\psi} A^x_a B^y B^{y'} A^x_a \ket{\psi}$.
%      \item $\Gamma_{B^y B^{y'} } = \E_{x} \sum_{\alpha}  \bra{\psi} A^x_{\alpha} B^y B^{y'} A^x_{\alpha} \ket{\psi}.$
% \end{itemize}

Let us emphasize that $\Gamma$ is in many ways \emph{not} like a covariance matrix. In particular, while it is Hermitian, it is \emph{not} necessarily positive semidefinite. However, we will show that it is ``close enough" to looking like a covariance matrix to enable us to use it to bound the game value.

Note that by this definition, the diagonal entries of $\Gamma$ are equal to 1. For the $A^x A^x$ entries this is true by definition. For the $B^y B^y$ entries we have
\begin{align}
    \Gamma_{B^y B^y} &= \E_{c= \Enc(x)} \sum_{\alpha} \bra{\psi} (A^c_\alpha)^\dagger (B^y)^2 (A^c_\alpha) \ket{\psi} \\
    &= \E_{c= \Enc(x)} \sum_{\alpha} \bra{\psi} (A^c_\alpha)^\dagger  (A^c_\alpha) \ket{\psi} \\
    &= 1.
\end{align}

The $AB$ entries of $\Gamma$ have an operational meaning: when the verifier samples a question pair $x,y$ in the crypto game and receives (decrypted) answers $a,b$, then $\Gamma_{A^x B^y}$ is precisely the expected value of $(-1)^{a \cdot b}$.
\begin{equation} \Gamma_{A^x B^y} = \E_{game}[ (-1)^{a \cdot b} | x,y]. \label{eq:gamma-ab-meaning} \end{equation}

This implies the following: suppose we have a game polynomial for a nonlocal XOR game:
\begin{equation} p_{game} = \sum_{x,y} (-1)^{s(x,y)} A^x B^y, \end{equation}
together with an associated SoS decomposition
\begin{equation} p_{game} = \omega^* \cdot \Id - \sum_j b_j (q_j^\dagger q_j) + \sum_i c_i [A^{x_i}, B^{y_i}] + \sum_j d_j (\Id - (A^{x_j})^2) + \sum_k e_k (\Id - (B^{y_k})^2),  \label{eq:sos-general} \end{equation}
where the coefficients $b_j$ are real and nonnegative, and the other coefficients are arbitrary complex numbers. Here we have written out the constraint terms explicitly, and the two sides are equal as formal polynomials. For instance, such a decomposition exists for the CHSH game:
\begin{align}
    p_{game} &= \frac{1}{2} \cdot \Id + \frac{1}{8} (A^0 B^0 + A^0 B^1 + A^1 B^0 - A^1 B^1) \\
    &= \underbrace{\left(\frac{1}{2} + \frac{\sqrt{2}}{4} \right)}_{\omega^*} \cdot \Id - \frac{\sqrt{2}}{16} \left(q_1^\dagger q_1 + q_2^\dagger q_2  + f + g \right)\\
    q_1 &:= A^0 - \frac{B^0 + B^1}{\sqrt{2}} \\
    q_2 &:= A^1 - \frac{B^0 - B^1}{\sqrt{2}}\\
    f &:= \sum_x (\Id - (A^x)^2) + \sum_y (\Id - (B^y)^2) \\
    g &:= \frac{1}{\sqrt{2}} ([B^0, A^0] + [B^0, A^1] + [B^1, A^0] - [B^1, A^1]).
    \end{align} 
    
    To see this, calculate:
    \begin{align}
    q_1^\dagger q_1 &= (A^0)^2 - \frac{A^0 B^0 + A^0 B^1 + B^0 A^0 + B^1 A^0}{\sqrt{2}} + \frac{(B^0)^2 + (B^1)^2 + B^0 B^1 + B^1 B^0}{2} \\
    q_2^\dagger q_2 &= (A^1)^2 - \frac{A^1 B^0 - A^1 B^1 + B^0 A^1 - B^1 A^1}{\sqrt{2}} + \frac{(B^0)^2 + (B^1)^2 - B^0 B^1 - B^1 B^0}{2} \\
    q_1^\dagger q_1 + q_2^\dagger q_2 &= 4 \cdot \Id - \sqrt{2}(A^0B^0 + A^0B^1 + A^1 B^0 - A^1 B^1) \nonumber \\
    &\qquad - (\Id - (A^0)^2) - (\Id - (A^1)^2) - (\Id - (B^0)^2) - (\Id - (B^1)^2) \nonumber \\
    &\qquad  - \frac{1}{\sqrt{2}} ([B^0, A^0] + [B^0, A^1] + [B^1, A^0] - [B^1, A^1]).
\end{align}

Observe that the game polynomial and the SoS decomposition both have the form  $\nu \cdot \Id + h$, where $h$ is a \emph{homogeneous} degree-2 polynomial in the variables $A^x, B^y$. Define the linear operator $\tilde{\E}[\cdot]$ mapping such polynomials to complex numbers by the following:
\begin{itemize}
    \item $\tilde{\E}[\Id] = 1$.
    \item $\tilde{\E}[\cdot]$ acting on a monomial of degree 2 in the $A$ and $B$ variables maps it to the corresponding entry of $\Gamma$, e.g. $\tilde{\E}[A^x B^y] = \Gamma_{A^x B^y}$.
    \item Extend this by linearity to all polynomials of the form $\nu \cdot \Id + h$.
\end{itemize}
This operator can be thought of as a ``pseudo-expectation" mapping polynomials in the formal $A, B$ variables to the expectation value of the corresponding operators on the state. In particular, it has the following two properties:
\begin{enumerate}
    \item The pseudo-expectation of the game polynomial is equal to the winning probability of the strategy used to construct $\Gamma$, by \Cref{eq:gamma-ab-meaning}:
\begin{equation} \tilde{\E}[p_{game}] = \Pr[\text{win}]. \end{equation}
    \item  The pseudo-expectation of the ``constraint terms" in the SoS decomposition is 0. 
    \begin{align}
        \tilde{\E}[\Id - (A^x)^2] &= \tilde{\E}[\Id - (B^y)^2] = 1 -1 = 0 \\
         \tilde{\E}[A^x B^y- B^x A^y] &= \Gamma_{A^x B^y} - \Gamma_{B^yA^x} = 0.
    \end{align}
\end{enumerate}
 
Hence, by applying $\tilde{\E}$ to both sides of \Cref{eq:sos-general} we get that
\begin{equation} \Pr[\text{win}] = \omega^* - \sum_j b_j( \tilde{\E}[q_j^\dagger q_j] ) = \omega^* - \sum_j b_j \cdot \sum_{k, k'} q_{jk}^* q_{jk'} \Gamma_{O_k O_{k'}} = \omega^* - \sum_j b_j( q_j^\dagger \Gamma q_j ) , \label{eq:win-prob-decomp} \end{equation}
where we have decomposed $q_j$ as a sum of variables $q_j = \sum_{k} q_{jk} O_k$ with each $O_k$ either an $A$ or a $B$ variable.

Thus, if we could show that $q_j^\dagger \Gamma q_j \geq 0$ for every $q_j$ that appears in the SoS for $p_{game}$, then it would follow the highest attainable value in the crypto game is $\omega^*$, which is the optimum deduced by NPA level 1. Note that this is weaker than showing that $\Gamma$ is positive semidefinite. The goal of the remainder of this section is to show this property for $q_1, q_2$ appearing in the SoS for the CHSH game.

We will start by showing $q_1^\dagger \Gamma q_1 \geq 0$. The calculation for $q_2$ will be exactly analogous. 
To show that $q_1^\dagger \Gamma q_1 \geq 0$, we will show that it is equal to an expectation of a square under some probability distribution. Specifically, define random variables $a, b$ with the joint distribution $\mu_1$ defined by the following process:

\begin{definition}[Probability distribution $\mu_1$]
$\ $
\begin{enumerate}
    \item First, generate a random encryption $c = \Enc(0)$ and measure the projective measurement $\{A^c_\alpha\}_\alpha$ on $\ket{\psi}$ to obtain an outcome ciphertext $\alpha$, and let $a \in \pm 1$ be obtained from the decryption of $\alpha$ by $a = (-1)^{\Dec(\alpha)}$.
    \item Next, on the post-measurement state, measure the observable $(B^0 + B^1)/\sqrt{2}$ to obtain an outcome $b \in \mathbb{R}$. Note that a priori, we cannot say anything about the possible values $b$ can take, other than that they are real (they are the eigenvalues of $(B^0 + B^1)/\sqrt{2}$).
\end{enumerate}
\end{definition}

\begin{claim}
\label{claim:q_1}
There exists a function $\delta_\crp(\lambda) = \sf{negl}(\lambda)$ such that
\[ q_1^\dagger \Gamma q_1 \approx_{\delta_{\crp}} \E_{\mu_1}[ (a - b)^2].\]
\end{claim}
\begin{proof}
For notational convenience, define $\scrB = (B^0 + B^1)/\sqrt{2}$, and the associated outcome projectors $\scrB_b$. Observe that
\[ \scrB = \sum_b b \cdot \scrB_b,\quad \scrB^2 = \sum_b b^2 \cdot \scrB_b. \]
Now we may calculate:
\begin{align}
    \E_{\mu_1}[(a-b)^2] &= \E_{c = \Enc(0)} \sum_{\alpha}  \sum_{b \in \mathrm{spec(\scrB)}} \bra{\psi} (A^c_{\alpha})^\dagger \scrB_b A^c_{\alpha} \ket{\psi} ((-1)^{\Dec(\alpha)} -b)^2 \\
    &= \E_{c = \Enc(0)} \sum_{\alpha} \sum_{b \in \mathrm{spec(\scrB)}} \bra{\psi} (A^c_{\alpha})^\dagger \scrB_b A^c_{\alpha} \ket{\psi} (1  + b^2 - 2 \cdot (-1)^{\Dec(\alpha)}\cdot b) \\
    &= 1 + \E_{c = \Enc(0)} \sum_{\alpha} \sum_{b \in \mathrm{spec(\scrB)}} \bra{\psi} (A^c_{\alpha})^\dagger \scrB_b A^c_{\alpha} \ket{\psi} b^2 \nonumber \\
        &\qquad - 2 \E_{c = \Enc(0)} \sum_{\alpha} \sum_{b \in \mathrm{spec(\scrB)}} \bra{\psi} (A^c_{\alpha})^\dagger \scrB_b A^c_{\alpha} \ket{\psi} (-1)^{\Dec(\alpha)} \cdot b \\
    &= 1 + \E_{c = \Enc(0)} \sum_{\alpha}  \bra{\psi} (A^c_{\alpha})^\dagger \scrB^2 A^c_{\alpha} \ket{\psi}  - \E_{c = \Enc(0)} 2\sum_{\alpha}  \bra{\psi} (A^c_{\alpha})^\dagger \scrB A^c_{\alpha} \ket{\psi} (-1)^{\Dec(\alpha)} \label{eq:eabsq-obs} \\
    &= 1 + \E_{c = \Enc(0)} \sum_\alpha \bra{\psi} (A^c_{\alpha})^{\dagger} \frac{(B^0 + B^1)^2}{2} A^c_{\alpha} \ket{\psi} \nonumber \\
    &\qquad - 2\E_{c = \Enc(0)} \sum_{\alpha} \bra{\psi} (A^c_{\alpha})^{\dagger} \frac{(B^0 + B^1)}{\sqrt2} A^c_{\alpha} \ket{\psi} (-1)^{\Dec(\alpha)} \label{eq:approx-delta-crypto-q1}
\end{align}
Now, applying \Cref{lem:bplusminus-ind} and \Cref{lem:bplusminus-squared-ind}, there exists a function $\delta_\crp(\lambda) = \sf{negl}(\lambda)$ such that
\begin{align}
    \text{\eqref{eq:approx-delta-crypto-q1}} &\approx_{\delta_\crp} 1 + \E_{b \in \{0,1\}} \E_{c' = \mathrm{Enc}(b)} \sum_{\alpha} \bra{\psi} (A^{c'}_{\alpha})^\dagger \frac{(B^0 + B^1)^2}{2} A^{c'}_{\alpha} \ket{\psi} \nonumber \\
    &\qquad - \E_{c = \Enc(0)} 2\sum_\alpha \bra{\psi} (A^c_{\alpha})^\dagger \frac{(B^0 + B^1)}{\sqrt2} A^c_{\alpha} \ket{\psi} (-1)^{\Dec(\alpha)} \label{eq:eabsq-crypto} \\
    &= 1 + 1 + \frac{1}{2} (\Gamma_{B^0 B^1} + \Gamma_{B^1 B^0} ) - \frac{2}{\sqrt{2}} (\Gamma_{A^0 B^0}  + \Gamma_{A^0 B^1} ) \\
    &= \Gamma_{A^0 A^0} + \frac{1}{2} (\Gamma_{B^0 B^0} + \Gamma_{B^1 B^1} + \Gamma_{B^0 B^1} + \Gamma_{B^1 B^0}) - \frac{1}{\sqrt{2}} (\Gamma_{A^0 B^0} + \Gamma_{B^0 A^0} + \Gamma_{A^0 B^1} + \Gamma_{B^1 A^0}) \\
    &= q_1^\dagger \Gamma q_1.
\end{align}
\end{proof}

Now we move onto showing that $q_2^\dagger \Gamma q_2 \geq 0$. Like in the case of $q_1$, to show that $q_2^\dagger \Gamma q_2 \geq 0$, we will show that it is equal to an expectation of a square under some probability distribution. Specifically, define random variables $a, b$ with the joint distribution $\mu_2$ defined by the following process:

\begin{definition}[Probability distribution $\mu_2$]
$\ $
\begin{enumerate}
    \item First, generate a random encryption $c = \Enc(1)$ and measure the projective measurement $\{A^c_\alpha\}_\alpha$ on $\ket{\psi}$ to obtain an outcome ciphertext $\alpha$, and let $a \in \pm 1$ be obtained from the decryption of $\alpha$ by $a = (-1)^{\Dec(\alpha)}$.
    \item Next, on the post-measurement state, measure the observable $(B^0 - B^1)/\sqrt{2}$ to obtain an outcome $b \in \mathbb{R}$. Note that a priori, we cannot say anything about the possible values $b$ can take, other than that they are real (they are the eigenvalues of $(B^0 - B^1)/\sqrt{2}$).
\end{enumerate}
\end{definition}

\begin{claim}
\label{claim:q_2}
There exists a function $\delta_\crp(\lambda) = \sf{negl}(\lambda)$ such that
\[ q_2^\dagger \Gamma q_2 \approx_{\delta_{\crp}} \E_{\mu_2}[ (a - b)^2].\]
\end{claim}
\begin{proof}
For notational convenience, define $\scrB = (B^0 - B^1)/\sqrt{2}$, and the associated outcome projectors $\scrB_b$. Observe that
\[ \scrB = \sum_b b \cdot \scrB_b,\quad \scrB^2 = \sum_b b^2 \cdot \scrB_b. \]
Now we may calculate:
\begin{align}
    \E_{\mu_1}[(a-b)^2] &= \E_{c = \Enc(1)} \sum_{\alpha}  \sum_{b \in \mathrm{spec(\scrB)}} \bra{\psi} (A^c_{\alpha})^\dagger \scrB_b A^c_{\alpha} \ket{\psi} ((-1)^{\Dec(\alpha)} -b)^2 \\
    &= \E_{c = \Enc(1)} \sum_{\alpha} \sum_{b \in \mathrm{spec(\scrB)}} \bra{\psi} (A^c_{\alpha})^\dagger \scrB_b A^c_{\alpha} \ket{\psi} (1  + b^2 - 2 \cdot (-1)^{\Dec(\alpha)}\cdot b) \\
    &= 1 + \E_{c = \Enc(1)} \sum_{\alpha} \sum_{b \in \mathrm{spec(\scrB)}} \bra{\psi} (A^c_{\alpha})^\dagger \scrB_b A^c_{\alpha} \ket{\psi} b^2 \nonumber \\
        &\qquad - 2 \E_{c = \Enc(1)} \sum_{\alpha} \sum_{b \in \mathrm{spec(\scrB)}} \bra{\psi} (A^c_{\alpha})^\dagger \scrB_b A^c_{\alpha} \ket{\psi} (-1)^{\Dec(\alpha)} \cdot b \\
    &= 1 + \E_{c = \Enc(1)} \sum_{\alpha}  \bra{\psi} (A^c_{\alpha})^\dagger \scrB^2 A^c_{\alpha} \ket{\psi}  - \E_{c = \Enc(0)} 2\sum_{\alpha}  \bra{\psi} (A^c_{\alpha})^\dagger \scrB A^c_{\alpha} \ket{\psi} (-1)^{\Dec(\alpha)} \label{eq:eabsq-obs} \\
    &= 1 + \E_{c = \Enc(1)} \sum_\alpha \bra{\psi} (A^0_{\alpha})^{\dagger} \frac{(B^0 - B^1)^2}{2} A^0_{\alpha} \ket{\psi} \nonumber \\
    &\qquad - 2\E_{c = \Enc(1)} \sum_{\alpha} \bra{\psi} (A^c_{\alpha})^{\dagger} \frac{(B^0 - B^1)}{\sqrt2} A^c_{\alpha} \ket{\psi} (-1)^{\Dec(\alpha)} \label{eq:approx-delta-crypto-q2}
\end{align}
Now, applying \Cref{lem:bplusminus-ind} and \Cref{lem:bplusminus-squared-ind}, there exists a function $\delta_\crp(\lambda) = \sf{negl}(\lambda)$ such that
\begin{align}
    \text{\eqref{eq:approx-delta-crypto-q2}} &\approx_{\delta_\crp} 1 + \E_{b \in \{0,1\}} \E_{c' = \mathrm{Enc}(b)} \sum_{\alpha} \bra{\psi} (A^{c'}_{\alpha})^\dagger \frac{(B^0 - B^1)^2}{2} A^{c'}_{\alpha} \ket{\psi} \nonumber \\
    &\qquad - \E_{c = \Enc(1)} 2\sum_\alpha \bra{\psi} (A^c_{\alpha})^\dagger \frac{(B^0 - B^1)}{\sqrt2} A^c_{\alpha} \ket{\psi} (-1)^{\Dec(\alpha)} \label{eq:eabsq-crypto} \\
    &= 1 + 1 - \frac{1}{2} (\Gamma_{B^0 B^1} + \Gamma_{B^1 B^0} ) - \frac{2}{\sqrt{2}} (\Gamma_{A^1 B^0}  - \Gamma_{A^1 B^1} ) \\
    &= \Gamma_{A^0 A^0} + \frac{1}{2} (\Gamma_{B^0 B^0} + \Gamma_{B^1 B^1} - \Gamma_{B^0 B^1} - \Gamma_{B^1 B^0}) - \frac{1}{\sqrt{2}} (\Gamma_{A^1 B^0} + \Gamma_{B^0 A^1} - \Gamma_{A^1 B^1} - \Gamma_{B^1 A^1}) \\
    &= q_2^\dagger \Gamma q_2.
\end{align}
\end{proof}

\subsection{Rigidity and anticommutation}
\label{sec:chsh-anticom-robust}
\subsubsection{Zero-error argument for anticommutation}
\label{sec:chsh-anticom-zero}
In this section, for intuition's sake, we will present a zero-error version of the robust rigidity argument that we present in the proof of \Cref{lem:chsh-anticom}.

Suppose we have an optimal strategy for the game. This is a strategy for which 
\[ \Gamma_{A^0 B^0} + \Gamma_{A^0 B^1} + \Gamma_{A^1 B^0} - \Gamma_{A^1 B^1} = 2\sqrt{2}. \]
By Lemmas \ref{claim:q_1} and \ref{claim:q_2}, it therefore holds that $q_j^\dagger \Gamma q_j \approx_{\delta_\crp} 0$  for $j \in \{1,2\}$. For the sake of illustration, in this section only, we will pretend that $\delta_\crp = 0$. Then, expanding this out, we get
\begin{align}
    0 &= q_1^\dagger \Gamma q_1 = q_2^\dagger \Gamma q_2.
\end{align}

Let us focus on the condition that $q_1^\dagger \Gamma q_1 = 0$ and therefore $\E_{\mu_1}[(a-b)^2] = 0$. These conditions imply that, for any $c = \Enc(0)$, after we measure the state with $A^c_\alpha$ to obtain an outcome $a = \Dec(\alpha)$ and a post-measurement state $\ket{\psi^c_\alpha}$, then 
\begin{equation}
    \frac{(B^0 + B^1)}{\sqrt{2}} \ket{\psi^c_\alpha} = (-1)^{\Dec(\alpha)} \ket{\psi^c_\alpha}. \label{eq:b-stabilizer-exact} 
\end{equation}

Let us now calculate  $\scrB^2 \ket{\psi^c_\alpha}$.
\begin{align}
    \ket{\psi^c_\alpha} &= \scrB^2 \ket{\psi^c_\alpha} \\
    &= \frac{1}{2} (2 + \{B^0, B^1 \}) \ket{\psi^c_\alpha} \\
    0 &= \{ B^0, B^1 \} \ket{\psi^c_\alpha}.
\end{align}
Thus, we conclude that in the zero-error case, the anticommutator annihilates the post-measurement state for all possible questions $c$ that are encryptions of 0 and for all possible measurement outcomes $\alpha$. The post-measurement states must span the entire Hilbert space of possible Bob states (states that the prover can be in at the start of the second round of interaction). Therefore, the anticommutator is $0$ when restricted to this subspace. A similar argument could be applied to argue that the square of the anticommutator is 0 when restricted to this subspace.

% \paragraph{Expectation of the anticommutator squared}

% Now, let us imagine applying $(B_0 + B_1)/\sqrt{2}$ to the state four times in sequence. We obtain
% \begin{align}
%     \frac{1}{4} (B^0 + B^1)^4 \ket{\psi^0_\alpha} &= \frac{1}{4} (4 + \{B^0, B^1 \}^2 + 2\{B^0, B^1\}) \ket{\psi^0_\alpha} = \ket{\psi^0_\alpha}. \label{eq:b4th-exact}  \\
% \end{align}

% Let us multiply on the left by $\bra{\psi^0_a}$ and average over $a$. Then we obtain
% \begin{equation}
%     1 + \frac{1}{4} \E_{\alpha} \bra{\psi^0_\alpha} (\{B^0, B^1\}^2 + 2 \{B^0, B^1\}) \ket{\psi^0_\alpha} = 1.
% \end{equation}
% We know from the previous section\anote{where?} that $\E_{\alpha} \bra{\psi^0_\alpha} \{B^0, B^1\} \ket{\psi^0_\alpha} = 0$ \anote{actually negligible}. Therefore, we conclude that $\E_\alpha \bra{\psi^0_\alpha} \{B^0, B^1\}^2 \ket{\psi^0_\alpha} = 0$. In other words, $B^0$ and $B^1$ anticommute on the postmeasurement states after measuring $A^0$. An analogous argument, \emph{mutatis mutandis}, shows that we have anticommutation on the postmeasurement states after measuring $A^1$ as well.

\subsubsection{Anticommutation with finite error}
Now, suppose we have a strategy that succeeds in the game with probability $p_{game} = \omega^* - \eps$. We would like to show that an $\eps$-approximate version of the argument presented in \Cref{sec:chsh-anticom-zero} holds.
\begin{lemma}\label{lem:chsh-anticom}
For any strategy that succeeds in the compiled CHSH game (see \Cref{sec:chsh-compiled-game}) with probability $p_{game} = \omega^* - \eps$, it holds that 
\begin{equation}
    \E_{c = \Enc(0)} \sum_\alpha \bra{\psi} (A^c_\alpha)^\dagger \cdot  | \{B^0, B^1\}|^2 \cdot A^c_\alpha \ket{\psi} \leq \delta_{\acom}(\eps),
\end{equation}
where
\begin{equation}
    \delta_{\acom}(\eps) = 96\sqrt{2} \cdot \eps + 12 \delta_{\crp}.
\end{equation}
\end{lemma}
\begin{proof}
Suppose the prover succeeds in the compiled CHSH game with probability $\omega^* - \eps$; then it holds by \Cref{eq:win-prob-decomp} that
\begin{align}
    q_1^\dagger \Gamma q_1 + q_2^\dagger \Gamma q_2  &= 8\sqrt{2} (\omega^* - p_{game})\\
    &= 8\sqrt{2} \cdot \eps \\
    q_1^\dagger \Gamma q_1 &\leq 8 \sqrt{2} \eps \\
    q_2^\dagger \Gamma q_2 &\leq 8\sqrt{2} \eps.
\end{align}

Let us analyze each term separately. First, for $q_1$, we have
\begin{align}
    8 \sqrt{2} \cdot \eps &\geq q_1^\dagger \Gamma q_1 \\
    &\geq \E_{\mu} [(a-b)^2] - \delta_{c\crp} \\
    &= 1 + \E_{c = \Enc(0)} \sum_\alpha \bra{\psi} (A^c_\alpha)^\dagger \scrB^2 (A^c_\alpha)\ket{\psi} - 2\E_{c = \Enc(0)}  \sum_{\alpha} \bra{\psi} (A^c_\alpha)^\dagger \scrB A^c_\alpha \ket{\psi} (-1)^{\Dec(\alpha)} - \delta_{\crp}, \label{eq:q1squared-approx}
\end{align}
where the last line is by \Cref{eq:eabsq-obs} and \Cref{eq:eabsq-crypto}.

Now, we would like to use this to derive an approximate version of \Cref{eq:b-stabilizer-exact}. To do this, start by writing the squared difference of the left and right sides of \Cref{eq:b-stabilizer-exact}, averaged over $c$ and summed over $\alpha$:
\begin{align}
    &\E_{c = \Enc(0)} \sum_\alpha \| \scrB A^c_\alpha \ket{\psi} - (-1)^{\Dec(\alpha)} A^c_\alpha \ket{\psi} \|^2 \\
    &\quad= \E_{c = \Enc(0)} \sum_\alpha \bra{\psi} (A^c_\alpha)^\dagger \scrB^2 A^c_\alpha \ket{\psi} + \E_{c = \Enc(0)} \sum_\alpha \bra{\psi} (A^c_\alpha)^\dagger (A^c_\alpha) \ket{\psi} \nonumber \\
    &\quad\qquad - 2\E_{c = \Enc(0)}  \sum_\alpha (-1)^{\Dec(\alpha)} \bra{\psi} (A^c_\alpha)^\dagger \scrB A^c_\alpha \ket{\psi}\\
     &\quad= 1+ \E_{c = \Enc(0)} \sum_\alpha \bra{\psi} (A^c_\alpha)^\dagger \scrB^2 A^c_\alpha \ket{\psi}  - 2 \E_{c = \Enc(0)} \sum_\alpha (-1)^{\Dec(\alpha)} \bra{\psi} (A^c_\alpha)^\dagger \scrB A^c_\alpha \ket{\psi}
\end{align}
Hence, applying \Cref{eq:q1squared-approx}, we deduce that
\begin{equation}
\E_{c = \Enc(0)} \sum_{\alpha} \| \scrB A^c_\alpha \ket{\psi} - (-1)^{\Dec(\alpha)} A^c_\alpha \ket{\psi} \|^2 \leq 8\sqrt{2} \cdot \eps + \delta_\crp. \label{eq:b-stabilizer-approx}
\end{equation}

Now, as in the exact case, we will study what happens when $\scrB^2$ is applied to $A^c_\alpha \ket{\psi}$. First, recall that by the definition of $\scrB$, the following equalities hold exactly for any choice of $B^0, B^1$:
\begin{equation}
    \scrB^2 = \frac{1}{2} (B^0 + B^1)^2 = \frac{1}{2} (2I + \{B^0, B^1\}).
\end{equation}
So we have
\begin{align}
    \scrB^2 A^c_\alpha \ket{\psi} &= (I + \frac{1}{2} \{B^0, B^1\}) A^c_\alpha \ket{\psi} \\
    &= A^c_\alpha \ket{\psi} + \underbrace{(\scrB A^c_\alpha - A^c_\alpha) \ket{\psi}}_{\ket{\Delta^{c,(1)}_\alpha}} + \underbrace{\scrB (\scrB A^c_\alpha - A^c_\alpha) \ket{\psi}}_{\ket{\Delta^{c,(2)}_\alpha}}. \\
    \{B^0, B^1 \} A^c_\alpha \ket{\psi} &= 2 \ket{\Delta^{c,(1)}_\alpha} + 2 \ket{\Delta^{c,(2)}_\alpha}.
\end{align}
We would like bound the square norm of the RHS of the last line, averaged over $c = \Enc(0)$ and summed over $\alpha$:
\begin{align}
    \E_{c = \Enc(0)} \sum_\alpha \| \{B^0, B^1\} A^c_\alpha \ket{\psi} \|^2 &\leq 4 \E_{c = \Enc(0)}  \sum_{\alpha}\| \ket{\Delta_\alpha^{c,(1)}}\|^2 + 4 \E_{c = \Enc(0)} \sum_\alpha \|\ket{\Delta_\alpha^{c,(2)}}\|^2 \\
    &\leq 4\E_{c = \Enc(0)}  \sum_{\alpha}\| \ket{\Delta_\alpha^{c,(1)}}\|^2 + 4 \E_{c = \Enc(0)} \sum_\alpha \|\scrB\|^2 \cdot \|\ket{\Delta_\alpha^{c,(2)}}\|^2 \\
    &\leq 12 \E_{c = \Enc(0)} \sum_{\alpha}\| \ket{\Delta_\alpha^{c,(1)}}\|^2, \\
    &\leq 96 \sqrt{2} \cdot \eps + 12 \delta_{\crp} = \delta_{\acom}(\eps)
\end{align}
where in the penultimate line we have used that $\| \scrB \| \leq \sqrt{2}$, and in the last line we have applied \Cref{eq:b-stabilizer-approx}. 

By expanding out the squared norm on the LHS we obtain the conclusion of the lemma:
\begin{equation}
    \E_{c = \Enc(0)} \sum_\alpha \bra{\psi} (A^c_\alpha)^\dagger \cdot  | \{B^0, B^1\}|^2 \cdot A^c_\alpha \ket{\psi} \leq \delta_{\acom}(\eps).
\end{equation}

\end{proof}

%%% Local Variables:
%%% mode: latex
%%% TeX-master: "../protocol"
%%% End:

\section{A verification protocol}

\label{sec:protocol}

Let $\sf{QHE} = (\sf{Gen},\sf{Enc},\sf{Dec},\sf{Eval})$ be a quantum secret-key homomorphic encryption scheme, as described in \Cref{sec:qhe}. ($\sf{QHE}$ needs to be capable of homomorphically evaluating the circuit family $C_A$ which we describe in step 4 of the protocol below.) Let $H = \sum_{Wij} p_{Wij} W(e_i + e_j)$, a Hamiltonian operator on $n$ qubits, be the XX/ZZ local Hamiltonian of interest, where we assume $\sum_{Wij} p_{Wij} = 1$. We are trying to decide whether the minimum eigenvalue of $H$, which ranges from $-1$ to $1$, is smaller than $\alpha \in [-1,1]$ or larger than $\beta = \alpha + \frac{1}{\mathrm{poly}(n)}$. Suppose that the honest prover receives a witness state $\rho$ which is $n$ qubits long.
%We firstly define the \emph{principal subgame} in our protocol (which will later be repeated many times).
%\paragraph{Principal subgame}
\begin{enumerate}
	\item The verifier sets $\lambda = n$ and samples a secret key $sk \leftarrow \sf{Gen}(1^\lambda)$.
	\item The verifier samples a pair of questions $q_A, q_B$ as follows. For notational convenience, define $U_n$ to be the uniform distribution on $n$ bits, and define $D_X$ to be the (renormalised) distribution over $X(e_i + e_j)$ operators induced by $H$: formally, $D_X$ is a distribution over $n$-bit bitstrings, defined by $D_X(z) = \begin{cases} \frac{p_{Xij}}{\sum_{ij} p_{Xij}} & z = e_i + e_j \\ 0 & \text{else} \end{cases}$. Define $D_Z$ similarly to $D_X$. Define $D_Q$ to be the distribution $U_n \otimes D_X$.
	
	Let $\kappa(n) = \Theta((\beta-\alpha)^2)$ be a parameter. The verifier chooses one of the subtests below, the first two with probability $(1-\kappa)/2$ and the last one with probability $\kappa$:
	\begin{enumerate}
	    \item \textbf{CHSH.}
%	    Let $S_{\acom} \subseteq \{0,1\}^n \times \{0,1\}^n$ be the set of all pairs $(a,b)$ such that $b$ has Hamming weight $2$ and $a \cdot b = 1$.
	    The verifier samples a pair $(a,b)$ from the distribution $D_Q = U_n \otimes D_X$, and keeps sampling until $a \cdot b = 1$. (We will refer to the distribution that the verifier rejection-samples from in this case as $D_Q^1$.) The verifier also chooses uniformly random bits $x \in \{0,1\}$ and $y \in \{0,1\}$. It sets $q_A = (\textsf{CHSH}, (a,b,x))$ and $q_B = y$.
	    \item \textbf{Commutation.} The verifier samples a pair $(a,b)$ from the distribution $D_Q = U_n \otimes D_X$, and keeps sampling until $a \cdot b = 0$. (We will refer to the distribution that the verifier rejection-samples from in this case as $D_Q^0$.) The verifier also chooses a uniformly random bit $y \in \{0,1\}$. It sets $q_A = (\textsf{Commutation}, (a,b))$ and $q_B = y$.
	    \item \textbf{Teleport.} The verifier sets $q_A = \textsf{Teleport}$. It samples uniformly random $y \in \{0,1\}$, and sets $q_B = y$.
%	     such that $y=0$ with probability $\sum_{ij} p_{Zij}$ and $y=1$ with probability $\sum_{ij} p_{Xij}$
	\end{enumerate}
	The verifier pads $q_A$ so that all Alice questions have the same bit length.
	
% 	\item The verifier samples an \emph{Alice question} $q_A$ for the prover. With probability $\frac{1}{4}$ each, the verifier chooses one of the following options. (We assume that the verifier pads the shorter questions so that each option always results in an Alice question $q_A$ of the same length.)
% 	\begin{enumerate}
% 	\item \textbf{CHSH.} The verifier chooses a uniformly random pair $a,b \in \{0,1\}^n$ such that $a\cdot b = 1$ and a uniformly random bit $x \in \{0,1\}$, and sets $q_A = (\sf{CHSH},(a,b,x))$.
% 	\item \textbf{Commutation.} The verifier chooses a uniformly random (ordered) pair $(a,b) \in \{0,1\}^n$ such that $a \cdot b = 0$, and sets $q_A = (\sf{Commutation}, (a,b))$. \anote{First observable is $Z$, second is $X$, by convention. This is why the pair is ordered.}
% 	\item \textbf{All-$X$ / all-$Z$.} The verifier chooses a bit $x \in \{0,1\}$ uniformly at random and sets $q_A = (\textsf{Single-basis}, x)$.
% 	\item \textbf{Teleport.} The verifier sets $q_A = \sf{Teleport}$.
% 	\end{enumerate}
    \item
	The verifier encrypts $q_A$ under $sk$ and sends $c = \sf{Enc}_{sk}(q_A)$ to the prover.
	\item The (honest) prover creates $n$ EPR pairs, and designates one half of each pair as an `Alice qubit' and the other half as a `Bob qubit' (so that there are $n$ Alice and $n$ Bob qubits in total). Then, using $c = \sf{Enc}_{sk}(q_A)$, it homomorphically evaluates a circuit $C_A$ with the following description. $C_A$ acts on $q_A$ as well as the concatenation of the `Alice qubits' and the witness $\rho$, and responds to each question type in the following way:
	\begin{enumerate}
	\item \textsf{CHSH}. Measure the prescribed CHSH Alice observable $A^{a,b,x} = (\sigma_Z(a) +(-1)^x \sigma_X(b))/\sqrt{2}$ on the `Alice qubits'; do nothing to $\rho$.
	\item \textsf{Commutation}. Measure $\sigma_Z(a)$ and $\sigma_X(b)$ on the `Alice qubits'; do nothing to $\rho$.
	\item \textsf{Teleport}. Teleport $\rho$ into the `Bob qubits' by doing a teleportation circuit on the `Alice qubits' and $\rho$, and measure the $X/Z$ corrections that arise from the Bell basis measurements as a $2n$-bit string.
	\end{enumerate}
	The prover reports the (encrypted) measurement outcome that results from homomorphically evaluating $C_A$. We will refer to the encrypted measurement outcome which the prover reports at this stage as the `Alice answer' $\alpha$.
	\item 
% 	After it receives the prover's response from step 3, the verifier generates a \emph{Bob question} $q_B$ for the prover as follows:
% 	\begin{enumerate}
% 	\item Sample a uniformly random bit $y \in \{0,1\}$.
% 	\item \textbf{Single observable.} With probability 1/2, sample a uniform $c \in \{0,1\}^n$ and set $q_B = (\textsf{Single-observable}, y, c)$.
% 	\item \textbf{All-$X$/All-$Z$} With probability 1/2, set $q_B = (\textsf{Single-basis}, y)$.
% 	\end{enumerate}
	The verifier sends $q_B$ to the prover in the clear. The prover measures all the `Bob qubits' in the basis $W$ indicated by $q_B$, where $W = X$ if $q_B = 0$ and $W = Z$ if $q_B = 1$, and obtains an $n$-bit string $s_B$. It reports the answer $s_B$. 
	
% 	The prover measures $Z(y) X(\overline y)$ on the unencrypted half EPR pairs. In other words, for all $i \in [n]$, the prover measures $B^{y_i, i}$, where the notation $B^{y_i, i}$ denotes the prescribed CHSH Bob observable $B^{y_i}$ acting on the $i$th half EPR pair. The prover reports its measurement outcomes as an unencrypted $n$-bit string $s_B$.
    \item The verifier decrypts the Alice answer $\alpha$ to obtain a string $s_A$, and then accepts or rejects according to the subtest.

    \begin{enumerate}
	    \item \textbf{CHSH.} In this case, recall that the Alice question was $(\textsf{CHSH},(a,b,x))$, and the Bob question was $y \in \{0,1\}$. $s_A$ in this case is a single bit, and $s_B$ is an $n$-bit string. The verifier computes $z := (1 - y)(a \cdot s_B) + y (b \cdot s_B)$, and accepts iff $s_A + z = x \cdot y$.
	    \item \textbf{Commutation.} 
	    In this case, recall that the Alice question was $(\textsf{Commutation}, (a,b))$, and the Bob question was $y \in \{0,1\}$. Recall also that $s_A \in \{0,1\}^2$ and $s_B \in \{0,1\}^n$. The verifier computes $z := (1 - y)(a \cdot s_B) + y (b \cdot s_B)$, and accepts iff $(s_A)_y = z$ (i.e. if $y = 0$, it checks that $z$ is equal to the first bit of $s_A$, and otherwise it checks that it is equal to the second bit of $s_A$).
	    
	    \item \textbf{Teleport.} The verifier samples a $w$ such that $w=0$ with probability $\sum_{ij} p_{Xij}$ and $w=1$ with probability $\sum_{ij} p_{Zij}$. If $w \neq p_B$, the verifier automatically accepts. If $w = p_B$, then the verifier samples a term $W(e_i + e_j)$ from the distribution induced by $p_{Wij}$, where $W = X$ if $w=0$ and $W =Z$ if $w=1$. We assume that the $2n$-bit string $s_A$ is in the form $s_A = \underbrace{z}_{n \text{ bits}} \| \underbrace{x}_{n \text{ bits}}$. (Here $z$ represents the $Z$-gate corrections that the verifier is supposed to apply, and $x$ represents the $X$-gate corrections. Note that the $Z$ gate corrections only affect the outcome if $W=X$, and vice versa.)
	    \begin{enumerate}
	    \item If $W = Z$, the verifier computes $(-1)^{(s_B)_i + (s_B)_j + (s_A)_i + (s_A)_j}$ and accepts iff the result is $-1$.
	    \item If $W = X$, the verifier computes $(-1)^{(s_B)_i + (s_B)_j + (s_A)_{n+i} + (s_A)_{n+j}}$ and accepts iff the result is $-1$.
	    \end{enumerate}
	\end{enumerate}

% 	\item OLD: The verifier decrypts the prover's response from step 3 to yield a string $s_A$, and performs the following checks, depending on the Alice question $q_A$.
% 	\begin{enumerate}
% 	\item \textbf{CHSH.} In this case, the Alice question was $(i,x)$ for $i \in [n], x \in \{0,1\}$, and the Alice answer $s_A$ was a single bit. The verifier checks whether $x \cdot y_i = s_A \oplus (s_B)_i$. It accepts if and only if this check passes.
% 	\item \textbf{Commutation.} In this case, the Alice question was an ordered pair of indices $(i,j)$, and the Alice answer $s_A$ was two bits. If $y_i = 0, y_j = 1$, the verifier checks that $(s_A)_1 = (s_B)_i$ and $(s_A)_2 = (s_B)_j$. Otherwise, it accepts automatically.
% 	\item \textbf{All-$X$ / all-$Z$.} In this case, the Alice question was a single bit $x$, and the Alice answer $s_A$ was an $n$-bit string. The verifier picks an index $i \in [n]$ uniformly at random. If $x=0$ and $y_i = 0$, or if $x=1$ and $y_i = 1$, the verifier checks that $(s_A)_i = (s_B)_i$. Otherwise, it automatically accepts.
% 	\item \textbf{Teleport.} The verifier chooses a Hamiltonian term $H_k$ uniformly at random. $H_k$ is either of the form 1. $-Z(e_i) Z(e_j)$ or of the form 2. $-X(e_i) X(e_j)$ for $i, j \in [n]$, $i\neq j$. In case 1, if $y_i = y_j = 1$, the verifier checks that $(s_B)_i (s_B)_j = 1$; in case 2, if $y_i = y_j = 0$, the verifier checks that $(s_B)_i (s_B)_j = 1$. Otherwise, it accepts automatically. \anote{Add verifier corrections for teleport.}
% 	\end{enumerate}
\end{enumerate}

\section{Completeness}

See \Cref{thm:main} for an analysis of the completeness-soundness gap that the protocol of \Cref{sec:protocol} achieves.

%%% Local Variables:
%%% mode: latex
%%% TeX-master: "../protocol"
%%% End:

\section{Soundness}

\label{sec:soundness}

\subsection{Modeling}
\label{sec:verification-modelling}

We model the prover in the protocol of \Cref{sec:protocol} as follows, largely following the notation in \Cref{sec:general-modeling}.

\begin{enumerate}
    \item The prover starts with a pure state $\ket{\psi}$. \textbf{Notational note:} in this section, for notational convenience, we may use the notational shorthand $\langle O \rangle$ in order to represent the expectation value of operator $O$ with respect to $\ket{\psi}$, i.e. $\langle O \rangle := \bra{\psi} O \ket{\psi}$.
    \item Upon receipt of an `Alice' question ciphertext $c$ (step 3), the prover applies a measurement specified by a collection of matrices $\{A^c_{\alpha}\}_{\alpha}$. (See \Cref{sec:general-modeling} for more details about how $\{A^c_{\alpha}\}_{\alpha}$ is defined.)  The prover replies to the verifier (step 4) with the measurement outcome $\alpha$.
    \item Upon receipt of a plaintext `Bob' question (step 5), the prover applies one of two projective measurements: $\{Z_\gamma\}_{\gamma \in \{0,1\}^n}$ or $\{X_\gamma\}_{\gamma \in \{0,1\}^n}$, depending on whether it receives question 0 or 1 respectively. The prover replies to the verifier (step 6) with the string $\gamma$. We assume wlog that both $Z$ and $X$ consist of a unitary followed by a (potentially partial) standard basis measurement, that is: $\{Z_\gamma\}_{\gamma \in \{0,1\}^n} = \{U_Z^\dagger (\ket{\gamma_1}\bra{\gamma_1} \otimes \cdots \otimes \ket{\gamma_n}\bra{\gamma_n} \otimes I) U_Z\}_{\gamma \in \{0,1\}^n}$ and $\{X_\gamma\}_{\gamma \in \{0,1\}^n} = \{U_X^\dagger (\ket{\gamma_1}\bra{\gamma_1} \otimes \cdots \otimes \ket{\gamma_n}\bra{\gamma_n} \otimes I) U_X\}_{\gamma \in \{0,1\}^n}$. The $I$ part simply represents the part of the system that the prover does not measure, and because it is not important, the dimensions will be left unspecified.
\end{enumerate}

\textbf{Notational note:} For notational purposes, we will define a set of binary observables $\{Z(a)\}_{a \in \{0,1\}^n}$ and $\{X(b)\}_{b \in \{0,1\}^n}$ from the prover's $Z$ and $X$ measurements (defined immediately above) as follows:
\begin{gather}
Z(a) = \sum_{\gamma \in \{0,1\}^n} (-1)^{a \cdot \gamma} U_Z^\dagger (\ket{\gamma_1}\bra{\gamma_1} \otimes \cdots \otimes \ket{\gamma_n}\bra{\gamma_n} \otimes I) U_Z \\
X(b) = \sum_{\gamma \in \{0,1\}^n} (-1)^{b \cdot \gamma} U_X^\dagger (\ket{\gamma_1}\bra{\gamma_1} \otimes \cdots \otimes \ket{\gamma_n}\bra{\gamma_n} \otimes I) U_X
\end{gather}

\begin{lemma}\label{lem:bob-exact-linearity}
For $W \in \{X,Z\}$, and for all $a, a' \in \{0,1\}^n$, $W(a)W(a') = W(a + a')$.
\end{lemma}
\begin{proof}
By definition.
\end{proof}

\subsection{Subtests}

\begin{lemma}\label{lem:bob-anticommutation}
Suppose the prover's strategy succeeds in the CHSH subtest with probability at least $\omega^*_{\CHSH} - \eps$. Then,
\begin{equation}
    \E_{(a,b) \leftarrow D_Q^1}  \E_{c = \Enc((\CHSH, (a,b,0))} \sum_\alpha \langle (A^{c}_{\alpha})^\dagger \cdot |\{Z(a), X(b)\}|^2 \cdot (A^{c}_{\alpha}) \rangle \leq \delta_{\acom}(\eps).
\end{equation}
\end{lemma}
\begin{proof}
Fix a pair $a,b$, and let $\omega^*_{\CHSH} - \eps_{a,b}$ be the probability of success of the prover's strategy conditioned on this choice of $a,b$. By the definition of conditional probability, it holds that $\E_{(a,b) \leftarrow D_Q^1} \eps_{a,b} = \eps$. By the analysis of the computational CHSH game (\Cref{lem:chsh-anticom}), we have that
\[ \E_{c = \Enc((\CHSH, (a,b,0))} \sum_\alpha \langle (A^{c}_{\alpha})^\dagger \cdot |\{Z(a), X(b)\}|^2 \cdot (A^{c}_{\alpha}) \rangle \leq \delta_{\acom}(\eps_{a,b} ).\]
This is because for a fixed $a,b$, the CHSH subtest reduces to the CHSH game with $Z(a)$ playing the role of Bob's observable $B^0$ and $X(b)$ playing the role of his observable $B^1$.

 Thus, it holds that
 \begin{equation}
    \E_{(a,b) \leftarrow D_Q^1}  \E_{c = \Enc((\CHSH, (a,b,0))} \sum_\alpha \langle (A^{c}_{\alpha})^\dagger \cdot |\{Z(a), X(b)\}|^2 \cdot (A^{c}_{\alpha}) \rangle \leq \E_{(a,b) \leftarrow D_Q^1} \delta_{\acom}(\eps_{a,b}).
\end{equation}
To complete the proof, we recall that 
$\delta_{\acom}(\eps) = 96 \sqrt{2} \cdot \eps + 12 \delta_{\crp}$
is linear in $\eps$ and therefore $\E_{a,b} \delta_{\acom} (\eps_{a,b}) = \delta_{\acom}(\E_{a,b} \eps_{a,b}) = \delta_{\acom}(\eps)$. 
\end{proof}

\begin{lemma}\label{lem:bob-commutation}
Suppose the prover's strategy succeeds in the commutation subtest with probability at least $1 - \eps$. Then,
\begin{equation}
    \E_{(a,b) \leftarrow D_Q^0}  \E_{c = \Enc((\mathsf{Commutation}, (a,b))} \sum_\alpha \langle (A^{c}_{\alpha})^\dagger \cdot |[Z(a), X(b)]|^2 \cdot (A^{c}_{\alpha}) \rangle \leq \delta_{\com}(\eps).
\end{equation}
\end{lemma}
\begin{proof}
For any fixed $a,b$, let the probability of success in this subtest conditioned on $a,b$ be $1 - \eps_{a,b}$. It holds that $\E_{(a,b) \in S_{\com}} \eps_{a,b} = \eps$. By the analysis of the commutation game, it holds that 
\begin{equation}
    \E_{c = \Enc((\mathsf{Commutation}, (a,b))} \sum_\alpha \langle (A^{c}_{\alpha})^\dagger \cdot |[Z(a), X(b)]|^2 \cdot (A^{c}_{\alpha}) \rangle \leq \delta_{\com}(\eps_{a,b}).
\end{equation}
Now, averaging both sides over $(a,b) \in S_{\com}$ and observing that $\delta_{\com}(\eps)$ is linear in $\eps$, we obtain the conclusion of the lemma.
\end{proof}

\begin{lemma}\label{lem:bob-phase}
  Suppose the prover's strategy succeeds in the CHSH subtest with probability at least $\omega^*_{\CHSH} - \eps$, and in the commutation subtest with probability $1 - \eps$. Then,
  \begin{equation}
    \E_{\substack{(a,b) \leftarrow D_Q \\ c = \Enc(\tlp)}} \sum_\alpha \langle (A^{c}_{\alpha})^\dagger \cdot |(-1)^{a \cdot b} Z(a) X(b) - X(b) Z(a) |^2 \cdot (A^{c}_{\alpha}) \rangle \leq \delta_{\phase}(\eps),
\end{equation}
where
\begin{equation} 
\delta_{\phase}(\eps) = \frac{1}{2}(\delta_\com + \delta_\acom) + \delta_\crp(\lambda). 
\end{equation}
\end{lemma}
\begin{proof}
By definition,
\begin{align}
    &\E_{\substack{(a,b) \leftarrow D_Q \\ c = \Enc(\tlp)}} \sum_\alpha \langle (A^{c}_{\alpha})^\dagger \cdot |(-1)^{a \cdot b} Z(a) X(b) - X(b) Z(a) |^2 \cdot (A^{c}_{\alpha}) \rangle \label{eq:both-parts} \\
    = \quad &\frac{1}{2} \E_{\substack{(a,b) \leftarrow D^0_Q \\ c = \Enc(\tlp)}} \sum_\alpha \langle (A^{c}_{\alpha})^\dagger \cdot |(-1)^{a \cdot b} Z(a) X(b) - X(b) Z(a) |^2 \cdot (A^{c}_{\alpha}) \rangle  \label{eq:com-part} \\
    +&\frac{1}{2} \E_{\substack{(a,b) \leftarrow D^1_Q \\ c = \Enc(\tlp)}} \sum_\alpha \langle (A^{c}_{\alpha})^\dagger \cdot |(-1)^{a \cdot b} Z(a) X(b) - X(b) Z(a) |^2 \cdot (A^{c}_{\alpha}) \rangle \label{eq:anticom-part}
\end{align}

Applying \Cref{lem:commutator-ind} and \Cref{lem:bob-commutation} to \Cref{eq:com-part}, we get that
\begin{align}
&\frac{1}{2} \E_{\substack{(a,b) \leftarrow D^0_Q \\ c = \Enc(\tlp)}} \sum_\alpha \langle (A^{c}_{\alpha})^\dagger \cdot |(-1)^{a \cdot b} Z(a) X(b) - X(b) Z(a) |^2 \cdot (A^{c}_{\alpha}) \rangle \\
&= \frac{1}{2} \E_{\substack{(a,b) \leftarrow D^0_Q \\ c = \Enc(\tlp)}} \sum_\alpha \langle (A^{c}_{\alpha})^\dagger \cdot |[ Z(a), X(b) ] |^2 \cdot (A^{c}_{\alpha}) \rangle \\
&\leq \frac{1}{2}(\delta_{\com} + \delta_\crp(\lambda)).
\end{align}

Applying \Cref{lem:commutator-ind} and \Cref{lem:bob-anticommutation} to \Cref{eq:anticom-part}, we get that
\begin{align}
&\frac{1}{2} \E_{\substack{(a,b) \leftarrow D^1_Q \\ c = \Enc(\tlp)}} \sum_\alpha \langle (A^{c}_{\alpha})^\dagger \cdot |(-1)^{a \cdot b} Z(a) X(b) - X(b) Z(a) |^2 \cdot (A^{c}_{\alpha}) \rangle \\
&= \frac{1}{2} \E_{\substack{(a,b) \leftarrow D^1_Q \\ c = \Enc(\tlp)}} \sum_\alpha \langle (A^{c}_{\alpha})^\dagger \cdot | \{ Z(a), X(b) \} |^2 \cdot (A^{c}_{\alpha}) \rangle \\
&\leq \frac{1}{2}(\delta_{\acom} + \delta_\crp(\lambda)).
\end{align}

Therefore, expanding \Cref{eq:both-parts},
\begin{align}
    &\E_{\substack{(a,b) \leftarrow D_Q \\ c = \Enc(\tlp)}} \sum_\alpha \langle (A^{c}_{\alpha})^\dagger \cdot |(-1)^{a \cdot b} Z(a) X(b) - X(b) Z(a) |^2 \cdot (A^{c}_{\alpha}) \rangle \\
    \leq &\quad \frac{1}{2}(\delta_\com + \delta_\acom) + \delta_\crp(\lambda).
\end{align}

\end{proof}

% \anote{To make this argument, need to write a lemma like \Cref{lem:bplusminus-ind} that argues that the anticommutator squared is measurable.}
%By the crypto indistinguishability, change the expectation over $c$ in \Cref{lem:bob-anticommutation} and \Cref{lem:bob-commutation}, incurring an error of $O(\delta_{\crp})$. The conclusion then follows directly. 
%\anote{TODO: get the precise error bound}. \znote{watch out that we need to combine the question distributions $D_Q^0$ and $D_Q^1$}
%\end{proof}

\subsection{The isometry}

\begin{lemma}\label{lem:zxz-close-to-x}
For any $u_1, u_2 \in \{0,1\}$, it holds that
\begin{equation}  \E_{\substack{(a,b) \leftarrow D_Q \\ c = \Enc(\tlp)}} \sum_{\substack{\alpha \::\: \\ \Dec(\alpha)_i = u_1, \\ \Dec(\alpha)_j = u_2}} \langle (A^{c}_\alpha)^\dagger \cdot |\: (-1)^{a \cdot b} Z(a) X(b) Z(a) - X(b) \:| \cdot A^{c}_\alpha \rangle \leq \delta_{\tlp}(\eps), \label{eq:soundness-1} \end{equation}
where $\delta_{\tlp}(\eps) = \delta_{\phase}(\eps)^{1/2}$. 
\end{lemma}

\begin{proof}
Essentially, we would like to prove \Cref{eq:soundness-1} by commuting the $X(b)$ in the first term to the right past the $Z(a)$. 
\begin{align}
    \delta &= \E_{\substack{(a,b) \leftarrow D_Q \\ c = \Enc(\tlp)}} \sum_{\substack{\alpha \::\: \\ \Dec(\alpha)_i = u_1, \\ \Dec(\alpha)_j = u_2}} \langle (A^{c})^\dagger_\alpha (\: (-1)^{a \cdot b} Z(a) X(b) Z(a) - X(b) \:) A^{c}_\alpha \rangle \\
    &= \E_{\substack{(a,b) \leftarrow D_Q \\ c = \Enc(\tlp)}} \sum_{\substack{\alpha \::\: \\ \Dec(\alpha)_i = u_1, \\ \Dec(\alpha)_j = u_2}}  \langle (A^{c}_\alpha)^\dagger [(-1)^{a \cdot b} Z(a)(X(b) Z(a) - (-1)^{a \cdot z} Z(a) X(b)) ] A^{c}_\alpha \rangle \\
    &\leq  \sqrt{\E_{\substack{(a,b) \leftarrow D_Q \\ c = \Enc(\tlp)}} \sum_{\substack{\alpha \::\: \\ \Dec(\alpha)_i = u_1, \\ \Dec(\alpha)_j = u_2}}  \| (-1)^{a \cdot z} Z(a) A^{c}_\alpha \|_\psi^2}  \nonumber \\
    &\qquad \cdot \sqrt{\E_{\substack{(a,b) \leftarrow D_Q \\ c = \Enc(\tlp)}} \sum_{\substack{\alpha \::\: \\ \Dec(\alpha)_i = u_1, \\ \Dec(\alpha)_j = u_2}}  \|(X(b)Z(a) - (-1)^{a\cdot b} Z(a) X(b)) A^{c}_\alpha \|_\psi^2 } \\
    &\leq 1 \cdot \sqrt{\E_{\substack{(a,b) \leftarrow D_Q \\ c = \Enc(\tlp)}} \sum_{\substack{\alpha \::\: \\ \Dec(\alpha)_i = u_1, \\ \Dec(\alpha)_j = u_2}}  \|(X(b)Z(a) - (-1)^{a\cdot b} Z(a) X(b)) A^{c}_\alpha \|_\psi^2 } \\
    &\leq \sqrt{\E_{\substack{(a,b) \leftarrow D_Q \\ c = \Enc(\tlp)}} 
    \sum_{\substack{\alpha \::\: \\ \Dec(\alpha)_i = u_1, \\ \Dec(\alpha)_j = u_2}}  \langle (A^{c}_\alpha)^\dagger (Z(a) X(b) - (-1)^{a \cdot z} X(b) Z(a)) (X(b) Z(a) - (-1)^{a \cdot b} Z(a) X(b) ) A^{c}_\alpha \rangle } \\
    &\leq \sqrt{ \E_{\substack{(a,b) \leftarrow D_Q \\ c = \Enc(\tlp)}} \sum_{\substack{\alpha \::\: \\ \Dec(\alpha)_i = u_1, \\ \Dec(\alpha)_j = u_2}}  (-1)^{a \cdot z} \langle (A^{c}_\alpha)^\dagger (X(b)Z(a) - (-1)^{a \cdot b} Z(a) X(b))^2 A^{c}_\alpha \rangle} \label{eq:anticom-y-constrained}
\end{align}
It thus suffices to prove the statement
\begin{equation} \E_{\substack{(a,b) \leftarrow D_Q \\ c = \Enc(\tlp)}}  \sum_{\alpha}  \langle (A^{c}_\alpha)^\dagger \cdot |\: (-1)^{a \cdot b} X(b) Z(a) - Z(a) X(b) \:|^2  \cdot A^{c}_\alpha \rangle \leq \delta_{\tlp}(\eps)^2, \label{eq:soundness-2} \end{equation}
where we sum over \emph{all} values of $y$; this is because the summand in \Cref{eq:soundness-2} is nonnegative for all $y$, and so the sum is an upper bound for the term inside the square root in \Cref{eq:anticom-y-constrained}.

Now, to conclude the proof, observe that \Cref{eq:soundness-2} is precisely the conclusion of \Cref{lem:bob-phase}, for our choice of $\delta_{\tlp}$. 

% To prove \Cref{eq:soundness-2}, we firstly specialise to the subgame where the string $a$ which Alice gets in the \textsf{CHSH} and \textsf{Commutation} subtests is always such that $a_i = v_1, a_j = v_2$. (The prover's success probability on this subgame is still $1 - 1/\mathrm{poly}$ if it was $1 - 1/\mathrm{poly}$ overall, since this subgame occurs with constant probability.) Define $p_{a,z}$ to be the probability that:
% \begin{enumerate}
% 	\item the prover passes the CHSH subtest when Alice receives $(\sf{CHSH}, a, z, x)$ for uniformly random $x$ and Bob receives $y$ for uniformly random $y$, if $\langle a, z \rangle = 1$;
% 	\item the prover passes the commutation subtest when Alice receives $(\sf{Commutation}, a, z)$ as a question and Bob receives uniformly random $y$, if $\langle a, z \rangle  = 0$.
% \end{enumerate}
% We then note that, by a Markov bound, a $1 - 1/\poly$ fraction of $(a,z)$ are such that $p_{a,z} \geq 1 - 1/\poly$. We can thus apply our SoS reasoning individually to the different $(a,z)$ pairs, and recover \Cref{eq:soundness-2}. r

\end{proof}

\begin{definition}
  Let $\calH_Q$ and $\calH_A$ be two copies of $(\mathbb{C}^{2})^{\ot n}$. The $n$-qubit SWAP isometry $V: \mathcal{H}_{prover} \to \mathcal{H}_{prover} \ot \mathcal{H}_Q \ot \mathcal{H}_A$ is defined by the following expression:
  \begin{equation}
    V \ket{\phi} =\left(\frac{1}{2^n} \sum_{u,v \in \{0,1\}^n}  Z(u) X(v) \ot \Id \ot \sigma_Z(u) \sigma_X(v)   \right) \ket{\phi} \ket{\phi^+}^{\ot n}.
  \end{equation}
\end{definition}

\begin{claim}
\label{claim:isometry-on-paulis}
Let $\ket{\phi} \in \cal{H}_{prover}$ and let $\rho = \tr_{prover, A}[V \ket{\phi} \bra{\phi} V^\dagger]$. Then for any $a,b \in \{0,1\}^n$ it holds that
\begin{align}
    \tr[\sigma_Z(a) \rho] &= \E_{u \leftarrow U_n} \bra{\phi} Z(u) Z(u + a)  \ket{\phi} \label{eq:isometry-z} \\
    \tr[\sigma_X(b) \rho] &= \E_{u,v \leftarrow U_n} (-1)^{u \cdot b} \bra{\phi} Z(u) X(v + b) X(v) Z(u) \ket{\phi} \label{eq:isometry-x}.
\end{align}
\end{claim}
\begin{proof}
We prove these each by direct calculation. For the first, write
\begin{align}
    \tr[\sigma_Z(z) \rho] &= \bra{\phi}V^\dagger (\Id \ot \sigma_Z(z) \ot \Id) V \ket{\phi} \\
    &= \frac{1}{4^n} \sum_{u,v,s,t} \bra{\phi} \bra{\phi^+}^{\otimes n} Z(s) X(t)X(v) Z(u)  \ot \sigma_Z(a) \ot \sigma_Z(s) \sigma_X(v+t) \sigma_Z(u) \ket{\phi} \ket{\phi^+}^{\ot n} \\
    &= \frac{1}{4^n} \sum_{u,s: u+s = a} \sum_{u} \bra{\phi} \bra{\phi^+}^{\otimes n} Z(s) {\color{red} \cancel{X(t) X(t)}} Z(u)  \ot \sigma_Z(a)  \ot \sigma_Z(a) \ket{\phi} \ket{\phi^+}^{\otimes n} \\
    &= \frac{1}{2^n} \sum_{u,s: u+s = a} \bra{\phi} Z(s) Z(u) \ket{\phi} \\
    &= \E_{u \leftarrow U_n} \bra{\phi} Z(u) Z(u + a) \ket{\phi}.
\end{align}
For the second, write
\begin{align}
    \tr[\sigma_X(b) \rho] &= \bra{\phi}V^\dagger (\Id \ot \sigma_X(b) \ot \Id) V \ket{\phi} \\
    &= \frac{1}{4^n} \sum_{u,v,s,t} \bra{\phi} \bra{\phi^+}^{\otimes n} Z(s) X(t)X(v) Z(u)  \ot \sigma_X(b) \ot \sigma_Z(s) \sigma_X(v+t) \sigma_Z(u) \ket{\phi} \ket{\phi^+}^{\ot n} \\
    &= \frac{1}{4^n} \sum_{v,t \::\: v+t = b} \sum_{u} \bra{\phi} \bra{\phi^+}^{\otimes n} Z(u)X(t) X(v) Z(u) \otimes \sigma_X(b) \otimes (-1)^{u \cdot z} \sigma_X(b) \ket{\phi} \ket{\phi^+}^{\otimes n} \\
    &= \frac{1}{4^n} \sum_{u,v} (-1)^{u \cdot b} \bra{\phi} Z(u) X(v + b) X(v) Z(u) \ket{\phi} \\
    &= \E_{u,v \leftarrow U_n} (-1)^{u \cdot b} \bra{\phi} Z(u) X(v + b) X(v) Z(u) \ket{\phi}.
\end{align} 
\end{proof}

\begin{claim}
  Let $H_X$ be $H$ restricted to the $XX$ terms. Let $\hat{\E}[H_X]$ be the expected value of the measurement outcome computed by the verifier in a teleport round, conditioned on 1) $w=p_B$, so that the verifier performs an energy check instead of accepting automatically, and 2) the verifier choosing an $XX$ term to check. Then
  \[ \hat{\E}[H_X] =  \sum_{u_1, u_2} (-1)^{u_1 + u_2} \E_{\substack{(b = e_i + e_j) \leftarrow D_X \\ c = \Enc(\tlp)}} \sum_{\substack{\alpha \::\: \\ \Dec(\alpha)_i = u_1, \\ \Dec(\alpha)_j = u_2}} \langle (A^c_\alpha)^\dagger X(b) A^c_\alpha \rangle. \]
  Similarly,
  \[ \hat{\E}[H_Z] =  \sum_{v_1, v_2} (-1)^{v_1 + v_2} \E_{\substack{(a = e_i + e_j) \leftarrow D_Z \\ c = \Enc(\tlp)}} \sum_{\substack{\alpha \::\: \\ \Dec(\alpha)_{n+i} = v_1, \\ \Dec(\alpha)_{n+j} = v_2}} \langle (A^c_\alpha)^\dagger Z(a) A^c_\alpha \rangle. \]
  \label{claim:energy-measurement-exact}
\end{claim}
\begin{proof}
By inspection of the verifier's and the prover's actions in the protocol of \Cref{sec:protocol}.
\end{proof}

\begin{lemma}
  Define $\rho_\alpha := \E_{c = \Enc(\tlp)} \tr_{prover, A}[ V (A^c_\alpha) \ket{\psi}\bra{\psi} (A^c_\alpha)^\dagger V^\dagger ]$. Then, assuming that the prover passes with probability $\omega^*_{\CHSH} - \eps$ in the CHSH subtest and with probability $1 - \eps$ in the commutation subtest,
  \begin{equation} \sum_{u_1, u_2} (-1)^{u_1 + u_2}  \sum_{\substack{\alpha \::\: \\ \Dec(\alpha)_i = u_1, \\ \Dec(\alpha)_j = u_2}} \E_{b \leftarrow D_X} \tr[ \sigma_X(b) \rho_\alpha] \approx_{4\delta_\tlp(\eps)} \hat{\E}[H_X] \label{eq:energy-test-to-true-energy}   \end{equation}
\end{lemma}

\begin{proof}
First, by \Cref{eq:isometry-x}, and by the exact linearity (\Cref{lem:bob-exact-linearity}) of the Bob operators, the LHS of \Cref{eq:energy-test-to-true-energy} is equal to
\begin{align}
\sum_{u_1, u_2} (-1)^{u_1 + u_2} \sum_{\substack{\alpha \::\: \\ \Dec(\alpha)_i = u_1, \\ \Dec(\alpha)_j = u_2}} \E_{\substack{b \leftarrow D_X \\ c = \Enc(\tlp)}} \E_{a \leftarrow U_n} (-1)^{a\cdot b} \bra{\psi} (A^c_\alpha)^\dagger  Z(a) X(b) Z(a) (A^c_\alpha) \ket{\psi}. 
\end{align} 
\znote{gosh the $a$ above looks so similar to the $\alpha$}
By definition, $a \leftarrow U_n, b \leftarrow D_X$ is equivalent to $(a,b) \leftarrow D_Q$. Applying \Cref{lem:zxz-close-to-x} to the RHS, we get
\begin{align}
&\sum_{u_1, u_2} (-1)^{u_1 + u_2} \sum_{\substack{\alpha \::\: \\ \Dec(\alpha)_i = u_1, \\ \Dec(\alpha)_j = u_2}} \E_{c = \Enc(\tlp)} \E_{(a,b) \leftarrow D_Q} (-1)^{a\cdot b} \bra{\psi} (A^c_\alpha)^\dagger  Z(a) X(b) Z(a) (A^c_\alpha) \ket{\psi} \\
&= \sum_{u_1, u_2} (-1)^{u_1 + u_2} \E_{\substack{(a,b) \leftarrow D_Q \\ c = \Enc(\tlp)}} \sum_{\substack{\alpha \::\: \\ \Dec(\alpha)_i = u_1, \\ \Dec(\alpha)_j = u_2}} \bra{\psi} (A^c_\alpha)^\dagger (-1)^{a\cdot b}  Z(a) X(b) Z(a) (A^c_\alpha) \ket{\psi} \\
&= \sum_{u_1, u_2} (-1)^{u_1 + u_2} \E_{\substack{(a,b) \leftarrow D_Q \\ c = \Enc(\tlp)}} \sum_{\substack{\alpha \::\: \\ \Dec(\alpha)_i = u_1, \\ \Dec(\alpha)_j = u_2}} \bra{\psi} (A^c_\alpha)^\dagger (-1)^{a\cdot b}  (X(b) + Z(a) X(b) Z(a) - X(b)) (A^c_\alpha) \ket{\psi} \\
&\approx_{4\delta_\tlp(\eps)} \sum_{u_1, u_2} (-1)^{u_1 + u_2} \E_{\substack{(a,b) \leftarrow D_Q \\ c = \Enc(\tlp)}} \sum_{\substack{\alpha \::\: \\ \Dec(\alpha)_i = u_1, \\ \Dec(\alpha)_j = u_2}} \bra{\psi} (A^c_\alpha)^\dagger (-1)^{a\cdot b}  X(b) (A^c_\alpha) \ket{\psi} \\
&= \sum_{u_1, u_2} (-1)^{u_1 + u_2} \E_{\substack{b \leftarrow D_X \\ c = \Enc(\tlp)}} \sum_{\substack{\alpha \::\: \\ \Dec(\alpha)_i = u_1, \\ \Dec(\alpha)_j = u_2}} \bra{\psi} (A^c_\alpha)^\dagger X(b) (A^c_\alpha) \ket{\psi} \\
&= \hat \E[H_X],
\end{align} 
by \Cref{claim:energy-measurement-exact}.
\end{proof}

\begin{lemma}
  Define $\rho_\alpha := \E_{c = \Enc(\tlp)} \tr_{prover, A}[ V (A^c_\alpha) \ket{\psi}\bra{\psi} (A^c_\alpha)^\dagger V^\dagger ]$. Then
  \begin{equation} \sum_{v_1, v_2} (-1)^{v_1 + v_2}  \sum_{\substack{\alpha \::\: \\ \Dec(\alpha)_{n+i} = v_1, \\ \Dec(\alpha)_{n+j} = v_2}} \E_{a \leftarrow D_Z} \tr[ \sigma_Z(a) \rho_\alpha] = \hat{\E}[H_Z].   \end{equation}
\end{lemma}

\begin{proof}
Follows from \Cref{eq:isometry-z} and exact linearity (\Cref{lem:bob-exact-linearity}).
\end{proof}

\begin{lemma}
\label{lem:single-rho}
Assuming that the prover passes with probability $\omega^*_{\CHSH} - \eps$ in the CHSH subtest and with probability $1 - \eps$ in the commutation subtest, there exists a state $\rho$ such that
\begin{gather}
\E_{a \leftarrow D_Z} \tr[\sigma_Z(a) \cdot \rho] = \hat \E[H_Z], \label{eq:single-rho-Z} \\
\E_{b \leftarrow D_X} \tr[\sigma_X(b) \cdot \rho] \approx_{4\delta_\tlp(\eps)} \hat \E[H_X] \label{eq:single-rho-X}.
\end{gather}
\end{lemma}

\begin{proof}
Define $\rho_\alpha := \E_{c = \Enc(\tlp)} \tr_{prover, A}[ V (A^c_\alpha) \ket{\psi}\bra{\psi} (A^c_\alpha)^\dagger V^\dagger ]$. For notational convenience, define $\underbrace{z}_{n \text{ bits}} \| \underbrace{x}_{n \text{ bits}} := \Dec(\alpha)$. (Here $z$ represents the $Z$-gate corrections that the verifier is supposed to apply, and $x$ represents the $X$-gate corrections.) Define $\rho = \sum_\alpha \sigma_X(x) \sigma_Z(z) \cdot \rho_\alpha \cdot \sigma_Z(z) \sigma_X(x)$. Then, for any fixed $a = e_i + e_j$, we have
\begin{align}
\tr[ \sigma_Z(a) \cdot \rho ] &= \tr\bigg[ \sigma_Z(a) \Big( \sum_\alpha \sigma_X(x) \sigma_Z(z) \cdot \rho_\alpha \cdot \sigma_Z(z) \sigma_X(x) \Big) \bigg] \\
&= \sum_\alpha \tr\Big[ \sigma_Z(a) \:\: \sigma_X(x) \sigma_Z(z) \cdot \rho_\alpha \cdot \sigma_Z(z) \sigma_X(x) \Big] \\
&= \sum_\alpha \tr\Big[ \sigma_Z(a) \:\: \sigma_X(x) \cdot \rho_\alpha \cdot \sigma_X(x) \Big] \\
&= \sum_\alpha \tr\Big[ \sigma_X(x) \: \sigma_Z(a) \: \sigma_X(x) \cdot \rho_\alpha \Big] \\
&= \sum_{v_1, v_2} \:\: \sum_{\substack{\alpha \::\: \alpha = z \| x, \\ x_i = v_1, \\ x_j = v_2}} (-1)^{v_1 + v_2} \tr[\sigma_Z(a) \rho_\alpha].
\end{align}
Therefore,
\begin{align}
\hat \E[H_X] &= \E_{a \leftarrow D_Z} \sum_{v_1, v_2} \:\: \sum_{\substack{\alpha \::\: \alpha = z \| x, \\ x_i = v_1, \\ x_j = v_2}} (-1)^{v_1 + v_2} \tr[\sigma_Z(a) \rho_\alpha] \\
&= \E_{a \leftarrow D_Z} \tr[\sigma_Z(a) \cdot \rho].
\end{align}
An analogous calculation holds to show \Cref{eq:single-rho-X}.
\end{proof}

\begin{theorem}
\label{thm:main}
Let $\omega^*_{\mathsf{ver}}$ be the optimal success probability in the protocol. Set the protocol's choice of security parameter $\lambda$ to be equal to $n$. Then there exists a choice of $\kappa = \Theta((\beta-\alpha)^2)$ such that, for all sufficiently large $n$, the following holds. If the lowest eigenvalue of $H$ is at most $\alpha$, then $\omega^*_{\mathsf{ver}}$ is at least $\frac{1}{2}(1-\kappa)(1 + \omega^*_{\CHSH}) +\kappa(1-\frac{1}{4}\alpha)$. Conversely, if the lowest eigenvalue of $H$ is at least $\beta$, then $\omega^*_{\mathsf{ver}}$ is at most $\frac{1}{2}(1-\kappa)(1 + \omega^*_{\CHSH}) +\kappa(1-\frac{1}{4}\alpha) + \nu$, for $\nu = \frac{\kappa}{8}(\beta-\alpha)$. Thus, the protocol of \Cref{sec:protocol} achieves a completeness-soundness gap of $\frac{\kappa}{8}(\beta-\alpha)$.
\end{theorem}

\begin{proof}
Suppose the lowest eigenvalue of $H$ is at most $\alpha$. Then the prover can pass in the CHSH subtest with probability $\omega^*_{\CHSH}$, in the commutation subtest with probability 1, and in the teleport subtest with probability $1 - \frac{1}{4}\alpha$. If we perform the first two subtests with probability $\frac{1}{2}(1-\kappa)$ each and the last with probability $\kappa$, then $\omega^*_{\mathsf{ver}}$ is at least $\frac{1}{2}(1-\kappa)(1 + \omega^*_{\CHSH}) +\kappa(1-\frac{1}{4}\alpha)$.

Now suppose the lowest eigenvalue of $H$ is at least $\beta$. Then suppose the prover passes with probability $\frac{1}{2}(1-\kappa)(1 + \omega^*_{\CHSH}) +\kappa(1 - \frac{1}{4}\beta) + \nu$. Since the maximum passing probability for the commutation test is 1, and the maximum passing probability for the CHSH test is $\omega^*_{\CHSH} + \mathsf{negl}(\lambda)$, we note that the prover must pass in the teleport subtest with probability at least $1 - \frac{1}{4}\beta + \frac{\nu}{\kappa} - \mathsf{negl}(\lambda)$. Moreover, it cannot do too badly in the other subtests either: since it passes with probability at least $\frac{1}{2}(1-\kappa)(1 + \omega^*_{\CHSH}) +\kappa(1-\frac{1}{4}\beta)$ overall, we have (letting $p_{\CHSH}$ be the probability that the prover passes in the CHSH subtest, $p_{\comm}$ be the probability the prover passes in the commutation subtest, and $p_{\tlp}$ be the probability the prover passes in the teleport subtest):
\begin{gather}
\frac{1}{2}(1-\kappa)(p_{\CHSH}+p_{\comm}) + \kappa \cdot p_{\tlp} \geq \frac{1}{2}(1-\kappa)(1 + \omega^*_{\CHSH}) +\kappa(1-\frac{1}{4}\beta) \\
\implies \frac{1}{2}(1-\kappa)(p_{\CHSH}+p_{\comm}) + \kappa \geq \frac{1}{2}(1-\kappa)(1 + \omega^*_{\CHSH}) +\kappa(1-\frac{1}{4}\beta) \\
\implies \frac{1}{2}(1-\kappa)(p_{\CHSH}+p_{\comm}) \geq \frac{1}{2}(1-\kappa)(1 + \omega^*_{\CHSH}) - \kappa \cdot \frac{1}{4}\beta \\
\implies \frac{1}{2}(1-\kappa)\Big(1 + \omega^*_{\CHSH} - (p_{\CHSH}+p_{\comm}) \Big) \leq \kappa \cdot \frac{1}{4}\beta \\
\implies 1 + \omega^*_{\CHSH} - (p_{\CHSH}+p_{\comm}) \leq \frac{\kappa}{2(1-\kappa)}
\end{gather}
Hence the prover passes in the CHSH subtest with probability at least $\omega^*_{\CHSH} - \frac{\kappa}{2(1-\kappa)}$, and in the commutation subtest with probability at least $1 - \frac{\kappa}{2(1-\kappa)}$. Define $\eps := \frac{\kappa}{2(1-\kappa)}$. Applying \Cref{lem:single-rho}, and recalling that the prover passes with probability at least $1 - \frac{1}{4}\beta + \frac{\nu}{\kappa} - \mathsf{negl}(\lambda)$ in the teleport subtest, there exists a state $\rho$ such that $\tr[H\rho]$ is at most $\beta - \frac{4\nu}{\kappa} + \mathsf{negl}(\lambda) + 4\delta_{\tlp}(\eps)$. If $4\delta_{\tlp}(\frac{\kappa}{2(1-\kappa)}) < \frac{4\nu}{\kappa} - \mathsf{negl}(\lambda)$, then we derive a contradiction.

The theorem statement sets $\nu = \frac{\kappa}{8}(\beta - \alpha)$. Substituting into $4\delta_{\tlp}(\frac{\kappa}{2(1-\kappa)}) < \frac{4\nu}{\kappa} - \mathsf{negl}(\lambda)$, we find we need to set $\kappa$ and $\lambda$ such that
\begin{gather}
4\delta_{\tlp}(\frac{\kappa}{2(1-\kappa)}) < \frac{(\kappa/2)(\beta - \alpha)}{\kappa} - \mathsf{negl}(\lambda) \\
O((\frac{\kappa}{1-\kappa})^{1/2}) < \frac{\beta - \alpha}{2} - O(\:(\delta_\crp(\lambda))^{1/2}\:).
\end{gather}
Recall that $\delta_\crp$ is equal to $\sf{negl}(\lambda)$. Then, setting $\lambda = n$ and choosing $n$ large enough, we have
\begin{equation}
O((\frac{\kappa}{1-\kappa})^{1/2}) < \frac{\beta - \alpha}{2}.
\end{equation}
For an appropriate choice of $\kappa = \Theta((\beta-\alpha)^2)$, this can be shown to hold for sufficiently large $n$.

\end{proof}

\ifnames
\section{Acknowledgements}
Part of this research was performed while the authors were visiting
the Simons Institute and Berkeley Bowl Marketplace.  We thank both
institutions for their delicious hospitality. We are thankful to
Alexandru Gheorghiu for the suggestion to use our CHSH results to
construct a verification protocol, and for several helpful
conversations on this topic. We are also thankful to Alex Lombardi and
Fermi Ma for many useful comments and for sharing unpublished results,
and to Thomas Vidick for a useful conversation that led us to discover
a bug in a previous version of this paper.
\else
\fi

% \bibliographystyle{myhalpha}
% \bibliography{bellqma2}
\newcommand{\etalchar}[1]{$^{#1}$}

%\notesendofpaper

\end{document}